\documentclass{article}
\usepackage[left=3cm, right=3cm, bottom=4cm]{geometry}

\usepackage{amsmath, amsthm, amssymb,amsbsy}
\usepackage{graphicx}
\usepackage[colorlinks,citecolor=blue]{hyperref}
\usepackage{tikz}
\usetikzlibrary{calc,positioning,shapes, patterns}
\usepackage{bbm}
\usepackage{float}
\usepackage{enumitem}
\usepackage{comment}
\usepackage[normalem]{ulem}
\usepackage{cancel}

\newcommand{\stkout}[1]{\ifmmode\text{\sout{\ensuremath{#1}}}\else\sout{#1}\fi}

\usepackage[mathlines]{lineno}

\newcommand*\patchAmsMathEnvironmentForLineno[1]{%
   \expandafter\let\csname old#1\expandafter\endcsname\csname #1\endcsname
   \expandafter\let\csname oldend#1\expandafter\endcsname\csname end#1\endcsname
   \renewenvironment{#1}%
      {\linenomath\csname old#1\endcsname}%
      {\csname oldend#1\endcsname\endlinenomath}}%
\newcommand*\patchBothAmsMathEnvironmentsForLineno[1]{%
   \patchAmsMathEnvironmentForLineno{#1}%
   \patchAmsMathEnvironmentForLineno{#1*}}%
\AtBeginDocument{%
  \patchBothAmsMathEnvironmentsForLineno{equation}%
\patchBothAmsMathEnvironmentsForLineno{align}%
\patchBothAmsMathEnvironmentsForLineno{flalign}%
\patchBothAmsMathEnvironmentsForLineno{alignat}%
\patchBothAmsMathEnvironmentsForLineno{gather}%
\patchBothAmsMathEnvironmentsForLineno{multline}%
}

\newcommand{\editstart}{\renewcommand{\linenumberfont}{\normalfont\sffamily\footnotesize\bf\color{blue}}}
\newcommand{\editfinish}{\renewcommand{\linenumberfont}{\normalfont\sffamily\tiny\color{black}}}

\title{Distributed Compression of Graphical Data\footnote{This paper was
     presented in part at 2018 IEEE International Symposium on Information Theory.}}

\author{Payam Delgosha\thanks{Department of Computer Science, University of
    Illinois Urbana-Champaign, \texttt{delgosha@illinois.edu}}
   \, and   Venkat Anantharam\thanks{Department of Electrical Engineering and Computer
    Sciences, University of California, Berkeley, \texttt{ananth@berkeley.edu}}}

\newcommand{\ev}[1]{\mathbb{E} \left [ #1 \right ] }
\newcommand{\evwrt}[2]{\mathbb{E}_{#1} \left [ #2 \right ] }
\newcommand{\pr}[1]{\mathbb{P} \left ( #1 \right ) }
\newcommand{\prs}[1]{\mathbb{P} ( #1  ) } 

\newcommand{\snorm}[1]{\Vert #1 \Vert}
\newcommand{\one}[1]{\mathbbm{1} \left [ #1 \right ]}

\DeclareMathOperator{\var}{Var}

\newtheorem{lem}{Lemma}
\newtheorem{thm}{Theorem}
\newtheorem{definition}{Definition}
\newtheorem{prop}{Proposition}

\newtheorem{rem}{Remark}

\newcommand{\mG}{\mathcal{G}}

\newcommand{\mEn}{\mathcal{E}^{(n)}}

\newcommand{\mDn}{\mathcal{D}^{(n)}}

\newcommand{\mWn}{\mathcal{W}^{(n)}}

\newcommand{\vm}{\vec{m}}
\newcommand{\vu}{\vec{u}}
\newcommand{\vd}{\vec{d}}

\newcommand{\van}{\vec{a}^{(n)}}
\newcommand{\an}{a^{(n)}}
\newcommand{\vmn}{\vec{m}^{(n)}}
\newcommand{\vun}{\vec{u}^{(n)}}
\newcommand{\vdn}{\vec{d}^{(n)}}
\newcommand{\mn}{m^{(n)}}
\newcommand{\un}{u^{(n)}}
\newcommand{\dn}{d^{(n)}}

\newcommand{\tMn}{\tilde{M}^{(n)}}

\newcommand{\vp}{\vec{p}}
\newcommand{\vr}{\vec{r}}
\newcommand{\vq}{\vec{q}}
\newcommand{\vgamma}{\vec{\gamma}}

\newcommand{\vdg}{\overrightarrow{\text{dg}}}
\newcommand{\dg}{\text{dg}}

\newcommand{\Pn}{P^{(n)}}

\newcommand{\PER}{P^{(n)}_\text{ER}}
\newcommand{\tPER}{\tilde{P}^{(n)}_\text{ER}}
\newcommand{\PCM}{P^{(n)}_\text{CM}}
\newcommand{\tPCM}{\tilde{P}^{(n)}_\text{CM}}

\newcommand{\mGn}{\mG^{(n)}}
\newcommand{\mGnmnun}{\mathcal{G}^{(n)}_{\vmn, \vun}}
\newcommand{\Ln}{L^{(n)}}
\newcommand{\fn}{f^{(n)}}
\newcommand{\decn}{g^{(n)}}
\newcommand{\Gn}{G^{(n)}}
\newcommand{\Fn}{F^{(n)}}
\newcommand{\Sn}{S^{(n)}}
\newcommand{\Hn}{H^{(n)}}
\newcommand{\tHn}{\widetilde{H}^{(n)}}
\newcommand{\tFn}{\tilde{F}^{(n)}}
\newcommand{\Bn}{B^{(n)}}
\newcommand{\tGn}{\tilde{G}^{(n)}}
\newcommand{\mMn}{\mathcal{M}^{(n)}}
\newcommand{\mUn}{\mathcal{U}^{(n)}}

\newcommand{\hP}{\widehat{P}}

\newcommand{\mAn}{\mathcal{A}^{(n)}}
\newcommand{\mP}{\mathcal{P}}

\newcommand{\mT}{\mathcal{T}}

\newcommand{\mR}{\mathcal{R}}
\newcommand{\mX}{\mathcal{X}}
\newcommand{\mY}{\mathcal{Y}}
\newcommand{\mI}{\mathcal{I}}

\newcommand{\reals}{\mathbb{R}}

\newcommand{\nats}{\mathbb{N}}

\newcommand{\ugwt}{\text{UGWT}}

\newcommand{\ER}{Erd\H{o}s--R\'{e}nyi }

\newcommand{\dlp}{d_\text{LP}} 
\newcommand{\bch}{ \Sigma} 
\newcommand{\bchover}{\overline{\Sigma}}
\newcommand{\bchunder}{\underbar{$\Sigma$}}
\newcommand{\condmnun}{|_{(\vmn, \vun)}} 

\newcommand{\muer}{\mu^{\text{ER}}}
\newcommand{\mucm}{\mu^{\text{CM}}}
\newcommand{\Ter}{T^\text{ER}}
\newcommand{\Tcm}{T^\text{CM}}
\newcommand{\der}{d^\text{ER}}
\newcommand{\dcm}{d^\text{CM}}

\let\oldmarginpar\marginpar
\renewcommand{\marginpar}[2][rectangle,draw,rounded corners,text width = 3cm, scale=0.7]{%
        \oldmarginpar{%
          \tikz \node at (0,0) [#1]{#2};}%
        }

\newcommand{\edgemark}{\Xi} 
\newcommand{\vermark}{\Theta} 

\newcommand{\vtype}{\Pi} 
\newcommand{\vvtype}{\vec{\Pi}} 
\newcommand{\vdeg}{\vec{\deg}} 

\definecolor{newcolor}{RGB}{0,0,0}
\definecolor{maybecolor}{RGB}{0,0,0} 

\newcommand{\btheta}{\text{\color{blue}blue}}
\newcommand{\rtheta}{\text{\color{red}red}}
\newcommand{\ntheta}{\text{\color{black}black}}
\newcommand{\bedge}{\,\tikz[baseline] \draw[very thick] (0,-0.2 em) -- (0,0.7 em);\,}

\allowdisplaybreaks 

\begin{document}

\editstart

\colorlet{Cyan}{cyan}
\colorlet{Orange}{orange}
\tikzstyle{Node} = [circle,fill,inner sep=1.5pt]
\tikzstyle{Node2} = [rectangle,fill,inner sep=2.1pt]
\tikzstyle{Root} = [circle,fill=magenta,inner sep=1.9pt]
\tikzstyle{Root2} = [rectangle,fill=magenta,inner sep=2.5pt]

\maketitle

\begin{abstract}
In contrast to time series, graphical data is data indexed by the vertices
and edges of a graph. Modern applications such as the internet, social networks,
genomics and proteomics generate graphical data, often at large scale.
The large scale argues for the need to compress such data for storage
and subsequent processing.
Since this data might have several components available in different locations, it is also important to study distributed compression of graphical data.
In this paper, we derive a rate region for this problem which is a counterpart of the 
 Slepian--Wolf theorem. 
We characterize the rate region when the statistical description of the
distributed graphical data can be modeled as being one of two types --
as a member of a sequence of marked  sparse \ER ensembles or as a member of a sequence of marked configuration model ensembles. 
Our results are in terms of a generalization of the notion of entropy introduced by Bordenave and Caputo in the study of 
local weak limits of sparse graphs.
Furthermore, we give a generalization of this result
  for \ER and configuration model ensembles with more than two sources.
\end{abstract}

\section{Introduction}
\label{sec:introduction}










Nowadays, storing and processing data that in its native form is indexed by  combinatorial objects other than just linearly ordered time or multidimensional arrays is of great importance in many applications such as the internet, social networks and biology. For instance,  a social network could be presented  as a graph where each vertex models an individual and each edge stands for  a friendship. Also, vertices and edges can carry marks,
e.g. the mark of a vertex might describe some characteristics of the individual represented by the vertex, and the mark of an edge might describe some property of the nature of the interaction between the two individuals whose friendship is represented by the edge. The overall graphical data is then comprised of both the structure of the underlying graph and the data indexed by the graph, i.e. the vertex and edge marks.
Due to the sheer amount of such data in many applications, the question of how to compress it for efficient storage has drawn attention,  see
e.g. \cite{boldi2004webgraph}, \cite{choi2012compression},  \cite{abbe2016graph},
\cite{magner2016lossless}, \cite{basu2017universal}, \cite{delgosha2020universal}.

As the data is not always available in one location,  it is also important to consider distributed compression of graphical data. This latter question is the focus of this paper.
Traditionally, when dealing with time series, distributed lossless compression is modeled using two (or more) 
possibly dependent jointly 
stationary and ergodic processes
representing the components of the data at the individual locations. In this case, the rate region, which characterizes how efficiently the data can be compressed, is given  by the Slepian--Wolf theorem \cite{cover2012elements}. 
We adopt an analogous framework, namely that two jointly defined marked random graphs on the same vertex set 
are presented to two encoders, one to each encoder. Each encoder is then required to individually compress its data such that a third party,
having access to the two compressed representations,
can recover both marked graph realizations with a vanishing probability of error in the asymptotic limit of the size of the data.

We characterize the compression rate region for two scenarios, namely, a sequence of marked sparse \ER ensembles and a sequence of marked configuration model ensembles.
 We employ the framework of local weak convergence, also called the objective method, as a counterpart for marked graphs of the notion of stochastic processes for time series \cite{BenjaminiSchramm01rec, aldous2004objective, aldous2007processes}.
Our characterization of the rate region is best understood in terms of a generalization of a
measure of entropy  introduced by Bordenave and Caputo \cite{bordenave2015large}, which we call the marked
BC entropy \cite{delgosha2019notion}. It turns out that,
for  the sequences of ensembles  we study in this paper and even more generally, as proved in \cite{delgosha2019notion}, this notion of entropy captures the per--vertex growth rate of the portion of the Shannon entropy of the graphical data that is
over and above an entropy of connectivity which is controlled entirely by the average degree of the graph ensemble and not the detailed statistics of the graphical data. Indeed, to the highest order, 
 the marked BC entropy captures 
the part of the overall entropy that truly depends on the empirical characteristics of the graphical data and not just on the underlying connectivity structure of the graph. This motivates the marked BC entropy as a natural measure governing the asymptotic compression bounds, since it is sensitive to the details of the statistics of the ensembles and scales linearly with the number of vertices of the underlying graph.
Moreover, we generalize 
the two graphical source result
to the case where there are
  more than two graphical sources. 

The paper is organized as follows. In Section~\ref{sec:prel-notat} we introduce
the notation and formally state the problem.
Sections~\ref{sec:framework-local-weak} and \ref{sec:bc-entropy} give a brief
introduction to the concept of local weak convergence and to the marked BC entropy,  mostly specialized for the examples we study. Finally, in Section~\ref{sec:main-results}, we characterize the rate region for distributed lossless compression in the scenarios we present in Section~\ref{sec:prel-notat}, i.e. graphical data analogs of the Slepian--Wolf theorem in these scenarios.
Also, in Section~\ref{sec:gen-more-sources}, we generalize this
  result to the case where there are more than two graphical sources.



We close this section by introducing some of the main 
notational conventions used in this paper.
The set of natural numbers is denoted by $\nats$ and the set of real numbers is denoted by $\reals$. 
For $n \in \nats$, $[n]$ denotes the set $\{1, 2, \dots, n \}$.
For a probability distribution $P$ on a finite set, $H(P)$ denotes its Shannon entropy. Also, for a random variable $X$ taking values in a finite set, we denote by $H(X)$ its Shannon entropy. 
We write $:=$ for equality by definition.
For a positive integer $N$ and a sequence of positive integers $\{a_i\}_{1 \leq i \leq k}$ such that $\sum_{i=1}^k a_i \leq N$, we define 
\begin{equation*}
  \binom{N}{\{a_i\}_{1 \leq i \leq k}} := \frac{N!}{a_1! \dots a_k! (N - a_1 - \dots- a_k)!}.
\end{equation*}
For sequences of 
 real numbers 
$a_n$ and $b_n$, defined for all sufficiently large values of $n$, we write $a_n = O(b_n)$ if, for some constant $C \geq 0$, we have $|a_n| \leq C |b_n|$ for $n$ large enough. We write $a_n = o(b_n)$ if $a_n / b_n \rightarrow 0$ as $n\rightarrow \infty$. 
We denote by $\one{A}$  the indicator of the event $A$. 
For a probability distribution $P$ , $X \sim P$ denotes that the random variable $X$ has law $P$. Throughout the paper logarithms are to the natural base.

\section{Problem Statement}
\label{sec:prel-notat}



Let $\Xi$ and $\Theta$ be finite sets.
A marked graph with edge mark set $\Xi$ and vertex mark set $\Theta$ is a graph
where each edge carries a mark in $\Xi$ and each vertex carries a mark in
$\Theta$. 
 All graphs encountered in this paper are assumed to be simple, 
i.e. without multiple edges or self loops, unless otherwise stated. Also, we
assume that all edge and vertex mark sets are finite. For two vertices $v$ and
$w$ in a graph $G$, $v \sim_G w$ denotes that $v$ and $w$ are adjacent in $G$.
We denote the set of vertices in $G$ by $V(G)$. 
A finite sequence of nonnegative integers $(d(1), \ldots, d(n))$ is said to be {\em graphic} if there is a simple graph on $n$ vertices with vertex $i$ having degree $d(i)$ for $1 \le i \le n$.
A simple characterization of graphic sequences is provided by
the well known theorem of Erd\"{o}s and Gallai 
\cite{choudumErdosGallai, erdosgallai}.

Let $G$ be a marked
graph on a finite vertex set with edges and vertices carrying marks in the sets  $\Xi$ and $\Theta$, respectively. 
We denote the edge mark count vector of $G$ by $\vm_G = \{m_G(x)\}_{x \in \Xi}$, where $m_G(x)$ is the number of edges in $G$ carrying mark $x$. We denote the vertex mark count vector of $G$ by $\vu_G = \{u_G(\theta)\}_{\theta \in \Theta}$, where $u_G(\theta)$ denotes the number of vertices in $G$ carrying mark $\theta$. Additionally, for a graph $G$ on the vertex set $[n]$,  we denote the degree sequence of $G$ by $\vdg_G = \{\dg_G(1), \dots, \dg_G(n)\}$, where $\dg_G(i)$ denotes the degree of vertex $i$. 
For a degree sequence $\vd = (d(1), \dots, d(n))$ and a nonnegative integer $k$, we define
\begin{equation}
  \label{eq:def-ck}
  c_k(\vd) := | \{ 1 \leq i \leq n : d(i) = k \}|.
\end{equation}
Also, for two degree sequences $\vd=(d(1), \dots, d(n))$ and $\vd'=(d'(1), \dots, d'(n))$, and two nonnegative integers $k$ and $l$, we define 
\begin{equation}
  \label{eq:def-ckl}
  c_{k,l}(\vd, \vd') := | \{ 1 \leq i \leq n: d(i) = k, d'(i) = l \} |.
\end{equation}
Given a degree sequence $\vd = (d(1), \dots, d(n))$, we let $\mGn_{\vd}$ denote the set of simple unmarked graphs $G$ on the vertex set $[n]$ such that $\dg_G(i) = d(i)$ for $1 \leq i \leq n$.

 When discussing distributed compression of graphical data with 
two sources, 
we assume that $\Xi_1$ and $\Xi_2$ are two fixed and finite sets 
 of edge marks and 
$\Theta_1$ and $\Theta_2$ are two fixed and finite sets of vertex marks. For $i \in \{1, 2\}$ and $n \in \nats$, let $\mGn_i$ denote the set of marked 
graphs on the vertex set $[n]$ with edge and vertex mark sets $\Xi_i$ and $\Theta_i$ respectively. For two graphs $G_1 \in \mGn_1$ and $G_2 \in \mGn_2$, $G_1 \oplus G_2$ denotes the superposition of $G_1$ and $G_2$ which is a marked graph
defined as follows: a vertex $1 \leq v \leq n$ in $G_1 \oplus G_2$ carries the mark $(\theta_1, \theta_2)$ where $\theta_i$ is the mark of $v$ in $G_i$. 
Furthermore, we place an edge in $G_1 \oplus G_2$ between vertices $v$ and $w$ if there is an edge between them in at least one of $G_1$ or $G_2$, and mark this edge $(x_1, x_2)$, where, for $1 \leq i \leq 2$, $x_i$ is the mark of the edge $(v,w)$ in $G_i$ if it exists and $\circ_i$ otherwise. Here $\circ_1$ and $\circ_2$ are auxiliary marks not present in $\Xi_1 \cup \Xi_2$. 
Note that $G_1 \oplus G_2$ is a marked graph with edge  and vertex mark sets $\Xi_{1,2} := (\Xi_1 \cup \{\circ_1\}) \times (\Xi_2 \cup \{ \circ_2 \}) \setminus \{(\circ_1, \circ_2)\}$ and $\Theta_{1,2} := \Theta_1 \times \Theta_2$, respectively. 
We use the terminology \emph{jointly marked graph} to refer to a 
marked graph with edge and vertex mark sets $\Xi_{1,2}$ and
$\Theta_{1,2}$ respectively.
With this, let $\mGn_{1,2}$ denote the set of jointly 
marked graphs on the vertex set $[n]$.
Moreover, for $i \in \{1, 2 \}$, we say that a graph is in the $i$--th domain if its edge and vertex marks come from $\Xi_i$ and $\Theta_i$ respectively. For a jointly marked graph $G_{1,2}$ 
and $1 \leq i \leq 2$, the $i$--th marginal of $G_{1,2}$, denoted by $G_i$, is the marked graph in the $i$--th domain
obtained by projecting all vertex and edge marks onto $\Xi_i$ and $\Theta_i$, respectively, followed by removing edges with  mark $\circ_i$. 
Note that any jointly marked graph $G_{1,2}$ is uniquely determined
by its marginals $G_1$ and $G_2$, because $G_{1,2} = G_1 \oplus
G_2$.
Given an 
 edge mark count vector 
$\vm = \{m(x)\}_{x \in \Xi_{1,2}}$, for $x_1 \in \Xi_1 \cup \{\circ_1\}$ and $x_2 \in \Xi_2 \cup \{\circ_2\}$, with an abuse of notation we define 
\begin{equation}
  \label{eq:mx1}
  m(x_1) := \sum_{(x'_1, x'_2) \in \Xi_{1,2} \,:\, x'_1 = x_1} m((x'_1,x'_2)), 
  \qquad
  m(x_2) := \sum_{(x'_1, x'_2) \in \Xi_{1,2} \,:\, x'_2 = x_2} m((x'_1,x'_2)).
\end{equation}
Likewise, given a vertex mark count vector $\vu = \{u(\theta)\}_{\theta \in \Theta_{1,2}}$, we define, for $\theta_1 \in \Theta_1$ and $\theta_2 \in \Theta_2$, 
\begin{equation}
  \label{eq:utheta1}
  u(\theta_1) := \sum_{\theta'_2 \in \Theta_2} u((\theta_1, \theta'_2)), \qquad 
  u(\theta_2) := \sum_{\theta'_1 \in \Theta_1} u((\theta'_1, \theta_2)).
\end{equation}

Assume that we have a sequence of random marked graphs 
$\Gn_{1,2} \in \mGn_{1,2}$, defined for all $n$ sufficiently large,
drawn for each $n$ according to some ensemble distribution on $\mGn_{1,2}$. Additionally, assume that there are two encoders who want to compress 
realizations of such jointly marked graphs
in a distributed fashion.
Namely, the $i$--th encoder, $1 \leq i \leq 2$, has only access to the $i$--th
marginal $\Gn_i$. 
We assume that the distribution of $\Gn_{1,2}$ is known.

\begin{definition}
\label{def:SW-code}
A sequence of $\langle n, \Ln_1,\Ln_2 \rangle$ codes is a sequence of triples $(\fn_1,\fn_2,\decn)$, defined for all sufficiently large $n$,
such that 
\begin{equation*}
  \fn_i : \mGn_{i} \rightarrow [\Ln_i], \qquad i \in \{1, 2\},
\end{equation*}
and
\begin{equation*}
  \decn : [\Ln_1] \times [\Ln_2] \rightarrow \mGn_{1,2}.
\end{equation*}
The probability of error for this code corresponding to the ensemble of $\Gn_{1,2}$, which is denoted by $\Pn_e$, is defined as 
\begin{equation*}
  \Pn_e := \pr{\decn(\fn_1(\Gn_1), \fn_2(\Gn_2)) \neq \Gn_{1,2}}.
\end{equation*}
\end{definition}

Now we define our achievability criterion. 

\begin{definition}
  \label{def:SW-rate-achievable}
  A rate tuple $(\alpha_1, R_1, \alpha_2, R_2) \in \reals^4$ is said to be achievable for 
  distributed compression of the
  sequence of random graphs 
$\Gn_{1,2} \in \mGn_{1,2}$
  if  there is a sequence of $\langle n, \Ln_1, \Ln_2 \rangle$ codes such that 
  \begin{equation}
\label{eq:achievable-L-limsup}
    \limsup_{n \rightarrow \infty}  \frac{\log \Ln_i - (\alpha_i n \log n + R_i n)}{n} \leq 0, \qquad i \in \{1, 2\},
  \end{equation}
and also $\Pn_e \rightarrow 0$. The rate region $\mR \in \reals^4$ is defined as follows: for fixed $\alpha_1$ and $\alpha_2$, if there are sequences $R^{(m)}_1$ and $R^{(m)}_2$ with limit points $R_1$ and $R_2$ in $\reals$, respectively, such that for each $m$ the rate tuple $(\alpha_1, R^{(m)}_1, \alpha_2, R^{(m)}_2)$ is achievable, then we include $(\alpha_1, R_1, \alpha_2, R_2)$ in the set $\mR$. 
\end{definition}

In this paper, we characterize the above rate region for the following two sequences of ensembles:

\textbf{A sequence of \ER ensembles:} 
Assume that nonnegative real numbers $\vp=\{p_{x}\}_{x \in \Xi_{1,2}}$ together
with a probability distribution $\vq = \{q_\theta\}_{\theta \in \Theta_{1,2}}$
are given such that, for all
$x_1 \in \Xi_1$
and
  $x_2 \in \Xi_2$,
we have 
\begin{equation}
  \label{eq:ER-p-marginal-assumption}
  \sum_{\stackrel{(x'_1, x'_2) \in \Xi_{1,2}}{x'_1 = x_1}} p_{(x'_1, x'_2)} > 0 \qquad \text{and} \qquad \sum_{\stackrel{(x'_1, x'_2) \in \Xi_{1,2}}{x'_2  = x_2}} p_{(x'_1, x'_2)} > 0,
\end{equation}
and, for all $(\theta_1,\theta_2) \in \Theta_{1,2}$, we have
\begin{equation}
    \label{eq:ER-q-marginal-assumption}
\sum_{\theta'_2 \in \Theta_2} q_{(\theta_1, \theta'_2)} > 0 \qquad \text{and} \qquad \sum_{\theta'_1 \in \Theta_1} q_{(\theta'_1, \theta_2)} > 0.
\end{equation}
For $n \in \nats$ large enough, we define the probability distribution $\mG(n; \vp, \vq)$ on  $\mGn_{1,2}$ as follows: for each pair of vertices $1 \leq i < j \leq n$, the edge $(i,j)$ is present in the graph and has mark $x \in \Xi_{1,2}$ with probability $p_x / n$, and is not present with probability $1 - \sum_{x \in \Xi_{1,2}} p_x/n$. Furthermore, each vertex in the graph is given a mark $\theta \in \Theta_{1,2}$ with probability $q_\theta$. The choice of edge and vertex marks is done independently. 

 The conditions in~\eqref{eq:ER-q-marginal-assumption}
and the conditions for $x_i \in \edgemark_i$, $i =1,2$,  in~\eqref{eq:ER-p-marginal-assumption} 
are required only to ensure
that the sets of vertex marks and edge marks are chosen to be as
small as possible, and these conditions could be relaxed if desired.

\textbf{A sequence of configuration model ensembles:} Fix $\Delta \in \nats$.
Suppose that a  probability distribution $\vr = \{r_k\}_{k=0}^\Delta$ supported on the set $\{0, \dots, \Delta\}$ is given, such that $r_0 < 1$. Moreover, assume that probability distributions $\vgamma = \{\gamma_x\}_{x \in \Xi_{1,2}}$ and $\vq = \{q_{\theta}\}_{\theta \in \Theta_{1,2}}$ on the sets $\Xi_{1,2}$ and $\Theta_{1,2}$, respectively, are given.  We assume that, for all $x_1 \in \Xi_1 \cup \{\circ_1\}$ and $x_2 \in \Xi_2 \cup \{\circ_2\}$, we have 
\begin{equation}
  \label{eq:CM-gamma-marginal-assumption}
  \sum_{\stackrel{(x'_1, x'_2) \in \Xi_{1,2}}{x'_1 = x_1}} \gamma_{(x'_1, x'_2)} > 0 \qquad \text{and} \qquad \sum_{\stackrel{(x'_1, x'_2) \in \Xi_{1,2}}{x'_2  = x_2}} \gamma_{(x'_1, x'_2)} > 0,
\end{equation}
and, for all $(\theta_1,\theta_2) \in \Theta_{1,2}$, we have
\begin{equation}
    \label{eq:CM-q-marginal-assumption}
\sum_{\theta'_2 \in \Theta_2} q_{(\theta_1, \theta'_2)} > 0 \qquad \text{and} \qquad \sum_{\theta'_1 \in \Theta_1} q_{(\theta'_1, \theta_2)} > 0.
\end{equation}
Furthermore, for each $n$, the degree sequence $\vdn = \{\dn(1), \dots,
\dn(n)\}$ is given  such that, for all $1 \leq i \leq n$, we have $\dn(i) \leq \Delta$ and also $\sum_{i=1}^n \dn(i)$ is even. Let $m_n := (\sum_{i=1}^n \dn(i))/2$. Additionally, if, for $0 \leq k \leq \Delta$, $c_k(\vdn)$ denotes the number of $1 \leq i \leq n$ such that $\dn(i) = k$, we assume that, for some constant $K > 0$, we have 
  \begin{equation}
    \label{eq:dn-r-n23}
    \sum_{k=0}^\Delta |c_k(\vdn) - n r_k| \leq K n^{1/2}.
  \end{equation}
Now, for fixed $\vr$, $\vgamma$ and $\vq$ as above, and a sequence $\vdn$ satisfying \eqref{eq:dn-r-n23}, we define the law $\mG(n; \vdn, \vgamma, \vq, \vr)$ on $\mGn_{1,2}$, for $n \in \nats$ large enough, as follows. First, we pick an unmarked graph on the vertex set $[n]$ uniformly at random among the set of graphs $G$ with maximum degree $\Delta$ such that for each $0 \leq k \leq \Delta$, $c_k(\vdg_G) = c_k(\vdn)$.\footnote{The fact that each degree is bounded by $\Delta$, $r_0 < 1$ and the sum of degrees is even implies that $\vdn$ is a graphic sequence for $n \in \nats$ large enough. This is, for instance, a consequence of Theorem~4.5 in \cite{bordenave2015large}.} Then, we assign i.i.d.\ marks with law $\vgamma$ on the edges and i.i.d.\ marks with law $\vq$ on the vertices. 

The conditions in~\eqref{eq:CM-q-marginal-assumption}
and the conditions for $x_i \in \edgemark_i$, $i =1,2$,  in~\eqref{eq:CM-gamma-marginal-assumption}
are required only to ensure
that the sets of vertex marks and edge marks are chosen to be as
small as possible, and these conditions could be relaxed if desired.
However, the conditions in~\eqref{eq:CM-gamma-marginal-assumption}
for $x_i = \circ_i$, $i = 1,2$, are essential, as will be pointed
out at the appropriate point in the proofs, since they ensure that
neither of the two underlying unmarked graphs is a subgraph of the
other. 

As we will discuss in Section~\ref{sec:framework-local-weak} below,  the sequence of \ER ensembles defined above converges in the local weak sense to a marked Poisson Galton Watson tree. Moreover, the sequence of configuration model ensembles converges in the same  sense to a marked Galton Watson process with degree distribution $\vr$.
In Section~\ref{sec:main-results},
we will characterize the achievable rate regions for lossless distributed compression of  graphical data modeled as coming from one of the two sequences of ensembles above in terms of these limiting objects for the above two sequences of ensembles respectively. The formulation of this result will be in terms of a measure of entropy, namely the marked BC entropy, discussed in Section~\ref{sec:bc-entropy} below. 

\begin{rem}
It should be pointed out that a rate region in the sense of Definition~\ref{def:SW-rate-achievable} need not
be a {\color{newcolor}topologically} closed set, in contrast to what one is used to in the discussion of the Slepian-Wolf region in the traditional case. 
Further, while $\alpha_1$ and $\alpha_2$
can be restricted to being nonnegative, $R_1$ and $R_2$ should be thought of as real numbers.
Indeed, the rate regions for the two sequences of ensembles considered in this
paper, which are characterized in Theorem~\ref{thm:SW}, are not
{\color{newcolor}topologically} closed sets. The correct way to think of such a rate region is in terms of the subsets of $(R_1,R_2) \in \reals^2$,
parametrized by $(\alpha_1, \alpha_2) \in \reals^2$,
for which $(\alpha_1, R_1, \alpha_2, R_2)$ lies 
in the rate region, and each such subset 
is {\color{newcolor}topologically} closed as a subset of $\reals^2$.
Further, for any 
$(\alpha_1, R_1, \alpha_2, R_2)$ in the rate region,
if $\alpha_1^\prime > \alpha_1$ then 
$(\alpha_1^\prime, R_1^\prime, \alpha_2, R_2)$
lies in the rate region for all $R_1^\prime \in \reals$,
and a similar statement holds if one replaces the index $1$ by the index $2$.
\end{rem}


\section{The Framework of Local Weak Convergence}
\label{sec:framework-local-weak}

In this section, we discuss the framework of local weak convergence mainly in the context of the \ER and configuration model ensembles discussed in Section~\ref{sec:prel-notat}. For a general discussion, the reader is referred to \cite{BenjaminiSchramm01rec, aldous2004objective, aldous2007processes}.

Let $\Xi$ and $\Theta$ be fixed finite sets. 
A \emph{rooted marked
graph} is 
a marked
graph $G$ with edge and vertex mark sets $\Xi$ and $\Theta$
respectively, 
together with a distinguished vertex $o$. We denote 
such a rooted marked graph by $(G, o)$. 
For a rooted marked 
graph $(G, o)$ and a nonnegative integer $h \geq 0$, $(G,o)_h$ denotes the $h$ neighborhood of $o$, i.e. the subgraph consisting of vertices with distance no more than $h$ from $o$. Note that $(G, o)_h$ is connected, by definition. 
Two rooted marked graphs $(G_1, o_1)$ and $(G_2, o_2)$ are said to be isomorphic if there is a vertex bijection between the connected components of the roots in the two graphs that maps $o_1$ to $o_2$, preserves adjacencies, and also preserves edge and vertex marks. With this, we denote the isomorphism class corresponding to a rooted marked 
graph $(G, o)$ by $[G, o]$. We use $[G,o]_h$ as a shorthand for $[(G,o)_h]$.

Let $\mG_*(\Xi, \Theta)$ denote the set of isomorphism classes $[G,o]$ of rooted marked 
graphs on a countable vertex set with edge and vertex marks coming from the sets $\Xi$ and $\Theta$, respectively.
It can be shown that $\mG_*(\Xi, \Theta)$ can be metrized as a Polish space, i.e. a complete separable metric space \cite{aldous2007processes}.
In order to do this, we employ the metric on $\mG_*(\edgemark, \vermark)$
denoted by $d_*$
defined as follows: given $[G, o]$ and $[G',o']$ in $\mG_*(\edgemark, \vermark)$, let $\hat{h}$ be the supremum over all nonnegative integers $h\geq 0$ such that $(G, o)_h \equiv (G', o')_h$, where 
$(G, o)$ and $(G',o')$ are arbitrary members in the isomorphism classes $[G, o]$ and $[G', o']$ respectively\footnote{As all elements in an isomorphism class are isomorphic, the definition is invariant under the choice of the representatives.
}.
If there is no such $h$ (which can only happen if the mark of $o$ and $o'$ in
$G$ and $G'$, respectively, are not the same), we define $\hat{h} = 0$. 
With this, $d_*([G,o], [G',o'])$ is defined to be $1/(1+\hat{h})$. One can
check that $d_*$ is a metric; in particular, it satisfies the triangle
inequality.
Let $\mT_*(\edgemark, \vermark)$ denote the subset of $\mG_*(\edgemark, \vermark)$ comprised of the 
isomorphism classes $[G, o]$ arising from some $(G,o)$ where
the graph underlying $G$ is a tree.

We write $\mP(\mG_*(\edgemark, \vermark))$ for the set of 
probability distributions on $\mG_*(\Xi, \Theta)$ when it is viewed as a 
complete separable metric space with its Borel $\sigma$-algebra.
Given $\mu \in \mP(\mG_*(\edgemark, \vermark))$, let $\deg(\mu)$
denote the expected degree at the root in $\mu$. For $x \in
\edgemark$, let $\deg_x(\mu)$ denote the expected number of edges in $\mu$ connected to
the root which carry mark $x$, and define $\vdeg(\mu) :=
\{\deg_x(\mu)\}_{x \in \edgemark}$. 
For $\theta \in \vermark$, let $\vtype_\theta(\mu)$ denote the
probability that the mark at the root in $\mu$ is $\theta$, and let
$\vvtype(\mu) := \{\vtype_\theta(\mu)\}_{\theta \in \vermark}$.

For a finite marked 
graph $G$ and a vertex $v$ in $G$, let $G(v)$ denote the connected component of $v$. With this, if $v$ is a vertex chosen uniformly at random in $G$, we define $U(G)$ be the law of $[G(v), v]$, which  is a probability distribution on $\mG_*(\Xi, \Theta)$. If $\mGn$ denotes the set of marked 
graphs on the vertex set $[n]$ with edge and vertex mark sets $\Xi$ and $\Theta$ respectively, then a sequence of graphs $G^{(n)} \in \mGn$
is said to converge in the local weak sense if the sequence of 
probability distributions $U(G^{(n)})$ converges weakly in the usual sense \cite{billingsley2013convergence} as probability distributions on $\mG_*(\Xi, \Theta)$. We now describe what this notion means in more detail in the context of the two sequences of ensembles that are studied in this paper.

Let $\Gn_{1,2}$ be a random jointly marked 
graph with law $\mG(n; \vp, \vq)$ and let $v_n$ be a vertex chosen uniformly at random in the set $[n]$. A simple Poisson approximation implies that $D_x(v_n)$, the number of edges adjacent to $v_n$ with mark $x \in \Xi_{1,2}$, converges in distribution to a Poisson random variable with mean $p_x$, as $n$ goes to infinity. Moreover, $\{D_x(v_n)\}_{x \in \Xi_{1,2}}$ are asymptotically mutually independent. 
A similar argument can be repeated for any other vertex in the neighborhood of $v_n$. Also, it can be shown that the probability of having cycles
of any fixed length converges to zero. 
In fact, the isomorphism class of $(\Gn_{1,2}, v_n)_h$ converges in distribution to that of a rooted 
marked Poisson Galton Watson tree with depth $h$. 

More precisely, let $(\Ter_{1,2}, o)$ be a rooted jointly marked
tree 
 defined as follows. 
 First, the mark of the root is chosen 
 with 
 distribution $\vq$. Then, for $x \in \Xi_{1,2}$, we independently generate $D_x$ with law  $\text{Poisson}(p_x)$. We then add $D_x$ many edges with mark $x$ to the root $o$. For each offspring, i.e. vertex at the other end of an edge connected to the root, we repeat the same procedure independently, i.e. 
 choose its vertex mark according to the distribution $\vq$ 
 and then attach additional edges 
 with each edge mark 
 from the corresponding Poisson distribution with mean $p_x$, 
 independently for each edge mark in $\Xi_{1,2}$. Recursively repeating this, we get a connected jointly marked
 tree $\Ter_{1,2}$ rooted at $o$, which has possibly countably 
 infinitely many vertices. 
Let $\muer_{1,2}$ denote the law of the isomorphism class $[\Ter_{1,2},o]$. Note that $\muer_{1,2}$ is a probability distribution on $\mG_*(\Xi_{1,2}, \Theta_{1,2})$. 
$\muer_{1,2}$ depends on the underlying choice of the parameters
$(\vp, \vq)$, but we suppress this from the notation, for readability. 
The above discussion implies that, for all $h \ge 0$, $[\Gn_{1,2}, v_n]_h$ converges in distribution to $[\Ter_{1,2},o]_h$.  In fact, even a stronger statement can be proved, which is the following: If we consider the sequence of random graphs $\Gn_{1,2}$  independently on a joint probability space, $U(\Gn_{1,2})$ converges weakly to $\muer_{1,2}$ with probability one. With this, we say that, almost surely, $\muer_{1,2}$  is the \emph{local weak limit} of the sequence $\Gn_{1,2}$, where the term ``local'' is meant to indicate that we require the convergence in distribution of the isomorphism class of each
fixed depth neighborhood of a typical vertex (i.e. a vertex chosen uniformly at random).


With the construction above, let $\Ter_i$  be the $i$--th marginal of $\Ter_{1,2}$, for $1 \leq i \leq 2$. Moreover, let $\muer_i$ be the law of $[\Ter_i(o), o]$. Therefore, $\muer_i$ is a probability distribution on $\mG_*(\Xi_i, \Theta_i)$. Similarly to the argument above, one can see that, almost surely, $\muer_i$ is the local weak limit of the sequence $\Gn_i$. 

A similar picture also holds for the configuration model. Let $(\Tcm_{1,2}, o)$ be a rooted jointly marked
 random tree constructed as follows. First, we generate the degree of the root $o$ with law $\vr$. Then, for each offspring $w$ of $o$, we independently generate the offspring count  of $w$ with law $\vec{r'} = \{r'_k\}_{k=0}^{\Delta-1}$ defined as 
\begin{equation*}
  r'_k = \frac{(k+1)r_{k+1}}{\ev{X}}, \qquad 0 \leq k \leq \Delta - 1,
\end{equation*}
where $X$ has law $\vr$. We continue this process recursively, i.e. for each
vertex other than the root, we independently generate its offspring count with
law $\vec{r'}$. The distribution $\vr'$ is called the \emph{size-biased} distribution,
and takes into account the fact that each vertex other than the root has an extra
edge by virtue of its being defined via an edge to an earlier defined vertex, and hence its degree should be biased in order to get the
correct degree distribution $\vr$. Then, for each vertex and edge existing in
the graph $\Tcm_{1,2}$, we generate marks independently with laws $\vq$ and
$\vgamma$, respectively. Let $\mucm_{1,2}$ be the law of $[\Tcm_{1,2}, o]$.
Moreover, for $1 \leq i \leq 2$, let $\mucm_i$  be the law of $[\Tcm_i(o), o]$.
It can be shown that if $\Gn_{1,2}$ has law $\mG(n; \vdn, \vgamma, \vq, \vr)$,
with these random graphs being constructed independently on a joint probability space, then,
almost surely, $\mucm_{1,2}$ is the local weak limit of $\Gn_{1,2}$, and
$\mucm_i$ is the local weak limit of $\Gn_i$, for  $1 \leq i \leq 2$. 
$\mucm_{1,2}$ depends on the choice of the underlying parameters 
$(\vgamma, \vq, \vr)$, but we suppress this from the notation, for readability. 

A probability distribution on $\mG_*(\Xi, \Theta)$ is called \emph{sofic} if it is the
local weak limit of a sequence of finite simple marked graphs. Not all
probability distributions on $\mG_*(\Xi, \Theta)$ are sofic. In fact, the condition that all vertices have the same chance of being chosen as the root for a finite graph manifests itself as a certain stationarity condition at the limit, called \emph{unimodularity} \cite{aldous2007processes}.
To define unimodularity, let $\mG_{**}(\edgemark, \vermark)$ be the set of isomorphism classes
$[G,o,v]$ where $G$ is a marked connected graph with two distinguished vertices
$o$ and $v$ in $V(G)$ (ordered, but not necessarily distinct). Here, isomorphism
is defined by an adjacency preserving vertex bijection which preserves
vertex and edge marks, and also maps the two distinguished vertices of one
object to the respective ones of the other. A measure $\mu \in
\mP(\mG_*(\edgemark, \vermark))$ is
said to be unimodular if, for all measurable functions $f:
\mG_{**}(\edgemark, \vermark) \rightarrow \reals_+$, we have
\begin{equation}
  \label{eq:unim-integral}
  \int \sum_{v \in V(G)} f([G,o,v]) d\mu([G,o]) = \int \sum_{v \in V(G)} f([G,v,o]) d\mu([G,o]).
\end{equation}
Here the summation is taken over all vertices $v$ which are in the same
connected component of $G$ as $o$. It can be seen that it suffices to check the
above condition for a function $f$ such that $f([G,o,v]) = 0$ unless $v \sim_G
o$. This is called \emph{involution invariance} \cite{aldous2007processes}.
Let $\mP_u(\mG_*(\edgemark, \vermark))$ denote the set of unimodular probability measures on
$\mG_*(\edgemark, \vermark)$. Also, since $\mT_*(\edgemark, \vermark) \subset \mG_*(\edgemark, \vermark)$, we can define the set of
unimodular probability measures on $\mT_*(\edgemark, \vermark)$ and denote it by
$\mP_u(\mT_*(\edgemark, \vermark))$. 
A sofic probability measure is unimodular. Whether the other direction also
holds is unknown.

\section{The BC Entropy}
\label{sec:bc-entropy}


In this section, we discuss a notion of
entropy for probability distributions on the space 
$\mG_*(\edgemark, \vermark)$ of isomorphism classes of rooted marked
graphs with
edge and vertex mark sets
$\edgemark$ and $\vermark$ respectively.
This is a 
generalization to the marked framework
of the notion of entropy 
introduced by Bordenave and Caputo
in \cite{bordenave2015large}, who considered the unmarked case.
This generalization is due
to us, and the reader is referred to \cite{delgosha2019notion} for more details.
To distinguish it from the Shannon entropy, we call this 
notion of entropy
the \emph{marked BC entropy}.
In fact, the discussion in \cite{delgosha2019notion} is for a more general
setting in which each edge is allowed to carry two directional marks, one towards each of
its endpoints. The setup in this paper, where an edge is allowed to carry only
one mark,  can be considered as a special case where the two directional marks
have the same value. 
In the following, we give the definition of the marked BC entropy from
\cite{delgosha2019notion}, restricted to the setting in this paper where each
edge is allowed to carry only one mark. 

The following general lemma, whose
proof is straightforward using Stirling's approximation,
is often used in this paper. See
Appendix~\ref{sec:lemma-Stirling-proof} for a proof. 

\begin{lem}
  \label{lem:binom-assymp}
  Let $k \in \nats$. Let $a_n$ and $b^n_1, \dots, b^n_k$ be sequences of integers, defined for all sufficiently large $n$. 
  \begin{enumerate}
  \item Assume that $a_n = \sum_{i=1}^k b^n_k$ for all $n$. If $a_n / n \rightarrow  a > 0$ and, for each $1 \leq i \leq k$, $b^n_i / n \rightarrow b_i\geq 0$ where $a = \sum_{i=1}^k b_i$, we have 
  \begin{equation*}
     \lim_{n \rightarrow \infty} \frac{1}{n} \log \binom{a_n}{\{b^n_i\}_{1 \leq i \leq k}} = a H\left ( \left \{ \frac{b_i}{a} \right \}_{1 \leq i \leq k} \right ).
  \end{equation*}
\item Assume that $a_n \geq \sum_{i=1}^k b^n_k$ for all $n$. If $a_n / \binom{n}{2} \rightarrow 1$ and $b^n_i / n \rightarrow b_i \geq 0$, we have 
  \begin{equation*}
    \lim_{n \rightarrow \infty} \frac{\log \binom{a_n}{\{b^n_i \}_{1 \leq i \leq k}} - \left ( \sum_{i=1}^k b^n_i \right ) \log n }{n} = \sum_{i=1}^k s(2b_i),
  \end{equation*}
where $s(x)$ is defined to be $\frac{x}{2} - \frac{x}{2} \log x$ for $x > 0$ and $0$ if $ x= 0$. 
  \end{enumerate}
\end{lem}

Throughout the discussion in this section, up to the definition of
BC entropy in Definition~\ref{def:BC-entropy-new}, 
we assume that the edge and vertex mark sets,
$\edgemark$ and $\vermark$ respectively, are fixed and finite.
For edge and vertex mark count vectors $\vm = \{m(x)\}_{x \in \edgemark}$ and $\vu
= \{u(\theta)\}_{\theta \in \vermark}$, respectively, define $\snorm{\vm}_1 :=
\sum_{x \in \edgemark} m(x)$ and $\snorm{\vu}_1 := \sum_{\theta \in \vermark}
u(\theta)$.

Given $n \in \nats$, together with edge and vertex mark count vectors 
$\vm = \{m(x)\}_{x \in \Xi}$ and $\vu = \{u(\theta)\}_{\theta \in \Theta}$ 
respectively, let $\mGn_{\vm, \vu}$ denote the set of marked graphs $G$ on the vertex
set $\{1, \dots, n\}$ such that $\vm_G = \vm$ and $\vu_G = \vu$. Note that
$\mGn_{\vm, \vu}$ is empty unless $\snorm{\vu}_1 = n$ and $\snorm{\vm}_1 \leq
\binom{n}{2}$. 

We define an \emph{average degree vector} to be a vector of nonnegative reals
$\vd := \{d_{x}\}_{x \in \edgemark}$ such that $\sum_{x \in \edgemark}
d_{x} > 0$.

\begin{definition}
\label{def:deg-seq-adapt}  

  Given an average degree vector $\vd$ and a probability distribution $Q = \{q_\theta\}_{\theta \in \vermark}$, we say that a sequence 
$(\vmn,\vun)$, comprised of edge mark count vectors and vertex mark count vectors $\vmn$ and $\vun$ respectively, is adapted to $(\vd, Q)$, if the following conditions hold:
  \begin{enumerate}
  \item For each $n$, we have $\snorm{\vmn}_1 \leq \binom{n}{2}$ and $\snorm{\vun}_1 = n$;
  \item For $x \in \edgemark$, we have $\mn(x) / n \rightarrow d_{x}/2$;
  \item For $\theta \in \vermark$, we have $\un(\theta) / n \rightarrow q_\theta$;
  \item \label{item:cond-adapt-dx0}For $x \in \edgemark$, $d_{x} = 0$ implies $\mn(x) = 0$ for all $n$;
  \item For $\theta \in \vermark$, $q_\theta  = 0$ implies $\un(\theta) = 0$ for all $n$.
  \end{enumerate}
\end{definition}

If $\vmn$ and $\vun$ are sequences such that $(\vmn,\vun)$ is adapted to $(\vd,
Q)$ then 
one can show as a simple consequence of Lemma~\ref{lem:binom-assymp} that 
  \begin{equation}
    \label{EQ:LOG-MGNMNUN-STIRLING}
    \log | \mGn_{\vmn, \vun} | =  \snorm{\vmn}_1 \log n + n H(Q)  + n \sum_{x\in \edgemark} s(d_{x}) + o(n).
  \end{equation}
  See Appendix~\ref{sec:app-Stirling} for a proof. 
To simplify the notation, we may write $s(\vd)$ for $\sum_{x \in \edgemark}
s(d_{x})$. 


To give the definition of the 
\color{maybecolor}
marked
\color{black}
BC entropy, we first define the  upper and the lower 
\color{maybecolor}
marked
\color{black}
BC
entropy.

\begin{definition}
  \label{def:BC-entropy}
  
Assume $\mu \in \mP(\mG_*(\edgemark, \vermark))$ is given, with $0 < \deg(\mu) < \infty$. For $\epsilon>0$, and edge and vertex mark count vectors
$\vm$ and $\vu$ respectively, define
  \begin{equation*}
    \mGn_{\vm, \vu} (\mu, \epsilon) := \{ G \in \mGn_{\vm, \vu}: \dlp(U(G), \mu) < \epsilon \}.
  \end{equation*}
Here, $\dlp$ denotes the Levy--Prokhorov distance
\cite{billingsley2013convergence}. 
  Fix an average degree vector $\vd$ and a probability distribution $Q = \{q_\theta\}_{
  \theta \in \vermark}$, and also fix sequences of edge and vertex mark
  count vectors $\vmn$ and $\vun$ respectively such that $(\vmn,\vun)$ is adapted to $(\vd, Q)$. With these, define
  \begin{equation*}
    \bchover_{\vd, Q}(\mu, \epsilon)\condmnun := \limsup_{n \rightarrow \infty} \frac{\log |\mGn_{\vmn, \vun}(\mu, \epsilon)| - \snorm{\vmn}_1 \log n}{n},
  \end{equation*}
which we call the $\epsilon$--upper 
\color{maybecolor}
marked
\color{black}
BC entropy. Since this is increasing
in $\epsilon$, we can define the {\em upper 
\color{maybecolor}
marked
\color{black}
BC entropy} as 
  \begin{equation*}
    \bchover_{\vd, Q}(\mu)\condmnun := \lim_{\epsilon \downarrow 0} \bchover_{\vd, Q}(\mu, \epsilon)\condmnun.
  \end{equation*}
We may define the $\epsilon$--lower 
\color{maybecolor}
marked
\color{black}
BC entropy $\bchunder_{\vd, Q}(\mu,
\epsilon)\condmnun$ similarly as
\begin{equation*}
    \bchunder_{\vd, Q}(\mu, \epsilon)\condmnun := \liminf_{n \rightarrow \infty} \frac{\log |\mGn_{\vmn, \vun}(\mu, \epsilon)| - \snorm{\vmn}_1 \log n}{n}.
  \end{equation*}
Since this is increasing
in $\epsilon$, we can define the {\em lower 
\color{maybecolor}
marked
\color{black}
BC entropy} $\bchunder_{\vd, Q}(\mu)\condmnun$ as
\begin{equation*}
    \bchunder_{\vd, Q}(\mu)\condmnun := \lim_{\epsilon \downarrow 0} \bchunder_{\vd, Q}(\mu, \epsilon)\condmnun.
  \end{equation*}
 
    \end{definition}

Now, we state the following properties of the upper and lower marked BC entropy,
which will lead to the definition of the marked BC entropy. The reader is referred to
\cite{delgosha2019notion} for a proof and more details.

\begin{thm}[Theorem 1 in \cite{delgosha2019notion}]
  \label{thm:badcases}
  Let an average degree vector $\vd = \{d_{x}\}_{x \in \edgemark}$ and a
  probability distribution $Q = \{q_\theta\}_{\theta \in \vermark}$ be given.
  Suppose $\mu \in \mP(\mG_*(\edgemark, \vermark))$ with
 $0 < \deg(\mu) < \infty$ satisfies any one of the following conditions:
 \begin{enumerate}
    \item $\mu$ is not unimodular;
    \item $\mu$ is not supported on $\mT_*(\edgemark, \vermark)$;
    \item $\deg_{x}(\mu) \neq d_{x}$ for some $x\in \edgemark$, or
      $\vtype_\theta(\mu) \neq q_\theta$ for some $\theta \in \vermark$. 
    \end{enumerate}
    Then, for any choice of the
    sequences $\vmn$ and $\vun$ such that $(\vmn,\vun)$ is adapted to $(\vd, Q)$, we have $\bchover_{\vd, Q}(\mu)\condmnun = -\infty$. 
  \end{thm}

A consequence of Theorem~\ref{thm:badcases}
is that the only case of interest in the discussion of marked
BC entropy is when $\mu \in \mP_u(\mT_*(\edgemark, \vermark))$,
$\vd = \vdeg(\mu)$, $Q = \vvtype(\mu)$,
and the
sequences $\vmn$ and $\vun$ are such that $(\vmn,\vun)$ is adapted to
$(\vdeg(\mu), \vvtype(\mu))$.
In particular, the only upper and lower marked BC entropies of interest are 
$\bchover_{\vdeg(\mu), \vvtype(\mu)}(\mu)\condmnun$ and $\bchunder_{\vdeg(\mu), \vvtype(\mu)}(\mu)\condmnun$ respectively.

The following 
theorem
establishes that the upper and lower
marked BC entropies do not depend on the 
choice of the defining pair of sequences 
$(\vmn,\vun)$. Further, 
this theorem establishes that
the upper marked BC entropy 
is always equal to the lower marked BC entropy. The reader is referred to
\cite{delgosha2019notion} for a proof and more details. 

\begin{thm}[Theorem 2 in \cite{delgosha2019notion}]
  \label{thm:bch-properties}
  Assume that an average degree vector $\vd = \{d_{x}\}_{x \in \edgemark}$ together with a
  probability distribution $Q = \{q_\theta\}_{\theta \in \vermark}$ are given. For
  any  $\mu \in \mP(\mG_*(\edgemark, \vermark))$ such that
 $0 < \deg(\mu) < \infty$, we have 
  \begin{enumerate}
  \item \label{thm:BC-invariant} The values of $\bchover_{\vd, Q}(\mu)\condmnun$ and
    $\bchunder_{\vd, Q}(\mu)\condmnun$ are invariant under the specific choice of the
    sequences $\vmn$ and $\vun$ such that $(\vmn,\vun)$ is adapted to $(\vd, Q)$. With this,
    we may simplify the notation and unambiguously write $\bchover_{\vd, Q}(\mu)$ and
    $\bchunder_{\vd, Q}(\mu)$. 
  \item \label{thm:BC-well} 
  $\bchover_{\vd, Q}(\mu) = \bchunder_{\vd, Q}(\mu)$. 
  We may therefore unambiguously write $\bch_{\vd, Q}(\mu)$ 
for this common value,
and call it the {\em marked BC entropy} of 
$\mu \in \mP(\mG_*(\edgemark, \vermark))$ for the 
average degree vector $\vd$ and a probability distribution $Q = \{q_\theta\}_{
  \theta \in \vermark}$.
  Moreover, $\bch_{\vd, Q}(\mu) \in [-\infty, s(\vd) + H(Q)]$.
  \end{enumerate}
\end{thm}

From Theorem~\ref{thm:badcases} we conclude that unless 
$\vd = \vdeg(\mu)$, $Q = \vvtype(\mu)$, and $\mu$
  is a unimodular measure on $\mT_*(\edgemark, \vermark)$, we have 
  $\bch_{\vd, Q}(\mu) = -\infty$. 
  In view of this, for $\mu \in \mP(\mG_*(\edgemark, \vermark))$
  with $0< \deg(\mu) < \infty$, we write $\bch(\mu)$
  for  $\bch_{\vdeg(\mu), \vvtype(\mu)}(\mu)$. Likewise, we may write
  $\bchunder(\mu)$ and $\bchover(\mu)$ for $\bchunder_{\vdeg(\mu),
    \vvtype(\mu)}(\mu)$ and $\bchover_{\vdeg(\mu),
    \vvtype(\mu)}(\mu)$, respectively. 
  These are both equal to $\bch(\mu)$ by part~\ref{thm:BC-well} of the theorem. 
    Note that, unless $\mu \in \mP_u(\mT_*(\edgemark, \vermark))$, 
    we have $\bchover(\mu) = \bchunder(\mu) = \bch(\mu) = -\infty$.
    
We are now in a position to define the marked BC entropy.

\begin{definition}
  \label{def:BC-entropy-new}
  For $\mu \in \mP(\mG_*(\edgemark, \vermark))$
  with $0 < \deg(\mu) < \infty$, the marked BC entropy of $\mu$ is defined to be $\bch(\mu)$.
\end{definition}

{\color{newcolor} In Appendix~\ref{app:bc-ent-calc-examples}, we have provided the details of
calculating the marked BC entropy for several examples.} We next connect the asymptotic behavior of the entropy of the ensembles defined
in Section~\ref{sec:prel-notat} to the marked BC entropy of their local weak limits. 
We first consider a sequence of \ER ensembles.
Let $n \in \nats$ be large enough, and
assume that  $\Gn_{1,2}$ has law $\mG(n; \vp, \vq)$. Let $\der_{1,2}:= \deg(\muer_{1,2}) = \sum_{x \in \Xi_{1,2}} p_x$. 
For $x_i \in \Xi_i$ and $\theta_i \in \Theta_i$, $1 \leq i \leq 2$, let
\begin{equation}
  \label{eq:er-px-utheta-conv}
\begin{gathered}
  p_{x_1}:= \sum_{x'_2 \in \Xi_2 \cup \{\circ_2\}} p_{(x_1, x'_2)}, \quad   p_{x_2}:= \sum_{x'_1 \in \Xi_1 \cup \{\circ_1\}} p_{(x'_1, x_2)}, \\
  q_{\theta_1} := \sum_{\theta'_2 \in \Theta_2} q_{(\theta_1, \theta'_2)}, \quad   q_{\theta_2} := \sum_{\theta'_1 \in \Theta_1} q_{(\theta'_1, \theta_2)}.
\end{gathered}
\end{equation}
For $1 \leq i \leq 2$, let $\der_i := \deg(\muer_i) = \sum_{x_i \in \Xi_i} p_{x_i}$. If $Q = (Q_1, Q_2)$ has law $\vq$, it can be 
verified by using Lemma~\ref{lem:binom-assymp} in a manner similar to the proof of \eqref{EQ:LOG-MGNMNUN-STIRLING} in Appendix~\ref{sec:app-Stirling} that 
we have 
{\small
\begin{subequations}
  \begin{align}
    H(\Gn_{1,2}) &= \frac{\der_{1,2}}{2} n \log n + n \left (H(Q) + \sum_{x \in \Xi_{1,2}} s(p_x)\right ) + o(n) \label{eq:ent-assympt-er-12},\\
    H(\Gn_{1}) &= \frac{\der_{1}}{2} n \log n + n \left (H(Q_1) + \sum_{x_1 \in \Xi_{1}} s(p_{x_1})\right ) + o(n) \label{eq:ent-assympt-er-1},\\
    H(\Gn_{2}) &= \frac{\der_{2}}{2} n \log n + n \left (H(Q_2) + \sum_{x_2 \in \Xi_{2}} s(p_{x_2})\right ) + o(n) \label{eq:ent-assympt-er-2}.
  \end{align}
\end{subequations}
}%
Using Theorem 3 in \cite{delgosha2019notion}, it can be seen that the coefficients of $n$ in equations
\eqref{eq:ent-assympt-er-12}--\eqref{eq:ent-assympt-er-2}
are $\bch(\muer_{1,2})$, $\bch(\muer_1)$ and $\bch(\muer_2)$, respectively
{\color{newcolor}(see Appendix~\ref{app:bc-ent-calc-examples} for details).}

Before discussing configuration model ensembles,
we state two lemmas, which are used at several
points.
The proof of the following Lemma~\ref{lem:thinning-entropy} straightforward, and is therefore omitted.

\begin{lem}
\label{lem:thinning-entropy}
Let $\Delta \in \nats$. 
  Let $Y$ be a random variable taking values in $\{0, 1, \dots, \Delta\}$, and let $0 \leq \epsilon \leq 1$. Let $ \{ V_i \}_{i \geq 1}$ be a sequence of i.i.d. Bernoulli random variables with $\pr{V_i=1} = \epsilon$,
  and let $Y_1 := \sum_{i=1}^Y V_i$, where $Y_1 = 0$ when $Y = 0$.
  Then, we have 
  \begin{equation*}
    H(Y_1, Y-Y_1) = H(Y_1, Y) = H(Y) + \ev{Y} H(V_1) - \ev{\log \binom{Y}{Y_1}}.
  \end{equation*}
\end{lem}
\hfill $\Box$

The proof of the following Lemma~\ref{lem:degree-cm-count} is given in
Appendix~\ref{sec:lemma-Stirling-proof-second}.

\begin{lem}
\label{lem:degree-cm-count}
  Let $\Delta \in \nats$. Let $Y$ 
  be a random variable taking values in 
  $\{0,1, \ldots \Delta\}$, 
  such that $d := \ev{Y}>0$. For all $n \in \nats$ large enough, let
  $\van = (\an(1), \ldots, \an(n))$ 
  be a degree sequence of length $n$ with entries
 bounded by $\Delta$  such that  $b_n:= \sum_{i=1}^n \an(i)$ is even and, for $0 \leq k \leq \Delta$,
 we have $c_k(\van) / n \rightarrow \prs{Y = k}$. 
 Then, we have 
  \begin{equation*}
    \lim_{n \rightarrow \infty} \frac{\log |\mGn_{\van}| - \frac{b_n}{2} \log n }{n} = -s(d) - \ev{\log Y!},
  \end{equation*}
where we recall that $\mGn_{\van}$ denotes the set of simple unmarked graphs $G$ on the vertex set $[n]$ such that $\dg_G(i) = \an(i)$ for $1 \leq i \leq n$.
\end{lem}


\begin{rem}
The assumption $\ev{Y} > 0$ in the above lemma is crucial and can not be
relaxed. To see this, consider the following example: let $Y = 0$ with
probability one, and let $\van$ be such that $\an(1) = \an(2) = 3$ and $\an(i) =
0$ for $i > 2$. Then, although $b_n$ is even, $\van$ is not 
graphic
and
$\mGn_{\van}$ is empty. Therefore, the above limit of interest is $- \infty$ and
the equality does not hold. 
\end{rem}


Consider now a sequence of configuration model ensembles. Namely, for all
$n \in \nats$ large enough, let $\Gn_{1,2}$ be distributed according to $\mG(n; \vdn, \vgamma, \vq, \vr)$. Let $X$ be a random variable with law $\vr$ and
$\Gamma^k = (\Gamma^k_1, \Gamma^k_2)$, $1 \le k \le \Delta$,
an i.i.d.\ sequence distributed according to $\vgamma$. With this, let 
\begin{equation}
  \label{eq:X1-X2-def}
  X_1 := \sum_{k=1}^X \one{\Gamma^k_1 \neq \circ_1}, \qquad X_2:= \sum_{k=1}^X \one{\Gamma^k_2 \neq \circ_2},
\end{equation}
where $X_1 = X_2 = 0$ if $X = 0$. 
Then,
if $\dcm_{1,2} := \deg(\mucm_{1,2})$ and, for $1 \leq i \leq 2$,
$\dcm_i := \deg(\mucm_{i})$, it can be seen that 
\begin{subequations}
  \begin{align}
    H(\Gn_{1,2}) &= \frac{\dcm_{1,2}}{2} n \log n + n \Big (-s(\dcm_{1,2}) + H(X) - \ev{\log X!} \nonumber\\
    &\qquad + H(Q) + \frac{\dcm_{1,2}}{2} H(\Gamma) \Big ) + o(n),  \label{eq:ent-assympt-cm-12}\\
        H(\Gn_{1}) &= \frac{\dcm_{1}}{2} n \log n + n \Big (-s(\dcm_{1}) + H(X_1) - \ev{\log X_1!} \nonumber\\
    &\qquad + H(Q_1) + \frac{\dcm_{1}}{2} H(\Gamma_1|\Gamma_1 \neq \circ_1) \Big ) + o(n), \label{eq:ent-assympt-cm-1}\\
   H(\Gn_{2}) &= \frac{\dcm_{2}}{2} n \log n + n \Big (-s(\dcm_{2}) + H(X_2) - \ev{\log X_2!} \nonumber\\
    &\qquad + H(Q_2) + \frac{\dcm_{2}}{2} H(\Gamma_2|\Gamma_2 \neq \circ_2) \Big ) + o(n), \label{eq:ent-assympt-cm-2}
  \end{align}
\end{subequations}
where $\Gamma$ is distributed according to $\vgamma$. 
Also, using Theorem~3 in \cite{delgosha2019notion}, it can be seen that the coefficients of $n$
in equations
\eqref{eq:ent-assympt-cm-12}--\eqref{eq:ent-assympt-cm-2}
are $\bch(\mucm_{1,2})$, $\bch(\mucm_1)$ and $\bch(\mucm_2)$, respectively {\color{newcolor}(see Appendix~\ref{app:bc-ent-calc-examples} for details).}
The proof of equations
\eqref{eq:ent-assympt-cm-12}--\eqref{eq:ent-assympt-cm-2}, which is given in Appendix
\ref{sec:asympt-behav-enropy-cm}, and depends on 
both Lemma~\ref{lem:thinning-entropy} and
Lemma~\ref{lem:degree-cm-count}.

If $\mu_{1,2}$ is any one of the two distributions $\muer_{1,2}$ or $\mucm_{1,2}$,
and $\mu_1$ and $\mu_2$ are its marginals, we define the \emph{conditional marked BC entropies} as $\bch(\mu_2| \mu_1) := \bch(\mu_{1,2}) - \bch(\mu_1)$ and $\bch(\mu_1 | \mu_2):= \bch(\mu_{1,2}) - \bch(\mu_2)$.

\section{Main Results}
\label{sec:main-results}

Now, we are ready to state our main result, which is to characterize the rate region in Definition~\ref{def:SW-rate-achievable} for a sequence of \ER ensembles and a sequence of configuration model ensembles. In the following, for pairs of reals $(\alpha, R)$ and $(\alpha', R')$, we write $(\alpha, R) \succ (\alpha', R')$ if either $\alpha > \alpha'$, or $\alpha = \alpha'$ and $R > R'$. We also write $(\alpha, R) \succeq (\alpha', R')$ if either $(\alpha, R) \succ (\alpha', R')$ or $(\alpha, R) = (\alpha', R')$. 

\begin{thm}
\label{thm:SW}
Assume $\mu_{1,2}$ 
is a member of
either of 
the two families of distributions
$\muer_{1,2}$ 
(parametrized by $(\vp, \vq)$) 
or $\mucm_{1,2}$ 
(parametrized by $(\vgamma, \vq, \vr)$) 
defined in 
Section~\ref{sec:framework-local-weak}. 
Then, if $\mR$ is the rate region for 
the sequence of ensembles 
corresponding to $\mu_{1,2}$, as defined in Section~\ref{sec:prel-notat}, 
a  rate tuple $(\alpha_1, R_1, \alpha_2, R_2) \in \mR$ if and only if 
\begin{subequations}
  \begin{align}
      (\alpha_1, R_1) &\succeq ((d_{1,2}-d_2)/2, \bch(\mu_1 | \mu_2)), \label{eq:thm-assumption-1}\\
      (\alpha_2, R_2) &\succeq ((d_{1,2}-d_1)/2, \bch(\mu_2 | \mu_1)), \label{eq:thm-assumption-2}\\
      (\alpha_1 + \alpha_2, R_1 + R_2) & \succeq (d_{1,2}/2, \bch(\mu_{1,2})), \label{eq:thm-assumption-12}
  \end{align}
\end{subequations}
where $d_{1,2} := \deg(\mu_{1,2})$, $d_1 := \deg(\mu_1)$ and $d_2 := \deg(\mu_2)$.
\end{thm}

We prove the achievability for the  \ER case and the configuration model case in  Sections \ref{sec:proof-achievability-ER} and \ref{sec:proof-achievability-conf}, respectively.  Subsequently, we prove the converses for the two cases in Sections \ref{sec:proof-converse-ER} and \ref{sec:proof-converse-CM}, respectively. 

{\color{newcolor}
\begin{rem}
  Although our achievability analysis shares some well-known concepts with the
  classical Slepian--Wolf, such as the random binning method, there are
  several factors that makes the analysis for graphical data much more
  challenging compared to the classical results for time series. For one thing,
  as we saw in Section~\ref{sec:bc-entropy}, our entropy analysis is up to the
  first two leading terms, one which scales like $n \log n$ and the other which
  scales like $n$. This is reflected in the statement of the above
  Theorem~\ref{thm:SW} as the appearance of  two rate parameters $\alpha$ and $R$ for each
  source. On the other hand, the classical operational meaning of the
  conditional Shannon entropy does not easily extend to similar operational
  meanings for the conditional marked BC entropy. More precisely, in the
  classical setting of two i.i.d.\ sources $\mX$ and $\mY$ with a joint
  distribution $p_{X, Y}$, roughly speaking, any typical sequence $(x_1, \dots, x_n)$
  has approximately the same number of conditional typical sequences $(y_1, \dots, y_n)$, and the
  number of such conditional typical sequences is asymptotically related to the
  conditional Shannon entropy $H(Y|X)$. However, it turns out that a similar
  property does not necessarily hold in our setting for sparse graphical data. See
  Appendix~\ref{app:constancy-counterexample} for details. This in part makes
  our analysis more complicated compared to the classical setting as we need to
  carefully control the number of conditional typical graphs. This requires separate
  treatment for the \ER and the configuration model ensembles, as is
  discussed in Sections~\ref{sec:proof-achievability-ER} and
  \ref{sec:proof-achievability-conf} below, respectively. 
\end{rem}

\begin{rem}
  Recall from Section~\ref{sec:bc-entropy} that the coefficient of $n$ in the
  ensemble entropies of the \ER and the configuration model ensembles and their
  marginals are equal to the marked BC entropy of their corresponding local weak
  limits. This is a key reason why  the rate region in
  Theorem~\ref{thm:SW} above is characterized in terms of the marked BC entropy. 
The reason why the ensemble entropies and the marked BC entropies match is that
the \ER and the configuration model ensembles are almost uniform over the
typical graphs with respect to their corresponding local weak limits. For the \ER
case, it is well known that the \ER ensemble is close in distribution to a
distribution on the set of graphs with a typical number of edges. For the
configuration model case, we learn from the techniques used in \cite{bordenave2015large} and
\cite{delgosha2019notion} to prove the properties of the BC entropy that the
configuration model ensemble covers the set of typical graphs roughly uniformly in an asymptotic
sense. In other words, the fact that the ensemble entropies and the marked BC
entropies match is not a coincidence. 
\end{rem}

}

As is the case for the 
classical 
Slepian--Wolf theorem, 
one can generalize the above result to more
than two sources. The definition of the rate region as well as its
characterization can be naturally extended to this case. In Section~\ref{sec:gen-more-sources} below, we
generalize 
the \ER and configuration model ensembles to more than two sources,
define the corresponding Slepian-Wolf rate region, 
and 
characterize the rate region for each of these cases 
in Theorem~\ref{thm:graph-SW-k-source}. The proof
structure is similar
to that for the scenario with two sources, 
and is highlighted in Appendix~\ref{sec:app-multisource-proof}.

\subsection{Proof of Achievability for the \ER case}
\label{sec:proof-achievability-ER}

Here we show that a rate tuple $(\alpha_1, R_1, \alpha_2, R_2)$ is achievable for the \ER ensemble if it satisfies the following
\begin{subequations}
  \begin{align}
    (\alpha_1, R_1) &\succ ((\der_{1,2}-\der_2)/2, \bch(\muer_1 | \muer_2)), \label{eq:ach-assumption-1}\\
    (\alpha_2, R_2) &\succ ((\der_{1,2}-\der_1)/2, \bch(\muer_2 | \muer_1)), \label{eq:ach-assumption-2}\\
    (\alpha_1 + \alpha_2, R_1 + R_2) & \succ (\der_{1,2}/2, \bch(\muer_{1,2})). \label{eq:ach-assumption-12}
  \end{align}
\end{subequations}
Note that if a rate tuple $(\alpha'_1, R'_1, \alpha'_2, R'_2)$ satisfies the weak inequalities \eqref{eq:thm-assumption-1}--\eqref{eq:thm-assumption-12} then, for any $\epsilon > 0$, $(\alpha'_1, R'_1 + \epsilon, \alpha'_2, R'_2 + \epsilon)$ satisfies 
the strict inequalities
\eqref{eq:ach-assumption-1}--\eqref{eq:ach-assumption-12}. 
As we show below, this implies that $(\alpha'_1, R'_1 + \epsilon, \alpha'_2, R'_2 + \epsilon)$ is achievable. Hence,  after sending $\epsilon \rightarrow 0$, we get $(\alpha'_1, R'_1, \alpha'_2, R'_2) \in \mR$. 

We show that any $(\alpha_1, R_1, \alpha_2, R_2)$ satisfying \eqref{eq:ach-assumption-1}--\eqref{eq:ach-assumption-12} is achievable by employing a random binning method. More precisely, for $i \in \{1,2\}$, we set $\Ln_i = \lfloor \exp(\alpha_i n \log n + R_i n )\rfloor$ and  for each $G_i \in \mGn_i$, we assign $\fn_i(G_i)$ uniformly at random  in the set $[\Ln_i]$ and independent of everything else.

To describe our decoding scheme, we first need to set up some notation. Let $\mMn$ denote the set of edge count vectors $\vm = \{m(x)\}_{x \in \Xi_{1,2}}$ such that 
\begin{equation*}
  \sum_{x \in \Xi_{1,2} } |m(x) - np_{x}/2| \leq n^{2/3}.
\end{equation*}
Moreover, let $\mUn$ denote the set of vertex mark count vectors $\vu = \{u(\theta)\}_{\theta \in \Theta_{1,2}}$ such that 
\begin{equation*}
  \sum_{\theta \in \Theta_{1,2}} | u(\theta) - n q_{\theta}| \leq n^{2/3}.
\end{equation*}
Furthermore, we define $\mGn_{\vp, \vq}$ to be the set of graphs $\Hn_{1,2} \in \mGn_{1,2}$ such that $\vm_{\Hn_{1,2}} \in \mMn$ and $\vu_{\Hn_{1,2}} \in \mUn$. Upon receiving $(i, j) \in [\Ln_1] \times [\Ln_2]$, we form the set of graphs $\Hn_{1,2} \in \mGn_{\vp, \vq}$ such that $\fn_1(\Hn_1) = i$ and $\fn_2(\Hn_2) = j$, where $\Hn_1$ and $\Hn_2$ are the marginals of $\Hn_{1,2}$. If this set has only one element, we output this element as the decoded graph; otherwise, we report an error. 

In what follows, assume that  $\Gn_{1,2}$ is a random graph with law $\mG(n; \vp, \vq)$. We consider the following four error events 
corresponding to the above scheme: 
\begin{align*}
  \mEn_1 &:= \{\Gn_{1,2} \notin \mGn_{\vp, \vq} \},\\
\mEn_2 &:= \{\exists \Hn_{1,2} \in \mGn_{\vp, \vq}: \Hn_1 \neq \Gn_1, \Hn_2 \neq \Gn_2,  \fn_i(\Hn_i) = \fn_i(\Gn_i), i \in \{1,2\} \},\\
\mEn_3 &:= \{ \exists \Hn_2 \neq \Gn_2: \Gn_1 \oplus \Hn_2 \in \mGn_{\vp, \vq}, \fn_2(\Hn_2) = \fn_2(\Gn_2) \},\\
\mEn_4 &:= \{ \exists \Hn_1 \neq \Gn_1: \Hn_1 \oplus \Gn_2 \in \mGn_{\vp, \vq}, \fn_1(\Hn_1) = \fn_1(\Gn_1) \}.
\end{align*}
Note that outside the above four events the decoder successfully  decodes the input graph $\Gn_{1,2}$. 

Using Chebyshev's inequality, for some $\kappa > 0$ we have $\prs{\mEn_1} \leq
\kappa n^{-1/3}$, which converges to zero as $n$ goes to infinity.
Moreover, using the union bound, we have 
\begin{equation}
\label{eq:pr-E2-upperbound-1}
  \pr{\mEn_2} \leq \frac{|\mGn_{\vp, \vq}|}{\Ln_1 \Ln_2}.
\end{equation}
Note that, for each graph 
$\Hn_{1,2} \in \mGn_{\vp, \vq}$, the mark count vectors
$\vm_{\Hn_{1,2}}$ and $\vu_{\Hn_{1,2}}$ are 
in the sets $\mMn$ and $\mUn$ respectively. 
Additionally, we have $|\mMn| \leq (2n^{2/3}+1)^{|\Xi_{1,2}|}$ and
$|\mUn| \leq (2 n^{2/3}+1)^{|\Theta_{1,2}|}$.  Therefore,
\begin{equation}
\label{eq:er-Gnpq-A1}
  |\mGn_{\vp, \vq}| \leq (2 n^{2/3}+1)^{(|\Xi_{1,2}|+|\Theta_{1,2}|)} \max_{\stackrel{\vm \in \mMn}{\vu \in \mUn}} A_1(\vm, \vu),
\end{equation}
where
{\small
\begin{align*}
  A_1(\vm, \vu) := \binom{n}{\{u(\theta)\}_{\theta \in \Theta_{1,2}}} \binom{\binom{n}{2}}{\{m(x)\}_{x \in \Xi_{1,2}}}.
\end{align*}
}%
Now, let $\vmn$ and $\vun$ be sequences in $\mMn$ and $\mUn$, respectively. Then,  for all $x \in \Xi_{1,2}$ and $\theta \in \Theta_{1,2}$, we have $\mn(x) / n \rightarrow p_{x} / 2$ and $\un(\theta) / n \rightarrow q_{\theta}$. Thereby, using Lemma~\ref{lem:binom-assymp}, we have 
\begin{align*}
  \lim_{n \rightarrow \infty}  \frac{\log A_1(\vmn, \vun) - (\sum_{x \in \Xi_{1,2}} \mn(x)) \log n}{n} \\= H(\vq) + \sum_{x \in \Xi_{1,2}} s(p_{x}) = \bch(\muer_{1,2}).
\end{align*}
Substituting this into \eqref{eq:er-Gnpq-A1} and using the fact that $\sum |\mn(x) - np_{x} / 2| \leq n^{2/3}$, we have 
\begin{equation}
\label{eq:log-GnpQ-bch-12}
    \limsup_{n \rightarrow \infty}  \frac{\log |\mGn_{\vp, \vq}| - n\frac{\der_{1,2}}{2} \log n}{n} \leq \bch(\muer_{1,2}).
\end{equation}
Substituting this into \eqref{eq:pr-E2-upperbound-1}, we have 
{\small
\begin{align*}
  &\limsup \frac{1}{n} \log \pr{\mEn_2} \\
  & \,\leq \limsup \frac{\log |\mGn_{\vp, \vq}| - n\frac{\der_{1,2}}{2} \log n - n\bch(\muer_{1,2})}{n} \\
  &\quad + \limsup \frac{n(\frac{\der_{1,2}}{2} - \alpha_1 - \alpha_2) \log n + n(\bch(\muer_{1,2}) - R_1 - R_2)}{n} \\
  &\quad + \limsup \frac{n(\alpha_1 + \alpha_2) \log n + n(R_1+R_2) - \log \Ln_1\Ln_2}{n}.
\end{align*}
}%
The first term is nonpositive due to \eqref{eq:log-GnpQ-bch-12}, the second term is strictly negative due to the assumption \eqref{eq:ach-assumption-12}, and the third term is nonpositive due to our choice of $\Ln_1$ and $\Ln_2$. Consequently, the RHS is strictly negative, which implies that $\prs{\mEn_2} \rightarrow 0$.

Now, we show that $\prs{\mEn_3 \setminus \mEn_1}$ vanishes. In order to do so, for $\Hn_1 \in \mGn_1$, define $\Sn_2(\Hn_1) := \{ \Hn_2 \in \mGn_2: \Hn_1 \oplus \Hn_2 \in \mGn_{\vp, \vq} \}$. Using the union bound, we have 
\begin{equation}
\label{eq:PE3E1c-upperbound-1}
\begin{split}
  \pr{\mEn_3 \setminus \mEn_1}  &\leq \sum_{\Hn_{1,2} \in \mGn_{\vp, \vq}} \prs{\Gn_{1,2} = \Hn_{1,2}} \frac{|\Sn_2(\Hn_1)|}{\Ln_2} \\
  &\leq \frac{1}{\Ln_2} \max_{\Hn_{1, 2} \in \mGn_{\vp, \vq}} |\Sn_2(\Hn_1)|.
\end{split}
\end{equation}
It can be shown that (See Appendix~\ref{sec:bounding-s_2g_1-er}) 
{\small
\begin{equation}
\label{eq:limsup-S2G1-conditional-BC}
\begin{aligned}
  \limsup_{n\rightarrow \infty} \frac{\displaystyle \max_{\Hn_{1,2} \in \mGn_{\vp, \vq}} \log |\Sn_2(\Hn_1)| - n\frac{\der_{1,2}-\der_1}{2} \log n }{n} \leq \bch(\muer_2 | \muer_1),
\end{aligned}
\end{equation}
}%
where $\Hn_1$ is the first marginal of $\Hn_{1,2}$. 
Substituting this in \eqref{eq:PE3E1c-upperbound-1}, 
we get 
{\small
\begin{equation}
\label{eq:ER-e3-e1-vanishes}
\begin{aligned}
  \limsup \frac{1}{n} \log \pr{\mEn_3 \setminus \mEn_1} & \leq \limsup \frac{n\frac{\der_{1,2}-\der_1}{2} \log n + n \bch(\muer_2 | \muer_1) - \log \Ln_2}{n}\\
  &\leq  \limsup \frac{n(\frac{\der_{1,2}-\der_1}{2} - \alpha_2)\log n + n(\bch(\muer_2|\muer_1) - R_2)}{n} \\
  &\qquad + \limsup \frac{n\alpha_2 \log n + n R_2 - \log \Ln_2}{n}.
\end{aligned}
\end{equation}
}%
Note that the first term is strictly negative due to the assumption 
\eqref{eq:ach-assumption-2}, while the  second term is nonpositive due to our way of choosing $\Ln_2$. This means that $\prs{\mEn_3 \setminus \mEn_1}$ goes to zero as $n$ goes to infinity. Similarly, $\prs{\mEn_4 \setminus \mEn_1}$ 
converges to zero as $n \to \infty$. 
This means that there exists a sequence of deterministic codebooks with vanishing probability of error, which completes the proof of achievability.

\subsection{Proof of Achievability for the Configuration model}
\label{sec:proof-achievability-conf}

Our achievability proof for this case is very similar in nature to that for the \ER case, with the modifications discussed below. 

Let $\mDn$ be the set of degree sequences $\vd$ with entries bounded by $\Delta$ such that $c_k(\vd) = c_k(\vdn)$ for all $0 \leq k \leq \Delta$. Moreover, redefine $\mMn$ to be the set of mark count vectors $\vm$ such that $\sum_{x \in \Xi_{1,2}} m(x) = m_n$ and $\sum_{x \in \Xi_{1,2}} |m(x) - m_n \gamma_{x}| \leq n^{2/3}$, where we recall that $m_n = (\sum_{i=1}^n \dn(i)) / 2$. We use the same definition for $\mUn$ as in the previous section, i.e. the set of vertex mark count vectors $\vu$ such that $\sum_{\theta \in \Theta_{1,2}} |u(\theta) - n q_\theta| \leq n^{2/3}$. 

In what follows, let $X$ be a random variable with law $\vr$,  $X_1$ and $X_2$ defined as in \eqref{eq:X1-X2-def}, and $\Gamma = (\Gamma_1, \Gamma_2)$ a random variable with law $\vgamma$. 

We define  $\mWn$ to be the set of graphs $\Hn_{1,2} \in \mGn_{1,2}$ such that: $(i)$ $\vdg_{\Hn_{1,2}} \in \mDn$, $(ii)$ $\vm_{\Hn_{1,2}} \in \mMn$, $(iii)$ $\vu_{\Hn_{1,2}} \in \mUn$,  $(iv)$ for all $0 \leq l \leq k \leq \Delta$, recalling the notation in \eqref{eq:def-ckl}, we have
\begin{equation}
  \label{eq:mWn-deg-count-1}
  |c_{k,l}(\vdg_{\Hn_{1,2}}, \vdg_{\Hn_1}) - n \pr{X = k,X_1 = l}| \leq n^{2/3},
\end{equation}
 and $(v)$ for all $ 0 \leq l \leq k \leq \Delta$ we have 
 \begin{equation}
   \label{eq:mWn-deg-count-2}
 |c_{k,l}(\vdg_{\Hn_{1,2}}, \vdg_{\Hn_2}) - n \pr{X = k, X_2 = l}| \leq n^{2/3}.
 \end{equation}

We employ a similar random binning framework as in Section~\ref{sec:proof-achievability-ER}. For decoding, upon receiving a pair $(i,j)$, we form the set of graphs $\Hn_{1,2} \in \mWn$ such that $\fn_1(\Hn_1) = i$ and $\fn_2(\Hn_2) = j$. If this set has only one element, we output it as the source graph; otherwise, we output an indication of error. In order to prove the achievability, we consider the four error events $\mEn_i$, $1 \leq i \leq 4$, defined exactly like those in the previous section, with $\mGn_{\vp, \vq}$ being replaced with $\mWn$.

It can be shown that if $\Gn_{1,2} \sim \mG(n; \vdn, \vgamma, \vq, \vr)$, the probability of $\Gn_{1,2} \in \mWn$ goes to one as $n$ goes to infinity (see Lemma~\ref{lem:mWn-highprobability} in Appendix~\ref{sec:asympt-behav-enropy-cm}). Therefore,  $\prs{\mEn_1}$ 
goes to zero as $n \to \infty$. 

To show that $\prs{\mEn_2}$ vanishes, similar to the analysis in Section~\ref{sec:proof-achievability-ER}, we find an asymptotic upper bound for 
$\log |\mWn|$.
By only considering the conditions $(i)$, $(ii)$ and $(iii)$ in the definition of $\mWn$, we have 
\begin{equation}
\label{eq:log-mWn-upperbound-1}
\begin{aligned}
  \log |\mWn| & \leq \log \binom{n}{\{c_k(\vdn)\}_{k=0}^\Delta} + \log |\mGn_{\vdn}| \\
  & \quad + \log \left ((2 n^{2/3}+1)^{| \Xi_{1,2} |} \max_{\vm \in \mMn} \binom{m_n}{\{m(x)\}_{x \in \Xi_{1,2}}} \right)\\
  & \quad + \log \left ( (2 n^{2/3}+1)^{ |\Theta_{1,2}|} \max_{\vu \in \mUn} \binom{n}{\{u(\theta)\}_{\theta \in \Theta_{1,2}}} \right ).
\end{aligned}
\end{equation}
By assumption, we have $r_0 < 1$, hence $\dcm_{1,2} > 0$. 
The condition \eqref{eq:dn-r-n23} together with Lemma~\ref{lem:degree-cm-count}
in Appendix~\ref{sec:asympt-behav-enropy-cm} then implies that 
\begin{equation}
\label{eq:log-Gndn-bound}
\begin{aligned}
  \lim_{n \rightarrow \infty} \frac{\log |\mGn_{\vdn}| - n \frac{\dcm_{1,2}}{2} \log n}{n} &= \lim_{n \rightarrow \infty} \frac{\log |\mGn_{\vdn}| - m_n \log n}{n} + \lim_{n \rightarrow \infty} \frac{(m_n - n \dcm_{1,2}/2) \log n}{n}\\
  & =-s(\dcm_{1,2}) - \ev{\log X!},
\end{aligned}
\end{equation}
where on the second line we have used 
the bound $|m_n - n \dcm_{1,2}/2| \leq K
\Delta n^{1/2}$ which is implied by \eqref{eq:dn-r-n23}.
Using this together with Lemma~\ref{lem:binom-assymp} for the other terms in~\eqref{eq:log-mWn-upperbound-1}, we have 
\begin{align*}
  \limsup_{n \rightarrow \infty} \frac{\log |\mWn| - n \frac{\dcm_{1,2}}{2} \log n}{n} \leq -s(\dcm_{1,2}) +H(X) \\ + \frac{\dcm_{1,2}}{2} H(\Gamma) + H(Q)- \ev{\log X!} = \bch(\mucm_{1,2}),
\end{align*}
where $\Gamma$ and $Q$ are random variables with law $\vgamma$ and $\vq$, respectively.

Now, in order to show that $\prs{\mEn_3 \setminus \mEn_1}$ vanishes, we prove a counterpart for \eqref{eq:limsup-S2G1-conditional-BC}. 
For $\Hn_1 \in \mGn_1$, 
we define $\Sn_2(\Hn_1)$ to be the set of graphs $\Hn_2 \in \mGn_2$ such that $\Hn_1 \oplus \Hn_2 \in \mWn$.
Then, it can be shown (see Appendix~\ref{sec:bound-s_2g_1-conf}) that
\begin{equation}
\label{eq:conf-S2G1-bch-2|1}
\begin{aligned}
  \limsup_{n \rightarrow \infty} \frac{\displaystyle \max_{\Hn_{1,2} \in \mWn} \log |\Sn_2(\Hn_1)| - n \frac{\dcm_{1,2} - \dcm_1}{2} \log n}{n} 
\leq \bch(\mucm_2 | \mucm_1).
\end{aligned}
\end{equation}
Then, similar to \eqref{eq:ER-e3-e1-vanishes}, this shows that $\prs{\mEn_3
  \setminus \mEn_1}$ 
vanishes as $n \to \infty$. 
  Similarly, $\prs{\mEn_4 \setminus \mEn_1}$ 
  vanishes as $n \to \infty$. 
  This completes the proof of achievability.

\subsection{Proof of the Converse for the \ER case}
\label{sec:proof-converse-ER}

In this section, we show that every rate tuple $(\alpha_1, R_1, \alpha_2, R_2) \in \mR$ for the \ER scenario must satisfy the conditions \eqref{eq:thm-assumption-1}--\eqref{eq:thm-assumption-12}. By definition, for a rate tuple $(\alpha_1, R_1, \alpha_2, R_2) \in \mR$, there exist sequences $R^{(m)}_1$ and $R^{(m)}_2$ such that for each $m$, $(\alpha_1, R^{(m)}_1, \alpha_2, R^{(m)}_2)$ is achievable and, besides, we have  $R^{(m)}_1 \rightarrow R_1$ and $R^{(m)}_2 \rightarrow R_2$. If we show that $(\alpha_1, R^{(m)}_1, \alpha_2, R^{(m)}_2)$ satisfies \eqref{eq:thm-assumption-1}--\eqref{eq:thm-assumption-12} for each $m$, it is easy to see that $(\alpha_1, R_1, \alpha_2, R_2)$ must also satisfy the same inequalities. Therefore, it suffices to show that any achievable rate tuple satisfies \eqref{eq:thm-assumption-1}--\eqref{eq:thm-assumption-12}.

For this, take an achievable rate tuple $(\alpha_1, R_1, \alpha_2, R_2)$ together with a corresponding sequence of $\langle n, \Ln_1, \Ln_2 \rangle$ codes $(\fn_1, \fn_2, \decn)$. By definition, we have 
\begin{equation}
  \label{eq:ER-converese-code-Ln}
  \limsup_{n \rightarrow \infty}  \frac{\log \Ln_i - (\alpha_i n \log n + R_i n)}{n} \leq 0 \qquad i \in \{1, 2\},
\end{equation}
and also the error probability $\Pn_e$ goes to zero as $n$ goes to infinity. Now, we define the set $\mAn \subseteq \mGn_{1,2}$ as 
\begin{equation}
  \label{eq:mAn-def}
  \mAn := \mGn_{\vp, \vq} \cap \{ \Hn_{1,2} \in \mGn_{1,2}: \decn(\fn_1(\Hn_1), \fn_2(\Hn_2)) = \Hn_{1,2} \},
\end{equation}
where $\mGn_{\vp, \vq}$ was defined in Section~\ref{sec:proof-achievability-ER}. In fact, $\mAn$ is the set of ``typical'' graphs with respect to the \ER model that are successfully decoded
by the code $(\fn_1, \fn_2, \decn)$.
In the following, let $\Gn_{1,2} \sim \mGn(n; \vp, \vq)$ be distributed according to the \ER model. Moreover, let $\PER$ be the law of $\Gn_{1,2}$, i.e. for $\Hn_{1,2} \in \mGn_{1,2}$, $\PER(\Hn_{1,2}) := \prs{\Gn_{1,2} = \Hn_{1,2}}$.  With this, we define a random variable $\tGn_{1,2}$ whose distribution is the conditional distribution of
$\Gn_{1,2}$, conditioned on lying in $\mAn$, 
i.e.
\begin{equation}
  \label{eq:ER-converse-Ptilde-def}
  \pr{\tGn_{1,2} = \Hn_{1,2}} =
  \begin{cases}
    \PER(\Hn_{1,2}) / \pi_n & \Hn_{1,2} \in \mAn, \\
    0 & \text{otherwise.}
  \end{cases}
\end{equation}
where $\pi_n := \pr{\Gn_{1,2} \in \mAn}$ is the normalizing factor.
Note that, since $\Pn_e \to 0$ as $n \to \infty$ and 
$P(\Gn_{1,2} \in \mAn) \to 1$ as $n \to \infty$, we have
$\pi_n > 0$ for all sufficiently large $n$, and in fact
$\pi_n \to 1$ as $n \to \infty$. 
Additionally, let $\tPER$ be the law of $\tGn_{1,2}$. If, for $i \in \{1,2\}$, $\tMn_i$ denotes $\fn_i(\tGn_i)$, we have 
\begin{equation}
  \label{eq:ER-conv-sum-1}
  \begin{aligned}
    \log \Ln_1 + \log \Ln_2 &\geq H(\tMn_1) + H(\tMn_2) \geq H(\tMn_1, \tMn_2) \\
    &= H(\tGn_{1,2}),
  \end{aligned}
\end{equation}
where the last equality follows from the fact that, by definition, $\tGn_{1,2}$ takes values among the graphs that are successfully decoded, and hence is uniquely identified given $\tMn_1$ and $\tMn_2$. 

Now, we find a lower bound for $H(\tGn_{1,2})$. For doing so, note that for $\Hn_{1,2} \in \mGn_{1,2}$ and $n$ large enough, we have
\begin{equation}
  \label{eq:logPER-general}
\begin{aligned}
  - \log \PER(\Hn_{1,2}) &= - \sum_{x \in \Xi_{1,2}} m_{\Hn_{1,2}}(x) \log \frac{p_x}{n} - \left [ \binom{n}{2} - \sum_{x \in \Xi_{1,2}} m_{\Hn_{1,2}} (x) \right ] \log \left ( 1 - \frac{\sum_{x \in \Xi_{1,2}} p_x}{n} \right ) \\
  &\qquad - \sum_{\theta \in \Theta_{1,2}} u_{\Hn_{1,2}}(\theta) \log q_\theta.
\end{aligned}  
\end{equation}
On the other hand, due to the definition of $\mGn_{\vp, \vq}$, if $\Hn_{1,2} \in \mGn_{\vp, \vq}$ then,  for all $x \in \Xi_{1,2}$ and $\theta \in \Theta_{1,2}$, we have
\begin{equation*}
\begin{gathered}
  n\frac{p_x}{2} - n^{2/3} \leq m_{\Hn_{1,2}}(x) \leq n\frac{p_x}{2} + n^{2/3}, \mbox{ and} \\
  n q_\theta - n^{2/3} \leq u_{\Hn_{1,2}}(\theta) \leq nq_\theta + n^{2/3}.
\end{gathered}  
\end{equation*}
Substituting these in \eqref{eq:logPER-general} and using the inequality $\log (1-x) \leq -x$ which holds for $x \in (0,1)$, for $n$ large enough, we have 
\begin{align*}
  - \log \PER(\Hn_{1,2}) & \geq \sum_{x \in \Xi_{1,2}} \left ( n \frac{p_x}{2} - n^{2/3} \right ) ( \log n - \log p_x) + \left [ \binom{n}{2} - \sum_{x \in \Xi_{1,2}} \left (n \frac{p_x}{2} + n^{2/3}\right) \right ] \frac{\sum_{x \in \Xi_{1,2}} p_x}{n} \\
  &\qquad - \sum_{\theta \in \Theta_{1,2}} (nq_\theta - n^{2/3}) \log q_\theta.
\end{align*}
Using $\sum_{x \in \Xi_{1,2}} p_x = \der_{1,2}$ and simplifying the above, we realize that there exists a  constant $c > 0$ that does not depend on  $n$ or $\Hn_{1,2}$, such that, 
for all $\Hn_{1,2} \in \mGn_{\vp, \vq}$ and thus, in particular,
for all $\Hn_{1,2} \in \mAn$, 
we have 
\begin{equation}
  \label{eq:er-converse-log-PER-lowerbound}
\begin{aligned}
  - \log \PER(\Hn_{1,2})  & \geq n \frac{\der_{1,2}}{2} \log n - n \sum_{x \in \Xi_{1,2}} \frac{p_x}{2} \log p_x + n \sum_{x \in \Xi_{1,2}} \frac{p_x}{2}  - n \sum_{\theta \in \Theta_{1,2}} q_\theta \log q_\theta - c n^{2/3} \log n \\
&= n\frac{\der_{1,2}}{2} \log n + n \bch(\muer_{1,2}) - c n^{2/3} \log n.
\end{aligned}  
\end{equation}
Now, if $\tGn_{1,2}$ is the random variable defined in \eqref{eq:ER-converse-Ptilde-def}, we have 
\begin{align*}
H(\tGn_{1,2}) & = - \sum_{\Hn_{1,2} \in \mAn} \tPER(\Hn_{1,2}) 
\log \tPER(\Hn_{1,2}) \\
& = \log \pi_n - \frac{1}{\pi_n} \sum_{\Hn_{1,2} \in \mAn} 
 \PER(\Hn_{1,2}) \log \PER(\Hn_{1,2}).
\end{align*}
Note that since the probability of error of the above code vanishes, i.e. $\Pn_e \rightarrow 0$, and $\pr{\Gn_{1,2} \in \mGn_{\vp, \vq}} \rightarrow 1$, we have $\pi_n \rightarrow 1$ as $n \rightarrow \infty$. On the other hand, with probability one, we have $\tGn_{1,2} \in \mGn_{\vp, \vq}$. 
Also, by the definition of $\pi_n$, we have $\sum_{\Hn_{1,2} \in \mAn}  \PER(\Hn_{1,2}) = \pi_n$. 
Thereby, employing the bound  \eqref{eq:er-converse-log-PER-lowerbound}, we have 
\begin{equation}
  \label{eq:er-converse-HtG-liminf}
  \liminf_{n \rightarrow \infty} \frac{H(\tGn_{1,2}) - n \frac{\der_{1,2}}{2} \log n }{n} \geq \bch(\muer_{1,2}).
\end{equation}
Now, using the assumption \eqref{eq:ER-converese-code-Ln} together with the bound \eqref{eq:ER-conv-sum-1}, we have 
\begin{equation}
\label{eq:er-conv-sum-0-1}
\begin{aligned}
  0 & \geq \limsup_{n \rightarrow \infty} \frac{ \log \Ln_1 + \log \Ln_2 - (\alpha_1 + \alpha_2) n \log n - n (R_1 + R_2)}{n} \\
  &\geq \liminf_{n \rightarrow \infty} \frac{H(\tGn_{1,2}) - n \frac{\der_{1,2}}{2} \log n - n \bch(\muer_{1,2})}{n} + \liminf_{n \rightarrow \infty} \frac{n \frac{\der_{1,2}}{2} \log n + n \bch(\muer_{1,2}) - (\alpha_1 + \alpha_2) n\log n - n (R_1 + R_2)}{n}.
\end{aligned}
\end{equation}
The first term is nonnegative due to \eqref{eq:er-converse-HtG-liminf}. Consequently, 
\begin{equation}
\label{eq:er-conv-sum-0-2}
  0 \geq \liminf_{n \rightarrow \infty} \frac{n \left ( \frac{\der_{1,2}}{2} - \alpha_1 - \alpha_2  \right ) \log n + n (\bch(\muer_{1,2}) - R_1 - R_2) }{n}.
\end{equation}
Note that this is impossible unless $\alpha_1 + \alpha_2 \geq \der_{1,2} / 2$. Furthermore, if $\alpha_1 + \alpha_2 = \der_{1,2}$, it must be the case that $R_1 + R_2 \geq \bch(\muer_{1,2})$. But this is precisely \eqref{eq:thm-assumption-12} for $\mu_{1,2} = \muer_{1,2}$. 

Now, we turn to showing \eqref{eq:thm-assumption-1}. We have
\begin{equation}
  \label{eq:er-conv-ln1-bound-1}
\begin{aligned}
  \log \Ln_1 &\geq H(\tMn_1) \geq H(\tMn_1|\tMn_2) \\
  &= H(\tGn_1 , \tMn_1 | \tMn_2) - H(\tGn_1 | \tMn_1, \tMn_2) \\
  &\stackrel{(a)}{=} H(\tGn_1 | \tMn_2) \\
  &\stackrel{(b)}{\geq} H(\tGn_1 | \tGn_2) \\
  &= H(\tGn_{1,2}) - H(\tGn_2),
\end{aligned}  
\end{equation}
where $(a)$ uses the facts that $\tMn_1$ is a function of $\tGn_1$ and also, since $\tGn_{1,2} \in \mAn$, given $\tMn_1$ and $\tMn_2$ we can unambiguously determine $\tGn_{1,2}$ and hence $\tGn_1$. Also, $(b)$ uses data processing inequality. Now, we find an upper bound for $H(\tGn_2)$. Note that since $\tGn_{1,2} \in \mAn$ with probability one, we have 
\begin{equation}
\label{eq:er-conv-H-tGn2-upp-1}
  H(\tGn_2) \leq \log |\mAn_2|,
\end{equation}
where
\begin{equation*}
  \mAn_2 := \{ \Hn_2 \in \mGn_2: \Hn_1 \oplus \Hn_2 \in \mAn \text{ for some } \Hn_1 \in \mGn_1 \}.
\end{equation*}
Now, take $\Hn_2 \in \mAn_2$ and let $\Hn_1 \in \mGn_1$ be such that $\Hn_{1,2}
:= \Hn_1
\oplus \Hn_2 \in \mAn$. 
Since $\mAn \subseteq \mGn_{\vp, \vq}$, by definition we have that, for all $x \in \Xi_{1,2}$ and all $\theta \in \Theta_{1,2}$, 
\begin{equation*}
  \sum_{x \in \Xi_{1,2}} |m_{\Hn_{1,2}}(x) - n p_x / 2| \leq n^{2/3} 
  \mbox{ and } \sum_{\theta \in \Theta_{1,2}} |u_{\Hn_{1,2}}(\theta)  - nq_\theta| \leq n^{2/3}.
\end{equation*}
Moreover,  for $x_2 \in \Xi_2$ and $\theta_2 \in \Theta_2$ we have $m_{\Hn_2}(x_2) = \sum_{x_1 \in \Xi_1 \cup \{\circ_1\}} m_{\Hn_{1,2}}((x_1, x_2))$ and $u_{\Hn_2}(\theta_2) = \sum_{\theta_1 \in \Theta_1} m_{\Hn_{1,2}}((\theta_1, \theta_2))$. Using this in the above and 
using the triangle inequality, 
we realize that 
for $\Hn_2 \in \mAn_2$ we have 
$\vm_{\Hn_2} \in \mMn_2$ and $\vu_{\Hn_2} \in \mUn_2$, where $\mMn_2$ is the set of edge mark count vectors $\vm$ such that $\sum_{x_2 \in \Xi_2} |m(x_2) - np_{x_2} / 2| \leq n^{2/3}$ and $\mUn_2$ is the set of vertex mark count vectors $\vu$  such that $\sum_{\theta_2 \in \Theta_2} |u(\theta_2) - nq_{\theta_2}| \leq n^{2/3}$. Consequently, we have
\begin{equation*}
  |\mAn_2| \leq (2 n^{2/3}+1)^{(|\Xi_2| + |\Theta_2|)} \left( \max_{\vm \in \mMn_2} \binom{\binom{n}{2}}{\{m(x_2)\}_{x_2 \in \Xi_2}} \right) \left( \max_{\vu \in \mUn_2} \binom{n}{\{u(\theta_2)\}_{\theta_2 \in \Theta_2}} \right).
\end{equation*}
Using Lemma~\ref{lem:binom-assymp} and the definition of $\mMn_2$ and $\mUn_2$
above, with $Q = (Q_1, Q_2) \sim \vq$, an argument 
similar to the one that was used to establish \eqref{eq:log-GnpQ-bch-12} 
implies that 
\begin{equation*}
  \limsup_{n \rightarrow \infty} \frac{\log |\mAn_2| - n \frac{\der_2}{2} \log n}{n} \leq  H(Q_2) + \sum_{x_2 \in \Xi_2} s(p_{x_2})  = \bch(\muer_2).
\end{equation*}
Substituting this into 
\eqref{eq:er-conv-H-tGn2-upp-1}, we get 
\begin{equation*}
  \limsup_{n \rightarrow \infty} \frac{\log H(\tGn_2) - n \frac{\der_2}{2} \log n}{n} \leq  \bch(\muer_2).
\end{equation*}
Using this together with \eqref{eq:er-converse-HtG-liminf} and substituting into \eqref{eq:er-conv-ln1-bound-1} we get 
\begin{equation*}
  \liminf_{n \rightarrow \infty} \frac{ \log \Ln_1 - n \frac{\der_{1,2} - \der_2}{2} \log n}{n} \geq \bch(\muer_{1,2}) - \bch(\muer_2) = \bch(\muer_1 | \muer_2).
\end{equation*}
Using a similar method as in \eqref{eq:er-conv-sum-0-1} and \eqref{eq:er-conv-sum-0-2}, this implies \eqref{eq:thm-assumption-1}. The proof of \eqref{eq:thm-assumption-2} is similar. This completes the proof of the converse for the \ER case.

\subsection{Proof of the Converse for the Configuration Model}
\label{sec:proof-converse-CM}

The proof of the converse for the configuration model is similar to that for the \ER model presented in the previous section. Take an achievable rate tuple $(\alpha_1, R_1, \alpha_2, R_2)$ together with a sequence of $\langle n, \Ln_1, \Ln_2 \rangle$ codes $(\fn_1, \fn_2, \decn)$ achieving this rate tuple. Moreover, redefine the set $\mAn$ to be 
\begin{equation}
  \label{eq:cm-mAn-def}
  \mAn := \mWn \cap \{ \Hn_{1,2} \in \mGn_{1,2}: \decn(\fn_1(\Hn_1), \fn_2(\Hn_2)) = \Hn_{1,2} \},
\end{equation}
where the set $\mWn$ was defined in Section~\ref{sec:proof-achievability-conf}. 
Now, let $\Gn_{1,2} \sim \mG(n; \vdn, \vgamma, \vq, \vr)$ be distributed according to the configuration model ensemble, and let  
$\tGn_{1,2} \in \mAn$ have the distribution obtained from 
that of $\Gn_{1,2}$ by conditioning on it lying in the set $\mAn$. 
Note that the normalizing constant $\pi_n := \prs{\Gn_{1,2} \in \mAn}$ goes to 1 as $n \rightarrow \infty$ since $\prs{\Gn_{1,2} \in \mWn} \rightarrow 1$ and the error probability of the code, $\Pn_e$, vanishes. Moreover, let $\PCM$ and $\tPCM$ be the laws of $\Gn_{1,2}$ and $\tGn_{1,2}$, respectively. In the following, we show that 
\begin{equation}
  \label{eq:cm-conv-ent-tGn12}
  \liminf_{n \rightarrow \infty} \frac{H(\tGn_{1,2}) - n \frac{\dcm_{1,2}}{2} \log n}{n} \geq \bch(\mucm_{1,2}),
\end{equation}
and 
\begin{equation}
  \label{eq:cm-conv-ent-tGn2}
  \limsup_{n \rightarrow \infty} \frac{H(\tGn_{2}) - n \frac{\dcm_{2}}{2} \log n}{n} \leq \bch(\mucm_2).
\end{equation}
The rest of the proof is then identical to that of the previous section, so we only focus on  proving 
the statements in~\eqref{eq:cm-conv-ent-tGn12}
and~\eqref{eq:cm-conv-ent-tGn2}. 

For \eqref{eq:cm-conv-ent-tGn12}, note that for $\Hn_{1,2} \in \mGn_{1,2}$ such that $\vdg_{\Hn_{1,2}} \in \mDn$, where $\mDn$ was defined in Section~\ref{sec:proof-achievability-conf}, we have 
\begin{equation*}
  - \log \PCM(\Hn_{1,2}) = \log \binom{n}{\{c_k(\vdn)\}_{k = 0}^\Delta} + \log |\mGn_{\vdn}| - \sum_{x \in \Xi_{1,2}} m_{\Hn_{1,2}}(x) \log \gamma_x - \sum_{\theta \in \Theta_{1,2}} u_{\Hn_{1,2}}(\theta) \log q_\theta.
\end{equation*}
Now, if $\Hn_{1,2} \in \mWn$, using the definition of $\mWn$ we realize that there exists a  constant $c > 0$ such that 
\begin{equation*}
  - \log \PCM(\Hn_{1,2}) \geq \log \binom{n}{\{c_k(\vdn)\}_{k = 0}^\Delta} + \log |\mGn_{\vdn}| - \sum_{x \in \Xi_{1,2}} m_n \gamma_x \log \gamma_x - \sum_{\theta \in \Theta_{1,2}} n q_\theta \log q_\theta - c n^{2/3} =: K_n.
\end{equation*}
Note that the right hand side is a constant independent of $\Hn_{1,2}$ and is denoted by $K_n$. Since $\tGn_{1,2}$ falls in $\mWn$ with probability one, this means that $H(\tGn_{1,2}) \geq \log \pi_n + K_n$. But $\pi_n \rightarrow 1$ as $n \rightarrow \infty$. Therefore, using the assumption~\eqref{eq:dn-r-n23} together with \eqref{eq:log-Gndn-bound} from Section~\ref{sec:proof-achievability-conf} and also the fact that $m_n / n \rightarrow \dcm_{1,2} / 2$, we realize that 
\begin{equation*}
  \liminf_{n \rightarrow \infty} \frac{H(\tGn_{1,2}) - n \frac{\dcm_{1,2}}{2} \log n }{n} \geq H(X) - s(\dcm_{1,2}) - \ev{\log X!} + \frac{\dcm_{1,2}}{2} H(\Gamma) + H(Q),
\end{equation*}
where $X \sim \vr$, $\Gamma \sim \vgamma$ and $Q \sim \vq$. Note that the right hand side is precisely $\bch(\mucm_{1,2})$. Hence we have proved \eqref{eq:cm-conv-ent-tGn12}.

In order to show \eqref{eq:cm-conv-ent-tGn2}, note that $H(\tGn_2) \leq \log |\mAn_2|$ where $\mAn_2$ consists of graphs $\Hn_2 \in \mGn_2$ such that, for some $\Hn_1 \in \mGn_1$, we have $\Hn_1 \oplus \Hn_2 \in \mAn$. Since $\mAn \subseteq \mWn$, we have for all $\Hn_2 \in \mAn_2$ that
\begin{equation}
  \label{eq:cm-conv-count-n23}
  \sum_{x_2 \in \Xi_2} |m_{\Hn_2}(x_2) - m_n \gamma_{x_2}| \leq n^{2/3} \mbox{ and } \sum_{\theta_2 \in \Theta_2} |u_{\Hn_2}(\theta_2) - n q_{\theta_2}| \leq n^{2/3}.
\end{equation}
On the other hand, the condition \eqref{eq:mWn-deg-count-2} implies that $\vdg_{\Hn_2} \in \mDn_2$ where $\mDn_2$ denotes the set of degree sequences $\vd$ of size $n$ with elements bounded by $\Delta$ such that
\begin{equation}
  \label{eq:cm-conv-Dn2}
  |c_k(\vd) - n \pr{X_2 = k}| \leq (\Delta + 1) n^{2/3}, \qquad \forall 0 \leq k \leq \Delta,
\end{equation}
where $X_2$ is the random variable defined in \eqref{eq:X1-X2-def}. Consequently, we have 
\begin{equation}
  \label{eq:cm-conv-An2-bound-1}
  \begin{aligned}
  \log |\mAn_2| &\leq \log |\mDn_2| + \max_{\vd \in \mDn_2} \log |\mGn_{\vd}| + \max_{\Hn_2 \in \mAn_2} \log \binom{\sum_{x_2 \in \Xi_2} m_{\Hn_2}(x_2)}{\{m_{\Hn_2}(x_2)\}_{x_2 \in \Xi_2}} \\
  &\qquad + \max_{\Hn_2 \in \mAn_2} \log \binom{n}{\{u_{\Hn_2}(\theta_2)\}_{\theta_2 \in \Theta_2} }.
  \end{aligned}
\end{equation}
Note that \eqref{eq:cm-conv-Dn2} implies that 
$|\mDn_2| \leq (2(\Delta +1) n^{2/3}+1)^{\Delta + 1} \max_{\vd \in \mDn_2} \binom{n}{\{c_k(\vd)\}_{k=0}^\Delta}$. Therefore,  Lemma~\ref{lem:binom-assymp} implies that 
\begin{equation}
  \label{eq:cm-conv-mDn-2-bound}
  \limsup_{n \rightarrow \infty} \frac{1}{n} \log |\mDn_2| \leq H(X_2).
\end{equation}
On the other hand, the assumptions $r_0 < 1$ and~\eqref{eq:CM-gamma-marginal-assumption} imply that $\dcm_2 > 0$. Hence, using Lemma~\ref{lem:degree-cm-count}, 
we have 
\begin{equation}
  \label{eq:cm-conv-max-mGn-vd-bound}
  \limsup_{n \rightarrow \infty} \frac{\max_{\vd \in \mDn_2} \log |\mGn_{\vd} | - n \frac{\dcm_2}{2}\log n }{n} \leq -s(\dcm_2) - \ev{\log X_2!}.
\end{equation}
Moreover, if $\Hn_2$ is a sequence in $\mAn_2$, from \eqref{eq:cm-conv-count-n23}, for all $x_2 \in \Xi_2$, we have
\begin{equation*}
  \lim_{n \rightarrow \infty} \frac{m_{\Hn_2}(x_2)}{\sum_{x'_2 \in \Xi_2} m_{\Hn_2}(x'_2)} = \frac{\gamma_{x_2}}{\sum_{x'_2 \in \Xi_2} \gamma_{x'_2}} = \pr{\Gamma_2 = x_2 | \Gamma_2 \neq \circ_2},
\end{equation*}
where $\Gamma = (\Gamma_1, \Gamma_2)$ has law $\vgamma$. Additionally, we have 
\begin{equation*}
  \lim_{n \rightarrow \infty} \frac{1}{n} \sum_{x_2 \in \Xi_2} m_{\Hn_2}(x_2) = \frac{\dcm_2}{2}.
\end{equation*}
Thereby, from Lemma~\ref{lem:binom-assymp}, we have 
\begin{equation}
  \label{eq:cm-conv-max-m-bound}
  \limsup_{n \rightarrow \infty} \frac{1}{n} \max_{\Hn_2 \in \mAn_2} \log \binom{\sum_{x_2 \in \Xi_2} m_{\Hn_2}(x_2)}{\{m_{\Hn_2}(x_2)\}_{x_2 \in \Xi_2}} \leq \frac{\dcm_2}{2} H(\Gamma_2 | \Gamma_2 \neq \circ_2).
\end{equation}
Finally, as we have $u_{\Hn_2}(\theta_2) / n \rightarrow q_{\theta_2}$ for all $\theta_2 \in \Theta_2$, another usage of Lemma~\ref{lem:binom-assymp} implies that
\begin{equation}
  \label{eq:cm-conv-max-u-bound}
  \limsup_{n \rightarrow \infty} \frac{1}{n} \max_{\Hn_2 \in \mAn_2} \log \binom{n}{\{u_{\Hn_2}(\theta_2)\}_{\theta_2 \in \Theta_2} } \leq H(Q_2),
\end{equation}
where $Q = (Q_1, Q_2)$ has law $\vq$. 
Now, combining \eqref{eq:cm-conv-mDn-2-bound}, \eqref{eq:cm-conv-max-mGn-vd-bound}, \eqref{eq:cm-conv-max-m-bound} and \eqref{eq:cm-conv-max-u-bound} and substituting into \eqref{eq:cm-conv-An2-bound-1}, and also using the bound $H(\tGn_2) \leq \log |\mAn_2|$, we realize that 
\begin{equation*}
  \limsup_{n \rightarrow \infty} \frac{H(\tGn_2) - n \frac{\dcm_2}{2} \log n}{n} \leq H(X_2) - s(\dcm_2) - \ev{\log X_2!} + \frac{\dcm_2}{2} H(\Gamma_2|\Gamma_2 \neq \circ_2) + H(Q_2).
\end{equation*}
But the right hand side is precisely $\bch(\mucm_2)$. This completes the proof of \eqref{eq:cm-conv-ent-tGn2}. As was mentioned before, the rest of the proof is identical to that in the previous section.

\subsection{Generalization to more than two sources}
\label{sec:gen-more-sources}
Assume we have $k \geq 2$ sources of graphical data. For $1 \leq i \leq k$, let
$\vermark_i$ and $\edgemark_i$
denote the vertex and edge mark sets for the $i$th domain.
For $i \in [k]$ and $n \in \nats$, $\mGn_i$ denotes the set of
marked graphs on the vertex set $[n]$ with vertex and edge marks coming from
$\vermark_i$ and $\edgemark_i$, respectively. Given $A \subseteq [k]$ nonempty
and for $G_i \in \mGn_i$, $i \in A$, we define  $\bigoplus_{i
\in A} G_i$ to be  the superposition of graphs in $A$, which is a simple marked graph on the
vertex set $[n]$ such that a vertex $v \in [n]$ carries a vertex mark $(\theta_i:
i \in A) \in \vermark_A := \prod_{i \in A} \vermark_i$ such that $\theta_i$ is
the mark of $v$ in $G_i$. Moreover, an edge between vertices $v$ and $w$ exists
in $\bigoplus_{i
\in A} G_i$ if such an edge exists 
in at least one of the graphs 
$G_i$, $i \in A$. If
this is the case, the mark of this edge is defined to be 
$(x_i: i \in A)$, where for $i \in A$, $x_i$ is the mark of the edge
$(v,w)$ in $G_i$ if such an edge exists in $G_i$. Otherwise, we set $x_i =
\circ_i$, where 
$\circ_i$ for $ i \in [k]$ is an auxiliary mark not present in
$\edgemark_i$. For nonempty $A \subseteq [k]$, we denote $(\circ_i: i \in A)$ by
$\circ_A$. 
Note that with $\edgemark_A := (\prod_{i \in A} (\edgemark_i \cup
\{\circ_i\})) \setminus \{\circ_A\}$, $\bigoplus_{i \in A} G_i$ is a
marked graph with vertex and edge mark sets $\vermark_A$ and $\edgemark_A$,
respectively. Let $\mGn_A$ denote 
the set of marked graphs in domain $A$, 
which is 
the set of marked graphs 
on the vertex set
$[n]$ together with vertex and edge mark sets $\vermark_A$ and $\edgemark_A$,
respectively. Given 
$G \in \mGn_{[k]}$ and $A \subset [k]$, 
we can naturally define
the projection of $G$ 
onto domain $A$ 
by projecting all vertex and edge marks
onto $\vermark_A$ and $\edgemark_A$, respectively, 
followed by removing edges
with mark $\circ_A$. 
It can be checked that the resulting graph, denoted by $G_A$, 
lies in domain $A$, i.e. $G_A \in \mGn_A$.

A sequence of $\langle  n, \Ln_i: i \in [k] \rangle$  codes 
is defined 
as a sequence of tuples $((\fn_i: i \in [k]), \decn)$ such that $\fn_i: \mGn_i
\rightarrow [\Ln_i]$ for $i \in A$ are encoding functions, and $\decn: \prod_{i
  \in [k]} [\Ln_i] \rightarrow \mGn_{[k]}$ is the corresponding decoding
function. Given a sequence of ensembles $\Gn_{[k]}$ on $\mGn_{[k]}$, the
probability of error $\Pn_e$ is defined to be the probability that
$\decn((\fn_i(\Gn_i): i \in [k])) \neq \Gn_{[k]}$.

We say that a rate tuple
$((\alpha_i, R_i): i \in [k])$ is achievable for the distributed compression of
the sequence of random graphs $\Gn_{[k]} \in \mGn_{[k]}$ if there is a sequence
of $\langle n, \Ln_i: i \in [k] \rangle$ codes such that for $i \in [k]$,
\begin{equation*}
  \limsup_{n \rightarrow \infty} \frac{\log \Ln_i - (\alpha_i n \log n + R_i n)}{n} \leq 0,
\end{equation*}
and also $\Pn_e \rightarrow 0$. 
We say that $((\alpha_i,R_i): i \in [k])$ lies in the rate region $\mR$ 
if there exist sequences $R^{(m)}_i$ for $i \in [k]$ 
such that $R^{(m)}_i \rightarrow R_i$ as $m \rightarrow \infty$ and, 
for each $m$,
$((\alpha_i, R^{(m)}_i): i \in [k])$ is achievable.

We can naturally generalize the \ER and the configuration model ensembles of
Section~\ref{sec:prel-notat} to
the above  setting.

\textbf{A sequence of \ER ensembles:} Given a sequence of nonnegative real
numbers $\vp = \{p_x\}_{x \in \edgemark_{[k]}}$ and a probability
distribution $\vq = \{q_\theta\}_{\theta \in \vermark_{[k]}}$, 
assume that for all $i \in [k]$ and $x_i \in \edgemark_i$ we have
\begin{equation}
  \label{eq:gen-er-p-positive}
  \sum_{(x'_j: j \in [k]) \in \edgemark_{[k]}: x'_i = x_i} p_{(x'_j: j \in [k])} > 0.
\end{equation}
Moreover, 
assume that for all $i \in [k]$ and all $\theta_i \in \vermark_i$
we have
\begin{equation}
  \label{eq:gen-er-q-positive}
  \sum_{(\theta'_i: i \in [k]) \in \vermark_{[k]}: \theta'_i = \theta_i} q_{(\theta'_i: i \in [k])} > 0.
\end{equation}
For $n \in \nats$ large enough, the probability distribution $\mG(n; \vp, \vq)$
on $\mGn_{[k]}$ is defined as follows: for each pair of vertices $1 \leq i < j
\leq n$, the edge $(i,j)$ exists and has a mark $x \in \edgemark_{[k]}$ with
probability $p_x / n$, and is not present with probability $1 - \sum_{x \in
  \edgemark_{[k]}} p_x / n$. Moreover, each vertex is independently given a mark
$\theta \in \vermark_{[k]}$ with probability $q_\theta$. 
The 
choices 
of edge and
vertex marks are done independently.

The conditions in~\eqref{eq:gen-er-p-positive}
and~\eqref{eq:gen-er-q-positive}
are required only to ensure
that the sets of vertex marks and edge marks are chosen to be as
small as possible, and these conditions could be relaxed if desired.

\textbf{A sequence of configuration model ensembles:} Similar to the
configuration model ensemble for two sources as we defined in Section~\ref{sec:prel-notat}, assume that $\Delta \in \nats$ and
a 
probability distribution $\vr = \{r_i\}_{i=0}^{\Delta}$ is given such that $r_0
< 1$. Moreover, for each $n$, the degree sequence $\vdn = \{\dn(1), \dots,
\dn(n)\}$ is given such that for $i \in [n]$, $\dn(i) \leq \Delta$, 
$\sum_{i=1}^n \dn(i)$ is even, and~\eqref{eq:dn-r-n23} is satisfied.
Additionally,  assume that probability distributions $\vgamma =
\{\gamma_x\}_{x \in \edgemark_{[k]}}$ and $\vq = \{q_{\theta}\}_{\theta \in
  \vermark_{[k]}}$ are given such that 
for all $i \in [k]$ and $x_i \in \edgemark_i$ we have
\begin{equation}
  \label{eq:general-conf-model-gamma-positive}
  \sum_{(x'_j: j \in [k]) \in \edgemark_{[k]}: x'_i = x_i} \gamma_{(x'_j: j \in [k])} > 0,
\end{equation}
and
for all $A \subset [k]$ nonempty,  $A \neq [k]$, we have
\begin{equation}
  \label{eq:general-conf-model-gamma-matters}
  \sum_{(x'_j: j \in [k]) \in \edgemark_{[k]}: (x'_i: i \in A) = \circ_A} \gamma_{(x'_j: j \in [k])} > 0.
\end{equation}
We also assume that for all $i \in [k]$ and 
$\theta_i \in \vermark_i$ we have 
\begin{equation}
  \label{eq:general-conf-model-q-positive}
  \sum_{(\theta'_j: j \in [k]) \in \vermark_{[k]}: \theta'_i = \theta_i} q_{(\theta'_j: j \in [k])} > 0.
\end{equation}
With these, for $n$ large enough, we define the probability distribution $\mG(n; \vdn, \vgamma, \vq,
\vr)$ on $\mGn_{[k]}$ as follows. Similar to the ensemble 
for
two sources, we
pick an unmarked graph on the vertex set $[n]$ uniformly at random  among the
set of graphs with maximum degree $\Delta$ such that for $0 \leq k \leq \Delta$,
$c_k(\vdg_G) = c_k(\vdn)$. Then, we assign i.i.d.\ marks with law $\vgamma$ 
on the edges 
and i.i.d.\ marks with law $\vq$ 
on the vertices. 

The conditions in~\eqref{eq:general-conf-model-gamma-positive}
and~\eqref{eq:general-conf-model-q-positive}
are required only to ensure
that the sets of vertex marks and edge marks are chosen to be as
small as possible, and these conditions could be relaxed if desired.
However, the conditions in~\eqref{eq:general-conf-model-gamma-matters}
are essential, since they ensure that for all
$A \subset [k]$ nonempty,  $A \neq [k]$,
the underlying unmarked graph of the projection of the 
overall graph onto domain $A$
is not a subgraph of
the underlying unmarked graph of the projection onto domain
$A^c$.  

Similar to our discussion in Section~\ref{sec:framework-local-weak}, it can be
seen that the local weak limit 
of the sequence of \ER ensembles above
is a marked Poisson
Galton--Watson tree, which we denote by $\muer_{[k]}$. Likewise, the local weak
limit 
of the sequence of configuration model ensembles above 
is a marked
Galton--Watson tree with degree distribution $\vr$, which we denote by $\mucm_{[k]}$.
For $A \subseteq [k]$ nonempty, we denote the projection of $\muer_{[k]}$ and
$\mucm_{[k]}$ to domain $A$ by $\muer_A$ and $\mucm_A$, respectively. For
nonempty $A \subset [k]$, $A \neq [k]$, we define $\bch(\muer_A | \muer_{A^c})$
to be $\bch(\muer_{[k]}) - \bch(\muer_{A^c})$. We similarly define $\bch(\mucm_A | \mucm_{A^c})$.

We are now ready to characterize the rate region 
for the multi-source scenarios above 
in the following Theorem~\ref{thm:graph-SW-k-source}. This is a
generalization of Theorem~\ref{thm:SW}, and its proof is similar to that of
Theorem~\ref{thm:SW}. We highlight the proof of
Theorem~\ref{thm:graph-SW-k-source} in Appendix~\ref{sec:app-multisource-proof}.

\begin{thm}
  \label{thm:graph-SW-k-source}
  Assume $\mu_{[k]}$ is either of the two distributions $\muer_{[k]}$ or
  $\mucm_{[k]}$ defined above. Then, if $\mR$ is the rate region for the
  sequence of ensembles 
  corresponding to $\mu_{[k]}$, as defined above, 
  a rate tuple
  $((\alpha_i, R_i): i \in [k]) \in \mR$  if and only if 
  for every nonempty 
  $A \subset
  [k]$, $A \neq [k]$, we have
  \begin{equation*}
    \left(\sum_{i \in A} \alpha_i, \sum_{i \in A} R_i\right) \succeq ((d_{[k]} - d_{A^c}) / 2, \bch(\mu_A | \mu_{A^c})),
  \end{equation*}
  and
  \begin{equation*}
    \left(\sum_{i \in [k]} \alpha_i, \sum_{i \in [k]} R_i\right) \succeq (d_{[k]} / 2, \bch(\mu_{[k]})),
  \end{equation*}
  where $d_{[k]} = \deg(\mu_{[k]})$ and $d_{A^c} = \deg(\mu_{A^c})$.
\end{thm}

\section{Conclusion}
\label{sec:conclusion}

We gave  a counterpart of the Slepian--Wolf Theorem for distributed compression of
graphical data,
employing the framework of local weak convergence. We derived the rate region
for two families of sequences of graph ensembles, namely sequences 
of \ER ensembles 
having a local weak limit and sequences 
of configuration model ensembles 
having a local
weak limit. 
Furthermore, we gave a generalization of this result
  for \ER and configuration model ensembles with more than two sources.

\section*{Acknowledgments}
The authors acknowledge support from the NSF grants ECCS--1343398, CNS--1527846, CCF--1618145, 
CCF--1901004, 
the NSF Science \& Technology Center grant CCF--0939370 (Science of Information), and the William and Flora Hewlett Foundation
supported Center for Long Term Cybersecurity at Berkeley.


\appendix

\section{Proof of Lemma~\ref{lem:binom-assymp}}
\label{sec:lemma-Stirling-proof}

Throughout this section, we treat $0
\log 0$ as equal to $0$. 
Consider the first part of Lemma~\ref{lem:binom-assymp}. 
Since $a_n / n \rightarrow a > 0$ as $n \to \infty$,
using Stirling's approximation we have 
$\log a_n! = a_n \log a_n - a_n + 
o(n)$.
Similarly, from the assumption that $b^n_i / n \rightarrow b_i \geq 0$ as $n \to \infty$ for $1 \leq i \leq k$, we have 
$\log b^n_i! = b^n_i \log b^n_i - b^n_i +
o(n)$, which holds irrespective of whether $b_i > 0$ or $b_i = 0$.
Hence we have 
\begin{align*}
\log \binom{a_n}{\{b^n_i\}_{1 \leq i \leq k}}  &= a_n \log a_n - a_n 
- \sum_{i=1}^k b^n_i \log b^n_i + \sum_{i=1}^k b^n_i + o(n)\\
&= a_n \log \frac{a_n}{n} - \sum_{i=1}^k b^n_i \log \frac{b^n_i}{n} + o(n),
\end{align*}
where we have used $a_n = \sum_{i=1}^k b^n_i$.
This gives
\begin{align*}
  \lim_{n \rightarrow \infty} \frac{1}{n} \log \binom{a_n}{\{b^n_i\}_{1 \leq i \leq k}} &= 
a \log a - \sum_{i=1}^k b_i \log b_i \\
&= a H\left( \left\{ \frac{b_i}{a} \right\}_{1 \leq i \leq k} \right). 
\end{align*}

Next, consider the second part of Lemma~\ref{lem:binom-assymp}.
Since $a_n / \binom{n}{2} \rightarrow 1$ as $n \to \infty$, using
Stirling's approximation we have 
$\log a_n! = a_n \log a_n - a_n + 
o(n)$.
As noted earlier, since
$b^n_i / n \rightarrow b_i \geq 0$ as $n \to \infty$ for $1 \leq i \leq k$, we have 
$\log b^n_i! = b^n_i \log b^n_i - b^n_i +
o(n)$, which holds irrespective of whether $b_i > 0$ or $b_i = 0$.
Moreover, with $b_n := \sum_{i=1}^k b^n_i$, we have $\log (a_n - b_n)! = (a_n -
b_n) \log (a_n - b_n) - (a_n - b_n) + 
o(n)$.
Therefore, we have 
\begin{align*}
 \log \binom{a_n}{\{b^n_i\}_{1 \leq i \leq k}} &= a_n \log a_n - a_n 
 - \sum_{i=1}^k b^n_i \log b^n_i + \sum_{i=1}^k b^n_i 
 - (a_n - b_n) \log (a_n - b_n) + (a_n - b_n) + o(n)\\
 &= a_n \log \frac{a_n}{n} - \sum_{i=1}^k b^n_i \log \frac{b^n_i}{n} 
 - (a_n - b_n) \log \frac{a_n - b_n}{n} +o(n),
\end{align*}
where we have used $a_n = b_n + (a_n -b_n)$.
This gives
\begin{equation}
  \label{eq:stirling-lemma-1}
  \begin{aligned}
    \frac{1}{n} \log \binom{a_n}{\{b^n_i\}_{1 \leq i \leq k}} 
    &= \frac{a_n}{n} \log \frac{a_n}{n} 
    - \sum_{i=1}^k \frac{b^n_i}{n} \log \frac{b^n_i}{n} 
    - \frac{a_n - b_n}{n} \log \frac{a_n - b_n}{n} + o(1) \\
    &= - \frac{a_n}{n} \log \left( 1 - \frac{b_n}{a_n} \right) + \frac{b_n}{n} \log \frac{a_n - b_n}{\frac{n^2}{2}} + \frac{b_n}{n} \log \frac{n}{2} - \sum_{i=1}^k \frac{b^n_i}{n} \log \frac{b^n_i}{n} + o(1).
  \end{aligned}
\end{equation}
Since $b_n / a_n \rightarrow 0$ as $n \to \infty$,
we write
$\log (1 - b_n / a_n) = - b_n/a_n + O(b_n^2 / a_n^2)$.  Consequently, we have
\begin{equation*}
  -\frac{a_n}{n} \log \left( 1 - \frac{b_n}{a_n} \right)  = \frac{b_n}{n} + O\left( \frac{b_n^2}{n a_n} \right),
\end{equation*}
and since $b_n^2 / (n a_n) \rightarrow 0$
we have
\begin{equation}
  \label{eq:stirling-appendix-star1}
\lim_{n \rightarrow \infty}  -\frac{a_n}{n} \log \left( 1 - \frac{b_n}{a_n} \right) = b.
\end{equation}
Further, since $(a_n - b_n)/(n^2/2) \rightarrow 1$
we have
\begin{equation}
  \label{eq:stirling-appendix-star2}
   \lim_{n \rightarrow \infty} \frac{b_n}{n} \log \frac{a_n - b_n}{\frac{n^2}{2}}  = 0.
\end{equation}
Using \eqref{eq:stirling-appendix-star1} and~\eqref{eq:stirling-appendix-star2}
in~\eqref{eq:stirling-lemma-1}, we get
\begin{align*}
  \lim_{n \rightarrow \infty} \frac{ \log \binom{a_n}{\{b^n_i\}_{1 \leq i \leq k}} - b_n \log n}{n} &= b - b \log 2 - \sum_{i=1}^k b_i \log b_i \\
  &= \sum_{i=1}^k s(2b_i),
\end{align*}
which completes the proof.

\section{Calculations for Deriving \eqref{EQ:LOG-MGNMNUN-STIRLING}}
\label{sec:app-Stirling}

Note that we have
\begin{equation}
  \label{eq:sizegnmnun-exact}
  |\mGn_{\vmn,\vun} | = \frac{n!}{\prod_{\theta \in \vermark} \un(\theta)!} \times \frac{\frac{n(n-1)}{2}!}{\prod_{x \in \edgemark} \mn(x)! \times \left ( \frac{n(n-1)}{2} - \snorm{\vmn}_1 \right)!}.
\end{equation}
Since $u^{(n)}(\theta) / n \rightarrow q_\theta$ for all $\theta \in
\vermark$, from part 1 of Lemma~\ref{lem:binom-assymp} we have
\begin{equation}
  \label{eq:n!-un!-HQ-on}
  \log \frac{n!}{\prod_{\theta \in \vermark} u^{(n)}(\theta)!} = n H(Q) + o(n).
\end{equation}
Moreover, since for all $x \in \edgemark$ we have
$\mn(x) / n \rightarrow d_{x}/2 < \infty$,  
from part 2 of Lemma~\ref{lem:binom-assymp} we have
\begin{equation}
  \label{eq:new-lemma-1-part-2}
\log \frac{\frac{n(n-1)}{2}!}{\prod_{x\in \edgemark} \mn(x)! \times \left ( \frac{n(n-1)}{2} - \snorm{\vmn}_1 \right)!} =
\snorm{\vmn}_1 \log n + n \sum_{x} s(d_{x}) + o(n).
  \end{equation}
Using \eqref{eq:n!-un!-HQ-on} and \eqref{eq:new-lemma-1-part-2} in
\eqref{eq:sizegnmnun-exact}, we get 
\begin{align*}
  \log |\mGnmnun| &= \snorm{\vmn}_1 \log n + n H(Q) + n \sum_{x} s(d_{x}) + o(n),
\end{align*}
which is precisely what was stated in~\eqref{EQ:LOG-MGNMNUN-STIRLING}.

\editfinish

\editstart 
\section{Proof of Lemma~\ref{lem:degree-cm-count}}
\label{sec:lemma-Stirling-proof-second}


The assumptions of the lemma imply that 
  $b_n/n \rightarrow d/2 > 0$ and, in particular, $b_n \rightarrow \infty$ as $n
  \rightarrow \infty$. Therefore, Theorem~4.6 in \cite{mckay1985asymptotics}
  implies that
  \begin{equation*}
    \lim_{n \rightarrow \infty} \frac{\displaystyle |\mGn_{\van}|}{\displaystyle \alpha_n \frac{(b_n-1)!!}{\prod_{i=1}^n \an(i)!}} = 1,
  \end{equation*}
  where
  \begin{equation*}
    \alpha_n := \exp\left(-\lambda_n - \lambda_n^2\right), \qquad \lambda_n := \frac{1}{2b_n} \sum_{i=1}^n \an(i) (\an(i) - 1),
  \end{equation*}
and
\begin{equation*}
  (b_n-1)!! := (b_n-1)\times (b_n-3) \times \dots \times 3 \times 1 = \frac{b_n!}{2^{b_n/2}(b_n/2)!}.
\end{equation*}
Under the assumptions of the lemma, we have
$b_n / n\rightarrow d/2$ as $n \to \infty$. Therefore, using Stirling's approximation, we have $\log (b_n-1)!! = \frac{b_n}{2} \log n - n s(d) +o(n)$. Moreover, since
$c_k(\van) / n \rightarrow \pr{Y = k}$ as
$n \to \infty$ for all $0 \leq k \leq \Delta$, we have \[
\frac{1}{n} \log \prod_{i=1}^n \an(i)! = 
\frac{1}{n} \sum_{k=0}^\Delta c_k(\van) \log k! = \ev{\log Y!} + o(1).
\]
On the other hand, we have
  \begin{align*}
    \lim_{n \rightarrow \infty} \lambda_n &= \lim_{n \rightarrow \infty} \frac{1}{2b_n / n } \frac{1}{n} \sum_{i=1}^n \an(i) (\an(i) - 1)\\
                                          &= \frac{1}{d} \lim_{n \rightarrow \infty} \frac{1}{n} \sum_{k=1}^\Delta c_k(\van) k(k-1) \\
                                          &= \frac{1}{d} \ev{Y(Y-1)} =: \lambda.
  \end{align*}
  This implies that, as $n \rightarrow \infty$, $\alpha_n \rightarrow \alpha :=
  \exp(- \lambda - \lambda^2) >0$. Therefore, $\frac{1}{n} \log \alpha_n
  \rightarrow 0$ as $n \rightarrow \infty$.
Putting these together, we get the desired result. 
%

\editfinish

\editstart

\section{Asymptotic behavior of the entropy of the configuration model }
\label{sec:asympt-behav-enropy-cm}

Here, we prove \eqref{eq:ent-assympt-cm-12}--\eqref{eq:ent-assympt-cm-2}. 
Let $X$ be a random variable with law $\vr$, 
and let $X_1$ and $X_2$ be defined 
as in \eqref{eq:X1-X2-def}.
Let $\Gamma = (\Gamma_1,
\Gamma_2)$ and $Q = (Q_1, Q_2)$ denote random variables with laws $\vgamma$ and
$\vq$, respectively.
Let $\beta_1 := \prs{\Gamma_1 \neq \circ_1}$ 
and let
$\tilde{\Gamma}_1$ be a random variable on $\Xi_1$ with the law of $\Gamma_1$
conditioned on $\Gamma_1 \neq \circ_1$.

As in Section \ref{sec:proof-achievability-conf}, we let $\mDn$ denote the set of degree sequences $\vd = (d(1), \dots, d(n))$ 
with entries bounded by $\Delta$ such that $c_k(\vd) = c_k(\vdn)$ for all $0 \leq k \leq \Delta$.
Let $\Fn_{1,2}$ be a simple unmarked
graph chosen uniformly at random from the set $\cup_{\vd \in \mDn} \mGn_{\vd}$,
where we recall that 
$\mGn_{\vd}$ denotes the set of simple unmarked graphs $G$ on the vertex set $[n]$ such that $\dg_G(i) = d(i)$ for $1 \leq i \leq n$.
By definition, $\Gn_{1,2} \sim \mG(n; \vdn, \vgamma, \vq, \vr)$ is obtained from
$\Fn_{1,2}$ by adding independent edge and vertex marks according to the laws  of $\vgamma$ and $\vq$ respectively.
If we first create $\Gn_{1,2}$ from 
$\Fn_{1,2}$, and then drop the edges with the
first domain mark $\circ_1$, if $\Fn_1$ denotes the unmarked version of the
resulting marked graph, then $\Fn_1$ is effectively obtained from $\Fn_{1,2}$ by
independently removing each edge
with probability $1-\beta_1$.
Also, the corresponding first domain marked graph, i.e. $\Gn_1$, obtained from $\Gn_{1,2}$ in this way is effectively
obtained from $\Fn_1$ by adding independent vertex and edge marks to $\Fn_1$ with the laws of $Q_1$ and $\tilde{\Gamma}_1$, respectively. 
With this viewpoint, we may consider $\Gn_{1,2}$, $\Fn_{1,2}$, $\Gn_1$ and
$\Fn_1$ 
as being defined on 
a joint probability space. 


As in Section \ref{sec:proof-achievability-conf}, we let
$\mWn$ denote the set of graphs $\Hn_{1,2} \in \mGn_{1,2}$ such that: $(i)$ $\vdg_{\Hn_{1,2}} \in \mDn$, $(ii)$ $\vm_{\Hn_{1,2}} \in \mMn$, $(iii)$ $\vu_{\Hn_{1,2}} \in \mUn$,  $(iv)$ for all $0 \leq l \leq k \leq \Delta$, recalling the notation in \eqref{eq:def-ckl}, we have
\begin{equation*}
  |c_{k,l}(\vdg_{\Hn_{1,2}}, \vdg_{\Hn_1}) - n \pr{X = k,X_1 = l}| \leq n^{2/3},
\end{equation*}
 and $(v)$ for all $ 0 \leq l \leq k \leq \Delta$ 
 we have 
 \begin{equation*}
 |c_{k,l}(\vdg_{\Hn_{1,2}}, \vdg_{\Hn_2}) - n \pr{X = k, X_2 = l}| \leq n^{2/3}.
 \end{equation*}
Here, as in Section \ref{sec:proof-achievability-conf}, $\mMn$ denotes the set of mark count vectors $\vm$ such that $\sum_{x \in \Xi_{1,2}} m(x) = m_n$ and $\sum_{x \in \Xi_{1,2}} |m(x) - m_n \gamma_{x}| \leq n^{2/3}$, where we recall that $m_n := (\sum_{i=1}^n \dn(i)) / 2$, while,
as in Section \ref{sec:proof-achievability-conf},
$\mUn$ denotes the set of vertex mark count vectors $\vu$ such that $\sum_{\theta \in \Theta_{1,2}} |u(\theta) - n q_\theta| \leq n^{2/3}$. 


We can now prove the following lemma.

\begin{lem}
\label{lem:mWn-highprobability}
If $\Gn_{1,2} \sim \mG(n; \vdn, \vgamma, \vq, \vr)$, 
we have $\prs{\Gn_{1,2} \notin \mWn} \leq \kappa n^{-1/3}$ for some constant $\kappa > 0$.
\end{lem}

\begin{proof}
Condition $(i)$ in the definition of $\mWn$ holds for every realization of $\Gn_{1,2}$. Chebyshev's inequality implies that conditions $(ii)$ and $(iii)$ hold with probability at least $1 - \kappa_1 n^{-1/3}$, for some $\kappa_1>0$. To show $(iv)$, fix $ 0 \leq l \leq k \leq \Delta$ and, for $1 \leq i \leq n$, let $Y_i$ be the indicator of the event that $\dg_{\Gn_{1,2}}(i) = k$ and $\dg_{\Gn_1}(i) = l$. With $Y := \sum_{i=1}^n Y_i$, we have  $c_{k,l}(\vdg_{\Gn_{1,2}}, \vdg_{\Gn_1}) = Y$. 
Note that an edge of $\Gn_{1,2}$ exists in $\Gn_1$ if its mark is not of the form $(\circ_1, x_2)$, which happens with probability $\beta_1$. Therefore, 
\begin{equation*}
  \ev{Y_i | \Fn_{1,2}} = \one{\dg_{\Fn_{1,2}}(i) = k} \binom{\dg_{\Fn_{1,2}}(i)}{l} \beta_1^l (1-\beta_1)^{k-l}.
\end{equation*}
Consequently, 
\begin{equation*}
  \ev{Y | \Fn_{1,2}} = c_k(\vdn) \binom{k}{l} \beta_1^l (1-\beta_1)^{k-l}.
\end{equation*}
Since this is a constant, it is also equal to $\ev{Y}$. 
Now, if $s_{k,l} := \pr{X = k, X_1 = l}$, we have $s_{k,l} = r_k \binom{k}{l} \beta_1^l (1-\beta_1)^{k-l}$. Hence the assumption in \eqref{eq:dn-r-n23} implies that 
\begin{equation}
  \label{eq:evX-nskl-bound}
|\ev{Y} - n s_{k,l}| \leq K n^{1/2}\binom{k}{l} \beta_1^l (1-\beta_1)^{k-l}.  
\end{equation}
Furthermore, since edge marks are chosen independently
conditioned on  $\Fn_{1,2}$,
 if $i$ and $j$ are nonadjacent vertices in $\Fn_{1,2}$, then $Y_i$ are $Y_j$
 are conditionally  independent,
 conditioned on  $\Fn_{1,2}$.
 As a result, if $\mI$ denotes the set of $(i,j)$ with  $1 \leq i \neq j \leq n$ such that $i$ and $j$ are not adjacent in $\Fn_{1,2}$, we have 
\begin{align*}
  \ev{Y^2 | \Fn_{1,2}} &= \sum_{i=1}^n \ev{Y_i^2|\Fn_{1,2}} + \sum_{1 \leq i \neq j \leq n} \ev{Y_i Y_j|\Fn_{1,2}} \\
  &\leq n + \sum_{(i,j) \notin \mI} \ev{Y_i Y_j|\Fn_{1,2}} + \sum_{(i,j) \in \mI} \ev{Y_i Y_j|\Fn_{1,2}} \\
  &\leq n + 2m_n + \sum_{(i,j) \in \mI} \ev{Y_i Y_j|\Fn_{1,2}} \\
  &\stackrel{(a)}{=} n + 2m_n + \sum_{(i,j) \in \mI} \ev{Y_i | \Fn_{1,2}} \ev{Y_j | \Fn_{1,2}} \\
  &\leq n + 2m_n + \sum_{1 \leq i \neq j \leq n} \ev{Y_i | \Fn_{1,2}} \ev{Y_j | \Fn_{1,2}} \\
  &\le n + 2m_n + \ev{Y|\Fn_{1,2}}^2,
\end{align*}
where $(a)$ uses the fact that, conditioned on $\Fn_{1,2}$, the random variables
$Y_i$ and $Y_j$ are conditionally independent for $(i,j) \in \mI$.
From \eqref{eq:dn-r-n23}, we have $|m_n - n\dcm_{1,2}/2| \leq \kappa_2 K n^{1/2}$, 
where $\kappa_2 := (\Delta+1)/2$
and
$\dcm_{1,2} := \deg(\mucm_{1,2}) = \sum_{k=0}^\Delta k r_k$.
As a consequence of the
above discussion, we have $\var(Y|\Fn_{1,2}) \leq \kappa_3 n $ for some
$\kappa_3 > 0$. On the other hand, as we saw above,
$\ev{Y|\Fn_{1,2}} = \ev{Y}$. Therefore, using the law of total variance, we have $\var(Y) \leq \kappa_3 n$. This, together with \eqref{eq:evX-nskl-bound} and
Chebyshev's inequality, implies that the condition $(iv)$ holds with  probability
at least $1 - \kappa_4 n^{-1/3}$, for some $\kappa_4 >0$. Similarly, the same
statement holds for condition $(v)$.   
\end{proof}

Let $\Bn_{1,2}$ be the set of pairs of degree sequences $\vd$ and $\vec{\delta}$ with $n$ elements bounded by $\Delta$ such that for all $0 \leq k, l \leq \Delta$, $|c_{k,l}(\vd, \vec{\delta}) - n \prs{X_1 = k, X -X_1 = l}| \leq n^{2/3}$. 
Moreover, let $\Bn_1$ be the set of $\vd$ such that for some $\vec{\delta}$, we have $(\vd, \vec{\delta}) \in \Bn_{1,2}$. For $\vd \in \Bn_1$, let $\Bn_{2|1}(\vd)$ be the set of degree sequences $\vec{\delta}$ such that $(\vd, \vec{\delta}) \in \Bn_{1,2}$. 

In order to show \eqref{eq:ent-assympt-cm-12}, note that 
since $\Gn_{1,2}$ is formed by adding independent vertex and edge marks to $\Fn_{1,2}$, we have 
\begin{equation*}
  H(\Gn_{1,2}) = \log | \mDn| + \log |\mGn_{\vdn}| + m_n H(\Gamma) + n H(Q).
\end{equation*}
From \eqref{eq:dn-r-n23}, we have $|m_n - n\dcm_{1,2}/2| \leq \frac{(\Delta+1) K}{2} n^{1/2}$. Moreover, we have $\ev{X} > 0$. Consequently, using Lemma~\ref{lem:degree-cm-count} and the fact that $\frac{1}{n} \log |\mDn| \rightarrow H(X)$, we get \eqref{eq:ent-assympt-cm-12}.

We now turn to showing \eqref{eq:ent-assympt-cm-1}. 
Since the expected number of the edges in $\Fn_1$ is $n \dcm_1/2$, we have 
\begin{equation}
  \label{eq:cm-HGn1-1}
  H(\Gn_1)  = H(\Fn_1) + n \frac{\dcm_1}{2} H(\Gamma_1 | \Gamma_1 \neq \circ_1) + n H(Q_1).
\end{equation}
With this, we focus on $H(\Fn_1)$. With $E_n$ being the indicator of
the event that $\Gn_{1,2} \notin \mWn$, we have
\begin{align*}
  H(\Fn_1) &\leq H(\Fn_1, E_n) \leq 1 + H(\Fn_1|E_n) \\
           &= 1+ H(\Fn_1|E_n = 0)\prs{E_n=0} \\
  &\quad + H(\Fn_1|E_n = 1) \prs{E_n=1}.
\end{align*}
  Note that $\Fn_1$ is obtained from $\Fn_{1,2}$ by removing some edges. Hence, we
may write
\begin{equation}
  \label{eq:H-Fn1-En--nlogn-n13}
\begin{aligned}
  H(\Fn_1 | E_n = 1) \leq H(\Fn_1) &\leq \log |\mDn| + \log |\mGn_{\vdn}| + m_n \log 2 \\
                     &\leq H(\Gn_{1,2}) + m_n \log 2 \\
                     &\leq \kappa' n \log n,
\end{aligned}
\end{equation}
where in the last line, $\kappa' > 0$ is obtained from~\eqref{eq:ent-assympt-cm-12}.
Putting this together with  Lemma~\ref{lem:mWn-highprobability}, we have 
\begin{equation}
\label{eq:cm-H-Fn1-E1-on}
  H(\Fn_1|E_n= 1) \prs{E_n=1} \leq \kappa' n \log n \kappa n^{-1/3}.
\end{equation}
Note that the right hand side of \eqref{eq:cm-H-Fn1-E1-on} above is
$o(n)$. On the other hand, by the definition of $\mWn$, if $E = 0$, we have
$\vdg_{\Fn_1} \in \Bn_1$. Therefore,  $  H(\Fn_1 | E_n = 0) \leq \log |\Bn_1| + \max_{\vd \in \Bn_1} \log |\mGn_{\vd}|$.
The assumption $r_0 < 1$ 
together with \eqref{eq:CM-gamma-marginal-assumption}
imply that $\dcm_1 > 0$. 
Additionally note that, by definition, 
for $\vd \in \Bn_1$, we have
$|c_k(\vd) - n \pr{X_1 = k} | \leq n^{2/3}$ for all $0 \leq k \leq \Delta$.
Thereby, we have
\begin{equation*}
  \limsup_{n \rightarrow \infty} \frac{\log |\Bn_1|}{n} \leq H(X_1).
\end{equation*}
Putting the above together with Lemma~\ref{lem:degree-cm-count} and
\eqref{eq:cm-H-Fn1-E1-on}, we have
\begin{equation}
\label{eq:cm-limsup-HFn1}
  \begin{aligned}
    \limsup_{n \rightarrow \infty} \frac{H(\Fn_1) - n \frac{\dcm_1}{2} \log n}{n} &\leq \limsup_{n \rightarrow \infty} \frac{\log |\Bn_1| + \max_{\vd \in \Bn_1} \log |\mGn_{\vd}| - n \frac{\dcm_1}{2} \log n}{n}\\
    &\leq \limsup_{n \rightarrow \infty} \frac{\log |\Bn_1|}{n} + \max_{\vd \in \Bn_1} \frac{\log |\mGn_{\vd}| - \frac{\sum_{i=1}^n d(i)}{2} \log n}{n}  \\
    &\qquad + \max_{\vd \in \Bn_1} \frac{\frac{\sum_{i=1}^n d(i)}{2} \log n  - n \frac{\dcm_1}{2} \log n}{n} \\
    &\leq - s(\dcm_1) +H(X_1) - \ev{\log X_1!},
\end{aligned}
\end{equation}
where in the last line, have used the fact that due to the definition of $\Bn_1$, for $\vd \in
\Bn_1$, we have $|\sum_{i=1}^n d(i) - \dcm_1| \leq \Delta n^{2/3} = o(n / \log n)$.
%
Now, let $\tFn_1$ be the unmarked graph consisting of the edges removed from $\Fn_{1,2}$ to obtain $\Fn_1$, and note that 
\begin{equation}
\label{eq:cm-HF1-lb}
  \begin{aligned}
    H(\Fn_1) &= H(\Fn_1, \tFn_1) - H(\tFn_1 | \Fn_1) \\
    &= H(\Fn_{1,2}) + m_n H(\beta_1) - H(\tFn_1 | \Fn_1).
  \end{aligned}
\end{equation}
Furthermore, conditioned on $E_n = 0$, we have $\vdg_{\tFn_1} \in \Bn_{2|1} (\vdg_{\Fn_1})$. 
Moreover, the assumption \eqref{eq:CM-gamma-marginal-assumption}, 
for $x_1 = \circ_1$, 
together with $r_0 < 1$, 
implies that   $\dcm_{1,2} - \dcm_1 > 0$.
Hence, using a similar method to that used in proving \eqref{eq:cm-limsup-HFn1}, we have 
\begin{equation}
  \label{eq:H-tFn1-Fn1-assymptotic}
\begin{aligned}
  \limsup_{n \rightarrow \infty} \frac{H(\tFn_1|\Fn_1) - n \frac{\dcm_{1,2} - \dcm_1}{2} \log n}{n} \leq -s(\dcm_{1,2} - \dcm_1) \\ + H(X - X_1|X_1) - \ev{\log (X - X_1)!}.
\end{aligned}
\end{equation}
  To see this, with $E_n$ 
  as defined previously, 
  we may write
  \begin{equation*}
    H(\tFn_1 | \Fn_1 ) \leq 1 + H(\tFn_1 | \Fn_1, E_n = 0) \pr{E_n = 0 } + H(\tFn_1 | \Fn_1, E_n = 1) \pr{E_n = 1}.
  \end{equation*}
  Since $\tFn_1$ is obtained from $\Fn_{1,2}$ by removing some edges, similar to
  \eqref{eq:H-Fn1-En--nlogn-n13}, we have $H(\tFn_1|\Fn_1, E_n  = 1) \pr{E_n =
    1} = o(n)$. Moreover, conditioned on $E_n = 0$, we have $\vdg_{\tFn_1} \in
  \Bn_{2|1}(\vdg_{\Fn_1})$. This implies that when $E_n = 0 $, for any realization $f^{(n)}_1$ of $\Fn_1$, we have 
  \begin{equation*}
    H(\tFn_1 | \Fn_1 = f^{(n)}_1, E_n = 0) \leq \log |\Bn_{2|1}(\vdg_{f^{(n)}_1})| + \max_{\vec{\delta} \in \Bn_{2|1}(\vdg_{f^{(n)}_1})} \log |\mGn_{\delta}|.
  \end{equation*}
  Note that, conditioned on $E_n = 0$, we have 
  $\vdg_{f^{(n)}_1} \in \Bn_1$. Hence, we have 
  $\log |\Bn_{2|1}(\vdg_{f^{(n)}_1})| = nH(X - X_1 | X_1) + o(n)$. 
Furthermore, using Lemma~\ref{lem:degree-cm-count} and the fact that for
$\vec{\delta} \in \Bn_{2|1}(\vdg_{\fn_1})$, we have $|\sum_{i=1}^n \delta(i) -
(\dcm_{1,2} - \dcm_1)/2| \leq \Delta n^{2/3} = o(n/\log n)$, we have 
\begin{equation*}
  \max_{\vec{\delta} \in \Bn_{2|1}(\vdg_{f^{(n)}_1})} \log |\mGn_{\delta}| = n (-s(\dcm_{1,2} - \dcm_1) - \ev{\log (X-X_1)!}) + n \frac{\dcm_{1,2} - \dcm_1}{2} \log n + o(n).
\end{equation*}
Putting the above together, we arrive at~\eqref{eq:H-tFn1-Fn1-assymptotic}.

On the other hand, using the definition of $\Fn_{1,2}$, we have $H(\Fn_{1,2}) =
\log |\mDn| + \log |\mGn_{\vdn}|$. Employing Lemma~\ref{lem:degree-cm-count} and
using the assumption~\eqref{eq:dn-r-n23}, we
have
\begin{equation*}
  \lim_{n \rightarrow \infty} \frac{\log |\mGn_{\vdn}| - n \frac{\dcm_{1,2}}{2} \log n}{n} = -s(\dcm_{1,2}) - \ev{\log X!}.
\end{equation*}
Furthermore, we have $\frac{1}{n} \log |\mDn| \rightarrow H(X)$. Therefore, we
have
\begin{equation}
  \label{eq:H-Fn12-assymptotic}
  \lim_{n \rightarrow \infty} \frac{H(\Fn_{1,2}) - n \frac{\dcm_{1,2}}{2} \log n}{n} = -s(\dcm_{1,2}) + H(X) - \ev{\log X!}. 
\end{equation}
Using~\eqref{eq:H-tFn1-Fn1-assymptotic} and~\eqref{eq:H-Fn12-assymptotic} back in~\eqref{eq:cm-HF1-lb}, followed by a simplification using Lemma~\ref{lem:thinning-entropy}, we get 
   \begin{equation}
    \label{eq:H-Fn12-assympt-explained-1}
    \begin{aligned}
      \liminf_{n \rightarrow \infty} \frac{H(\Fn_1) - n \frac{\dcm_1}{2} \log n}{n} & = \liminf_{n \rightarrow \infty} \frac{H(\Fn_{1,2}) - n \frac{\dcm_{1,2}}{2} \log n + m_n H(\beta_1) - H(\tFn_1|\Fn_1) + n \frac{\dcm_{1,2} - \dcm_1}{2} \log n}{n} \\
      &\geq \liminf_{n \rightarrow \infty} \frac{H(\Fn_{1,2}) - n \frac{\dcm_{1,2}}{2} \log n}{n} + \frac{\dcm_{1,2}}{2} H(\beta_1) \\
      &\qquad - \limsup_{n \rightarrow \infty} \frac{H(\tFn_1|\Fn_1) - n \frac{\dcm_{1,2} - \dcm_1}{2} \log n}{n}  \\
      &\geq -s(\dcm_{1,2}) + H(X) - \ev{\log X!} + \frac{\dcm_{1,2}}{2} H(\beta_1)  \\
      &\qquad - \left( -s(\dcm_{1,2} - \dcm_1) + H(X - X_1|X_1) - \ev{\log (X-X_1)!} \right) \\
      &= H(X) + \dcm_{1,2} H(\beta_1) - \ev{\log \binom{X}{X_1}} - \ev{\log X_1!} \\
      &\qquad - \frac{\dcm_{1,2}}{2} H(\beta_1) - s(\dcm_{1,2}) + s(\dcm_{1,2} - \dcm_1) - H(X - X_1 | X_1) \\
      &\stackrel{(a)}{=} H(X_1, X - X_1) - H(X - X_1 | X_1) - \ev{\log X_1!} \\
      &\qquad - s(\dcm_{1,2}) + s(\dcm_{1,2} - \dcm_1) - \frac{\dcm_{1,2}}{2} H(\beta_1) \\
      & = H(X_1) - \ev{\log X_1!} - s(\dcm_{1,2}) + s(\dcm_{1,2} - \dcm_1) - \frac{\dcm_{1,2}}{2} H(\beta_1),
    \end{aligned}
  \end{equation}
  where in $(a)$, we have used Lemma~\ref{lem:thinning-entropy}. 
  Since 
  $\beta_1 = \dcm_1 / \dcm_{1,2}$, we may write
  \begin{align*}
    -s(\dcm_{1,2}) + s(\dcm_{1,2} - \dcm_1) - \frac{\dcm_{1,2}}{2} H(\beta_1) &= \frac{\dcm_{1,2}}{2} \log \dcm_{1,2} - \frac{\dcm_{1,2}}{2} + \frac{\dcm_{1,2}}{2} - \frac{\dcm_1}{2} \\
                                                                          &\qquad - \frac{\dcm_{1,2}}{2} \log (\dcm_{1,2}- \dcm_1) + \frac{\dcm_1}{2} \log (\dcm_{1,2} - \dcm_1) \\
                                                                          &\qquad + \frac{\dcm_1}{2} \log \dcm_1 - \frac{\dcm_1}{2} \log \dcm_{1,2} + \frac{\dcm_{1,2}}{2} \log (\dcm_{1,2} - \dcm_1) \\
                                                                          &\qquad - \frac{\dcm_1}{2} \log (\dcm_{1,2} - \dcm_1) - \frac{\dcm_{1,2}}{2} \log \dcm_{1,2} + \frac{\dcm_1}{2} \log \dcm_{1,2} \\
                                                                          &= -\frac{\dcm_1}{2} + \frac{\dcm_1}{2} \log \dcm_1 \\
    &= -s(\dcm_1).
  \end{align*}
  Substituting this into~\eqref{eq:H-Fn12-assympt-explained-1}, we arrive at 
\begin{equation}
\label{eq:cm-liminf-HFn1}
\begin{gathered}
  \liminf \frac{H(\Fn_1) - n \frac{\dcm_1}{2} \log n}{n} \geq - s(\dcm_1) +H(X_1) - \ev{\log X_1!}.
\end{gathered}
\end{equation}
This, together with \eqref{eq:cm-limsup-HFn1} and \eqref{eq:cm-HGn1-1}, completes the proof of \eqref{eq:ent-assympt-cm-1}. The proof of \eqref{eq:ent-assympt-cm-2} is similar.

\section{Bounding $|\Sn_2(\Hn_1)|$ for $\Hn_{1,2} \in \mGn_{1,2}$ in the \ER case }
\label{sec:bounding-s_2g_1-er}

Note that for $\Hn_{1,2} \in \mGn_{1,2}$ 
and $\Gn_2 \in \mGn_2$, if $\Hn_1 \oplus \Gn_2 \in \mGn_{\vp, \vq}$, we have $\vm_{\Hn_1 \oplus \Gn_2} \in \mMn$ and $\vu_{\Hn_1 \oplus \Gn_2} \in \mUn$. On the other hand, for fixed $\vm \in \mMn$ and $\vu \in \mUn$,
 the number of $\Gn_2$ such that $\vm_{\Hn_1 \oplus \Gn_2} = \vm$ and $\vu_{\Hn_1 \oplus \Gn_2} = \vu$ is at most 
{\small
\begin{align*}
  A_2(\vm, \vu):= &\left ( \prod_{x_1 \in \Xi_1} \binom{m(x_1)}{\{ m(x_1, x_2)\}_{x_2 \in \Xi_2 \cup \{\circ_2 \}}} \right ) \times \binom{\binom{n}{2} - \sum_{x_1 \in \Xi_1} m(x_1)}{\{ m(\circ_1, x_2) \}_{x_2 \in \Xi_2}} \times \left ( \prod_{\theta_1 \in \Theta_1} \binom{u(\theta_1)}{\{ u(\theta_1, \theta_2) \}_{\theta_2 \in \Theta_2} } \right ),
\end{align*}
}%
where we have used the notational conventions in \eqref{eq:mx1} and \eqref{eq:utheta1}.
Consequently, we have 
\begin{equation}
\label{eq:max-S2-max-A2}
\begin{split}
  \max_{\Hn_{1,2} \in \mGn_{\vp, \vq}} |\Sn_2(\Hn_1)| &\leq |\mMn||\mUn| \max_{\stackrel{\vm \in \mMn}{\vu \in \mUn}} A_2(\vm, \vu) \\
&\leq (2 n^{2/3}+1)^{(|\Xi_{1,2}|+|\Theta_{1,2}|)} \max_{\stackrel{\vm \in \mMn}{\vu \in \mUn}} A_2(\vm, \vu).
\end{split}
\end{equation}
Now, if $\vmn$ and $\vun$ are sequences in $\mMn$ and $\mUn$, respectively, then for all $x \in \Xi_{1,2}$ we have $\mn(x) / n \rightarrow p_{x}/2$. Furthermore, 
for all $x_1 \in \Xi_1$ and $\theta_1 \in \Theta_1$, we have $\mn(x_1)/n \rightarrow p_{x_1}/2$ and $\un(\theta_1) / n \rightarrow q_{\theta_1}$. As a result, using Lemma~\ref{lem:binom-assymp}, for any such sequences $\vmn$ and $\vun$, with $Q = (Q_1, Q_2)$ having law $\vq$, we have 
\begin{align*}
  &\lim_{n \rightarrow \infty} \frac{\log A_2(\vmn, \vun) - (\sum_{x_2 \in \Xi_2} \mn(\circ_1, x_2)) \log n}{n}\\
  &\qquad\qquad \qquad = \sum_{x_2 \in \Xi_2} s(p_{\circ_1, x_2}) + \sum_{x_1 \in \Xi_1} \frac{p_{x_1}}{2} H\left ( \left \{ \frac{p_{(x_1,x_2)}}{p_{x_1}} \right \}_{x_2 \in \Xi_2 \cup \{ \circ_2 \}} \right ) \\
  &\qquad\qquad\qquad \quad  + \sum_{\theta_1 \in \Theta_1} q_{\theta_1} H \left ( \left \{ \frac{q_{\theta_1, \theta_2}}{q_{\theta_1}} \right\}_{\theta_2 \in \Theta_2} \right ) \\
  &\qquad\qquad\qquad= H(Q_2 |Q_1) + \sum_{x \in \Xi_{1,2}} s(p_{x})- \sum_{x_1 \in \Xi_1} s(p_{x_1}) \\
  &\qquad\qquad\qquad= \bch(\muer_2 | \muer_1),
\end{align*}
where the second equality follows by rearranging the terms and using the definition of $s(.)$.  
Using the fact that $|\mn(\circ_1, x_2) - n p_{\circ_1, x_2}/2| \leq n^{2/3}$, we have
\begin{equation*}
  \lim_{n \rightarrow \infty} \frac{\log A_2(\vmn, \vun) -n\frac{\der_{1,2}-\der_1}{2} \log n}{n} = \bch(\muer_2 | \muer_1).
\end{equation*}
This together with \eqref{eq:max-S2-max-A2} implies \eqref{eq:limsup-S2G1-conditional-BC}.

\section{Bounding $|\Sn_2(\Hn_1)|$ for 
$\Hn_{1,2} \in \mWn$ in the configuration model }
\label{sec:bound-s_2g_1-conf}

Here, we find an upper bound for $\max_{\Hn_{1,2} \in \mWn}|\Sn_2(\Hn_1)|$, where $\mWn$ is defined in
Section~\ref{sec:proof-achievability-conf}, and
use it to show \eqref{eq:conf-S2G1-bch-2|1}. Take $\Hn_{1,2} \in \mWn$ and
assume $\hat{H}^{(n)}_2 \in \Sn_2(\Hn_1)$. With $\hat{H}^{(n)}_{1,2} := \Hn_1 \oplus \hat{H}^{(n)}_2$, let
$\tHn_2$ be the subgraph of $\hat{H}^{(n)}_{1,2}$ consisting of the edges not present in
$\Hn_1$. Employing the notation of Appendix~\ref{sec:asympt-behav-enropy-cm}, 
we have $\vdg_{\tHn_2} \in
\Bn_{2|1}(\vdg_{\Hn_1})$, which follows from
the definition of the set $\mWn$. Therefore, we can think of $\hat{H}^{(n)}_{1,2}$ as
being constructed
from $\Hn_1$ by adding a graph to $\Hn_1$ with degree sequence $\vdg_{\tHn_2}
\in \Bn_{2|1}(\vdg_{\Hn_1})$,
marking its edges,  adding second domain marks to edges in $\Hn_1$, and  also
adding second domain marks to vertices. Motivated by this, we have 
{\small
\begin{equation}
  \label{eq:conf-max-S2(G1)-1}
\begin{aligned}
  &\max_{\Hn_{1,2} \in \mWn} \log |\Sn_2(\Hn_1)| \leq \max_{\Hn_{1,2} \in \mWn} \log |\Bn_{2|1}(\vdg_{\Hn_1})|  + \max_{\Hn_{1,2} \in \mWn, \vec{\delta} \in \Bn_{2|1}(\vdg_{\Hn_1})} \log |\mGn_{\vec{\delta}}| \\
  &\qquad + \max_{\vm \in \mMn} \log \binom{m_n - \sum_{x_1 \in \Xi_1} m(x_1)}{\{m((\circ_1, x_2))\}_{x_2 \in \Xi_2}} \prod_{x_1 \in \Xi_1} \binom{m(x_1)}{\{m((x_1, x_2))\}_{x_2 \in \Xi_2}} \\
  &\qquad + \max_{\vu \in \mUn} \log \prod_{\theta_1 \in \Theta_1} \binom{u(\theta_1)}{\{u((\theta_1, \theta_2))\}_{\theta_2 \in \Theta_2}}.
\end{aligned}
\end{equation}
}%
We establish an upper bound for each term.
The definition of $\Bn_{2|1}$ implies that 
\begin{equation}
  \label{eq:app-S-conf-B-bound}
  \lim_{n \rightarrow \infty} \frac{1}{n} \max_{\Hn_{1,2} \in \mWn} \log |\Bn_{2|1}(\vdg_{\Hn_1})| = H(X - X_1 | X_1),
\end{equation}
where $(X,X_1)$ are defined as in
Section~\ref{sec:proof-achievability-conf}.
Note that  the assumption \eqref{eq:CM-gamma-marginal-assumption}, 
for $x_1 = \circ_1$, 
together with
$r_0 < 1$, implies that $\dcm_{1,2} - \dcm_1 >0$.
On the other hand, we have
\begin{equation}
  \label{eq:app-S-conf-Gndelta-1}
  \begin{aligned}
    &\limsup_{n \rightarrow \infty} \max_{\stackrel{\Hn_{1,2} \in \mWn}{\vec{\delta} \in B^{(n)}_{2|1}(\vdg_{\Hn_1})}} \frac{\log |\mGn_{\vec{\delta}}|  - n\frac{\dcm_{1,2} - \dcm_1}{2} \log n}{n} \leq \limsup_{n \rightarrow \infty} \max_{\stackrel{\Hn_{1,2} \in \mWn}{\vec{\delta} \in B^{(n)}_{2|1}(\vdg_{\Hn_1})}} \frac{\log |\mGn_{\vec{\delta}}|  - \frac{\sum_{i=1}^n \delta_i}{2} \log n}{n} \\
    &\quad + \limsup_{n \rightarrow \infty} \max_{\stackrel{\Hn_{1,2} \in \mWn}{\vec{\delta} \in B^{(n)}_{2|1}(\vdg_{\Hn_1})}} \frac{1}{n} \left( \frac{\sum_{i=1}^n \delta_i}{2} \log n - n \frac{\dcm_{1,2} - \dcm_1}{2} \log n \right).
    \end{aligned}
\end{equation}
By definition, for $\vec{\delta} = (\delta_1, \dots, \delta_n) \in
\Bn_{2|1}(\vdg_{\Hn_1})$, we have
\begin{align*}
  \left|\left( \sum_{i=1}^n \delta_i \right) - n(\dcm_{1,2} - \dcm_1) \right| &= \left|\left( \sum_{k=0}^\Delta k c_k(\vec{\delta})  \right) - n \ev{X - X_1} \right| \\
                                                    &\leq \sum_{k=0}^\Delta k |c_k(\vec{\delta}) - n \pr{X - X_1 = k} | \\
                                                    &\leq \sum_{k=0}^\Delta k \sum_{j=0}^\Delta |c_{j,k}(\vdg_{\Hn_1}, \vec{\delta}) - n \pr{X_1 = j, X - X_1 = k} | \\
                                                    &\leq \Delta^3 n^{2/3}.
\end{align*}
This implies that the second term in the right hand side of
\eqref{eq:app-S-conf-Gndelta-1} vanishes.
Therefore,  Lemma~\ref{lem:degree-cm-count}
implies that 
\begin{equation}
  \label{eq:app-S-conf-Gndelta-bound}
\begin{aligned}
  \limsup_{n \rightarrow \infty} \max_{\stackrel{\Hn_{1,2} \in \mWn}{\vec{\delta} \in B^{(n)}_{2|1}(\vdg_{\Hn_1})}} \frac{\log |\mGn_{\vec{\delta}}|  - n\frac{\dcm_{1,2} - \dcm_1}{2} \log n}{n} \leq -s(\dcm_{1,2} - \dcm_1) - \ev{\log (X-X_1)!}.
\end{aligned}
\end{equation}
Furthermore, if $\vmn$ is a sequence in
$\mMn$,  by definition we have $\sum_{x \in \edgemark_{1,2}} |\mn(x) - m_n
\gamma_x| \leq n^{2/3}$. Therefore, we have
\begin{equation*}
  \lim_{n \rightarrow \infty} \frac{m_n - \sum_{x_1 \in \edgemark_1} \mn(x_1)}{n} = \frac{\dcm_{1,2}}{2} \left( 1 - \sum_{x_1 \in \edgemark_1} \gamma_{x_1} \right),
\end{equation*}
where $\gamma_{x_1}$ for $x_1 \in \edgemark_1$ is defined to be $\sum_{x_2 \in
  \edgemark_2 \cup \{\circ_2\}} \gamma_{(x_1, x_2)}$. Similarly, for $x_2 \in
\edgemark_2$, we have
\begin{equation*}
  \lim_{n \rightarrow \infty} \frac{\mn((\circ_1, x_2))}{m_n - \sum_{x_1 \in \edgemark_1} \mn(x_1)} = \frac{\gamma_{(\circ_1, x_2)}}{1- \sum_{x_1 \in \edgemark_1} \gamma_{x_1}} = \frac{\gamma_{(\circ_1, x_2)}}{\sum_{x'_2 \in \edgemark_2} \gamma_{(\circ_1, x'_2)}}.
\end{equation*}
Consequently, using Lemma~\ref{lem:binom-assymp}, we have
\begin{equation}
  \label{eq:app-conf-S-mn-1}
  \begin{aligned}
    \lim_{n \rightarrow \infty} \frac{1}{n} \log \binom{m_n - \sum_{x_1 \in \edgemark_1} \mn(x_1)}{\{ \mn((\circ_1, x_2))\}_{x_2 \in \edgemark_2}} &= \frac{\dcm_{1,2}}{2} \left(1 - \sum_{x_1 \in \edgemark_1 } \gamma_{x_1}\right) H\left( \left\{ \frac{\gamma_{(\circ_1, x_2)}}{\sum_{x'_2 \in \edgemark_2} \gamma_{(\circ_1, x'_2)}} \right\}_{x_2 \in \edgemark_2} \right) \\
    &= \frac{\dcm_{1,2}}{2} \pr{\Gamma_1 = \circ_1} H(\Gamma_2 | \Gamma_1 = \circ_1).
  \end{aligned}
\end{equation}
Here, $\Gamma = (\Gamma_1, \Gamma_2)$ has law $\vgamma$.
On the other hand, for $x_1 \in \edgemark_1$ and $x_2 \in \edgemark_2$, we have $\mn(x_1)/ n \rightarrow
\frac{\dcm_{1,2}}{2} \gamma_{x_1}$, and $\frac{\mn((x_1, x_2))}{\mn(x_1)} \rightarrow
\frac{\gamma_{(x_1, x_2)}}{\gamma_{x_1}}$.  Consequently, 
another use of 
Lemma~\ref{lem:binom-assymp} implies that for all $x_1 \in \edgemark_1$, we have
\begin{equation}
  \label{eq:app-conf-S-mn-2}
  \begin{aligned}
    \lim_{n \rightarrow \infty} \frac{1}{n} \log \binom{\mn(x_1)}{\{\mn((x_1, x_2))\}_{x_2 \in \edgemark_2}} &= \frac{\dcm_{1,2}}{2} \gamma_{x_1} H\left( \left\{ \frac{\gamma_{(x_1, x_2)}}{\gamma_{x_1}} \right\}_{x_2 \in \edgemark_2} \right) \\
    &= \frac{\dcm_{1,2}}{2} \pr{\Gamma_1 = x_1} H(\Gamma_2 |\Gamma_1 = x_1). 
  \end{aligned}
\end{equation}
Putting together~\eqref{eq:app-conf-S-mn-1} and \eqref{eq:app-conf-S-mn-2}, we
realize that 
\begin{equation}
  \label{eq:app-conf-S-mn-together}
  \begin{aligned}
    &\lim_{n \rightarrow \infty} \max_{\vm \in \mMn} \log \binom{m_n - \sum_{x_1 \in \Xi_1} m(x_1)}{\{m((\circ_1, x_2))\}_{x_2 \in \Xi_2}} \prod_{x_1 \in \Xi_1} \binom{m(x_1)}{\{m((x_1, x_2))\}_{x_2 \in \Xi_2}} = \frac{\dcm_{1,2}}{2} \Bigg ( \pr{\Gamma_1 = \circ_1}H(\Gamma_2 | \Gamma_1 = \circ_1) \\
    &\quad \qquad \qquad + \sum_{x_1 \in \edgemark_1} \pr{\Gamma_1 = x_1} H(\Gamma_2 | \Gamma_1 = x_1) \Bigg ) \\
    &\qquad \qquad   = \frac{\dcm_{1,2}}{2} H(\Gamma_2 | \Gamma_1).
    \end{aligned}
  \end{equation}

  Using a similar technique, if $\vun$ is a sequence in $\mUn$, for all
  $\theta_1 \in \vermark_1$ and $\theta_2 \in \vermark_2$, we have
  $\frac{\un(\theta_1)}{n} \rightarrow q_{\theta_1}$ and $\frac{\un((\theta_1,
    \theta_2))}{\un(\theta_1)} \rightarrow \frac{q_{(\theta_1,
      \theta_2)}}{q_{\theta_1}}$. Thereby, using Lemma~\ref{lem:binom-assymp},
  for all $\theta_1 \in \vermark_1$, we have
  \begin{equation*}
    \lim_{n \rightarrow \infty} \frac{1}{n} \log \binom{\un(\theta_1)}{\{\un((\theta_1, \theta_2))\}_{\theta_2 \in \vermark_2\}}} = q_{\theta_1} H\left( \left\{ \frac{q_{(\theta_1,\theta_2)}}{q_{\theta_1}} \right\}_{\theta_2 \in \vermark_2} \right) = \pr{Q_1 = \theta_1} H(Q_2 | Q_1 = \theta_1),
  \end{equation*}
  where $Q = (Q_1, Q_2)$ has law $\vq$. Consequently, we have
  \begin{equation}
    \label{eq:app-conf-S-un}
    \lim_{n \rightarrow \infty} \frac{1}{n} \max_{\vu \in \mUn} \log \prod_{\theta_1 \in \Theta_1} \binom{u(\theta_1)}{\{u((\theta_1, \theta_2))\}_{\theta_2 \in \Theta_2}} = \sum_{\theta_1 \in \vermark_1} \pr{Q_1 = \theta_1} H(Q_2 | Q_1 = \theta_1) = H(Q_2 | Q_1).
  \end{equation}
Putting \eqref{eq:app-S-conf-B-bound}, \eqref{eq:app-S-conf-Gndelta-bound},
\eqref{eq:app-conf-S-mn-together}, and \eqref{eq:app-conf-S-un} back into
\eqref{eq:conf-max-S2(G1)-1}, we get 
\begin{align*}
  \lim_{n \rightarrow \infty} \frac{\max_{\Hn_{1,2} \in \mWn} \log |\Sn_2(\Hn_1)| - n \frac{\dcm_{1,2} - \dcm_1}{2} \log n}{n} 
  = -s(\dcm_{1,2} - \dcm_1) + H(X-X_1| X_1) \\ - \ev{\log (X-X_1)!} 
+ \frac{\dcm_{1,2}}{2} H(\Gamma_2 | \Gamma_1) + H(Q_2 | Q_1).
\end{align*}
Using Lemma~\ref{lem:thinning-entropy} and rearranging, this is precisely equal to $\bch(\mucm_2 | \mucm_1)$, which completes the proof of \eqref{eq:conf-S2G1-bch-2|1}.

\section{Proof of Theorem~\ref{thm:graph-SW-k-source}: generalization to
  multiple sources}
\label{sec:app-multisource-proof}

The proof of Theorem~\ref{thm:graph-SW-k-source} is similar to that of
Theorem~\ref{thm:SW} which was given in Section~\ref{sec:main-results}. 

It is easy to verify that if $\Gn_{[k]}$ is distributed according to 
either the multi-source \ER ensembles or the multi-source 
configuration model ensembles discussed in 
Section~\ref{sec:gen-more-sources}, 
then given any 
nonempty $A \subset [k]$, $A \neq [k]$, 
the joint distribution of $(\Gn_A, \Gn_{A^c})$ is
similar to that of a two--source ensemble as in Section~\ref{sec:prel-notat}
with the following mark sets: 
   \begin{align*}
    &\tilde{\edgemark}_1 := \left\{ x_A \in \edgemark_A : \sum_{(x'_j : j \in [k]) : (x'_j : j \in A) = x_A} p_{(x'_j : j \in [k])} > 0 \right\} \\
     &\tilde{\edgemark}_2 :=  \left\{ x_{A^c} \in \edgemark_{A^c} : \sum_{(x'_j : j \in [k]) : (x'_j : j \in A^c) = x_{A^c}} p_{(x'_j : j \in [k])} > 0 \right\} \\
     &\tilde{\edgemark}_{1,2} := \edgemark_{[k]} \\
    &\tilde{\vermark}_1 := \left\{ \theta_A \in \vermark_A : \sum_{(\theta'_j : j \in [k]) : (\theta'_j : j \in A) = \theta_A} q_{(\theta'_j : j \in [k])} > 0 \right\} \\
     &\tilde{\vermark}_2 :=  \left\{ \theta_{A^c} \in \vermark_{A^c} : \sum_{(\theta'_j : j \in [k]) : (\theta'_j : j \in A^c) = \theta_{A^c}} q_{(\theta'_j : j \in [k])} > 0 \right\} \\
     &\tilde{\vermark}_{1,2} := \vermark_{[k]}
  \end{align*}
Moreover, we set $\tilde{\circ}_1 := \circ_A$ and $\tilde{\circ}_2 :=
\circ_{A^c}$. To establish the analogy, for the \ER ensemble, we define $
\vec{\tilde{p}} = \{ \tilde{p}_x\}_{x \in \tilde{\edgemark}_{1,2}}$ such that
$\tilde{p}_x = p_x$ for $x \in \tilde{\edgemark}_{1,2}$. Furthermore, we define $
\vec{\tilde{q}} = \{ \tilde{q}_\theta\}_{\theta \in \tilde{\vermark}_{1,2}}$ 
such that $\tilde{q}_\theta = q_\theta$ for $\theta \in \tilde{\vermark}_{1,2}$.
Similarly, for the configuration model ensemble, we let $\vec{\tilde{\gamma}} =
\{ \tilde{\gamma}_x\}_{x \in \tilde{\edgemark}_{1,2}}$ such that
$\tilde{\gamma}_x = \gamma_x$ for $x \in \tilde{\edgemark}_{1,2}$, and define
$\vec{\tilde{q}} = \{ \tilde{q}_\theta\}_{\theta \in \tilde{\vermark}_{1,2}}$
where $\tilde{q}_\theta = q_\theta$ for $\theta \in \tilde{\vermark}_{1,2}$.
It can be  easily verified that \eqref{eq:ER-p-marginal-assumption} and
  \eqref{eq:ER-q-marginal-assumption} follow from the assumptions
  \eqref{eq:gen-er-p-positive}
  and \eqref{eq:gen-er-q-positive}. 
Likewise, \eqref{eq:CM-gamma-marginal-assumption} and
  \eqref{eq:CM-q-marginal-assumption} follow from
  \eqref{eq:general-conf-model-gamma-positive}, \eqref{eq:general-conf-model-gamma-matters} and
  \eqref{eq:general-conf-model-q-positive}. 

Using this observation together with
\eqref{eq:ent-assympt-er-12}--\eqref{eq:ent-assympt-er-2}, we realize that for
the 
multi-source 
\ER ensemble and nonempty $A \subseteq [k]$, we have
\begin{equation}
  \label{eq:ent-assymp-multi-ER}
  H(\Gn_A) = \frac{\der_A}{2} n \log n + n \left( H(Q_A) + \sum_{x \in \edgemark_A} s(p_x) \right) + o(n),
\end{equation}
where $\der_A := \deg(\muer_A)$, and with  $Q = (Q_i: i \in [k])$ having law
$\vq$, we let $Q_A := (Q_i: i \in A)$. In fact, the coefficient of $n$ in the above
expression is $\bch(\muer_A)$. Similarly, the above observation together
with \eqref{eq:ent-assympt-cm-12}--\eqref{eq:ent-assympt-cm-2} establishes that
for the multi-source configuration model ensemble and 
for nonempty $A \subseteq [k]$,
  \begin{align*}
    H(\Gn_A) &= \frac{\dcm_A}{2} n \log n + n \Big ( - s(\dcm_A) + H(X_A) - \ev{\log X_A!} \\
             &\quad + H(Q_A) + \frac{\dcm_A}{2} H(\Gamma_A|\Gamma_A \neq \circ_A) \Big ) + o(n)
  \end{align*}
where $\dcm_A := \deg(\mucm_A)$. In the above expression, with $X \sim \vr$ and $\Gamma^i =
(\Gamma^i_j: j \in [k])$ for $1 \leq i \leq \Delta$ which are i.i.d.\ with law $\vgamma$,
we define $X_A := \sum_{i=1}^X \one{\Gamma^i_j \neq \circ_j \text{ for some } j \in A}$.
Here, if $X=0$, then $X_A := 0$.
Moreover, $Q = (Q_i: i \in [k])$ has law $\vq$ and $Q_A := (Q_i: i \in A)$.
Furthermore, $\Gamma = (\Gamma_i: i \in [k])$ has law $\vgamma$ and $\Gamma_A :=
(\Gamma_i: i \in A)$. It can be seen that the coefficient of $n$ in the above
expression is $\bch(\mucm_A)$.

\subsection{Proof of converse}
\label{sec:multi-converse}

Observe that for both the  \ER and the configuration model ensembles,  for $A \subset
[k]$ nonempty, $A \neq [k]$, even if all the encoders in
  the set $A$ as well as all the encoders in the set $A^c$ can cooperate, since the
  distribution of $(\Gn_A, \Gn_{A^c})$ is identical to a two--source ensemble as
  was discussed above,
  using the converse result corresponding to the two--source case (i.e.\
  Sections~\ref{sec:proof-converse-ER} and \ref{sec:proof-converse-CM}),  with
  $\alpha_B := \sum_{i \in B} \alpha_i$ and $R_B := \sum_{i \in B} R_i$ for $B
  \subset [k]$, for $((\alpha_i, R_i): i \in [k]) \in \mR$, we must have
  \begin{align*}
    (\alpha_A, R_A) &\succeq ((d_{[k]} - d_{A^c})/2, \bch(\mu_A | \mu_{A^c}))\\
    (\alpha_{A^c}, R_{A^c}) &\succeq ((d_{[k]} - d_{A})/2, \bch(\mu_{A^c} | \mu_{A}))  \\
    (\alpha_{[k]}, R_{[k]}) & \succeq (d_{[k]}/2, \bch(\mu_{[k]})).
  \end{align*}
  Here, $\mu$ denotes $\muer$ or $\mucm$, 
  depending on the ensemble. Repeating
  this for all nonempty $A \subset [k]$, $A \neq [k]$, recovers all the
  necessary inequalities and completes the converse proof.

\subsection{Proof of achievability for the \ER ensemble}
\label{sec:multi-achieve-er}

Similar to Section~\ref{sec:proof-achievability-ER}, we employ a random binning
codebook construction with $\Ln_i = \lfloor  \exp (\alpha_i n \log n + R_i n) \rfloor$  for $i
\in [k]$. More precisely, For $i \in [k]$ and $\Hn_i \in \mGn_i$, we generate $\fn_i(\Hn_i)$
uniformly in $[\Ln_i]$. The choice of $\fn_i(\Hn_i)$ 
is made independently for each 
$\Hn_i \in \mGn_i$ and also for each domain $i \in [k]$. To explain the decoding
procedure, similar to Section~\ref{sec:proof-achievability-ER}, let $\mMn$ be the set of $\vm = \{ m(x)\}_{x \in \edgemark_{[k]}}$ such
that $\sum_{x \in \edgemark_{[k]}} |m(x) - n p_x / 2|  \leq n^{2/3}$.
Furthermore, let $\mUn$ be the set of $\vu = \{u(\theta)\}_{\theta \in \vermark_{[k]}}$
such that $\sum_{\theta \in \vermark_{[k]}} |u(\theta) - n q_\theta| \leq
n^{2/3}$. With these, let $\mGn_{\vp, \vq}$ be the set of $\Hn_{[k]} \in \mGn_{[k]}$ such that
$\vm_{\Hn_{[k]}} \in \mMn$ and $\vu_{\Hn_{[k]}} \in \mUn$. At the receiver,
upon receiving bin indices $i_j, 1 \leq j \leq k$, we form the set of
$\Hn_{[k]} \in \mGn_{\vp, \vq}$ such that  $\fn_j(\Hn_j) = i_j$ for $j \in [k]$. If
there is only one graph in this set, the decoder outputs that graph;  otherwise, it
reports an error. It can be easily seen that the error events are as follows:
\begin{equation*}
  \mEn_1 = \{ \Gn_{[k]} \notin \mGn_{\vp, \vq}\},
\end{equation*}
and, for each nonempty $A \subset [k]$, 
\begin{align*}
  \mEn_A = \{ \exists \Hn_{[k]} \in \mGn_{\vp, \vq}:\, \, &\Hn_i = \Gn_i \text{ for } i \notin A, \\
  & \Hn_i \neq \Gn_i, \fn_i(\Hn_i) = \fn_i(\Gn_i) \text{ for } i \in A\}.
\end{align*}
For nonempty $A \subset [k]$ and $\Hn_{A} \in \mGn_A$, we denote
$(\fn_i(\Hn_i): i \in A)$ by $\fn_A(\Hn_A)$. Note that we may treat $\fn_A(\Hn_A)$ as an integer in
  the range $\prod_{i \in A} \Ln_i \approx \lfloor  \exp(\alpha_A n \log n + R_A
  n) \rfloor$.
Recall that $\alpha_A = \sum_{i \in A} \alpha_i$ and $R_A =
  \sum_{i \in A} R_i$. Observe that due to our random binning procedure,
  $\fn_A(\Hn_A)$ is uniformly distributed in the range $\prod_{i \in A} \Ln_i$.
  Moreover, for $\Hn_{[k]}$ such that $\Hn_i \neq \Gn_i$ for $i \in A$,
  $\fn_A(\Hn_A)$ is independent from $\fn_A(\Gn_A)$. Thereby, for nonempty $A \subset
  [k]$, $A \neq [k]$, using the
  previously discussed fact
  that $(\Gn_A, \Gn_{A^c})$ is distributed according to a two--source ensemble,
  and using the analysis of Section~\ref{sec:proof-achievability-ER}, we realize
  that the probabilities of the error events $\mEn_A$, $\mEn_{A^c}$,
  $\mEn_{[k]}$, and $\mEn_1$ vanish as $n \rightarrow
  \infty$ given that $(\alpha_A, R_A)  \succeq ((d_{[k]} - d_{A^c})/2,
  \bch(\muer_A | \muer_{A^c}))$,
  $(\alpha_{A^c}, R_{A^c}) \succeq ((d_{[k]} - d_{A})/2, \bch(\muer_{A^c} | \muer_{A}))$, and $(\alpha_{[k]}, R_{[k]})
  \succeq (d_{[k]} /2, \bch(\muer_{[k]}))$. Repeating this argument for all nonempty $A
\subset [k]$, $A \neq [k]$, 
we realize that 
the probabilities of all error
events vanish, which completes the proof of achievability.

\subsection{Proof of achievability for the configuration model ensemble}
\label{sec:multi-achieve-conf}
  
We again employ a random binning procedure as in the above, 
where, for $i \in [k]$ and $\Hn_i \in \mGn_i$, 
we choose $\fn_i(\Hn_i)$ uniformly in the set 
$[\Ln_i]$ with $\Ln_i = \lfloor \exp(\alpha_i n \log n + R_i n) \rfloor$. To explain the decoding procedure, similar to the setup in
Section~\ref{sec:proof-achievability-conf}, we define $\mDn$ be the set of
degree sequences $\vd$
such that $c_i(\vd)  = c_i(\vdn)$ for all $0 \leq i \leq \Delta$. Moreover, let
$\mMn$ be the set of $\vm = (m(x): x \in \edgemark_{[k]})$ such that 
$\sum_{x \in \edgemark_{[k]}}  m(x) = m_n$,
where $m_n := (\sum_{i=1}^n \dn(i))/2$, 
and $\sum_{x \in \edgemark_{[k]}} |m(x) - m_n \gamma_x| \leq  n^{2/3}$. Also,
let $\mUn$ be the set of $\vu = (u(\theta): \theta \in \vermark_{[k]})$ such that $\sum_{\theta \in \vermark_{[k]}}
  |u(\theta) - nq_\theta| \leq n^{2/3}$.
Let the random variables $X$ and $X_A$ for $A \subset
[k]$ nonempty be defined as above, i.e.\ $X \sim \vr$ and with $\Gamma^i =
(\Gamma^i_j: j \in [k])$ for $1 \leq i \leq \Delta$ being i.i.d.\ with law
$\vgamma$, we define $X_A := \sum_{i=1}^X \one{\Gamma^i_j \neq \circ_j \text{ for
  some } j \in A}$ if $X > 0$, and $X_A := 0$ if $X = 0$. With this, 
 let $\mWn$ be the set of $\Hn_{[k]} \in \mGn_{[k]}$ such that  $(i)$
  $\vdg_{\Hn_{[k]}} \in \mDn$, $(ii)$ $\vm_{\Hn_{[k]}} \in \mMn$, $(iii)$
  $\vu_{\Hn_{[k]}} \in \mUn$, and $(iv)$ for all $A \subset [k]$ nonempty and $0
  \leq j \leq i \leq \Delta$, we have 
  \begin{equation*}
    |c_{i,j}(\vdg_{\Hn_{[k]}}, \vdg_{\Hn_A}) - n \pr{X = i, X_A = j} | \leq n^{2/3}.
  \end{equation*}
At the decoder, upon receiving $i_j: 1 \leq j \leq k$, we form the set of graphs $\Hn_{[k]} \in \mWn$  such that
  $\fn(\Hn_j) = i_j$ for $1 \leq j \leq k$. If there is only one graph in this
  set, the decoder outputs this graph; otherwise, it reports an error. It can be
  easily seen that the error events are as follows:
  \begin{equation*}
    \mEn_1 = \{\mGn_{[k]} \notin \mWn \},
  \end{equation*}
and for nonempty $A \subset [k]$, 
\begin{align*}
  \mEn_A = \{ \exists \Hn_{[k]} \in \mWn:\, \, & \Hn_i = \Gn_i \text{ for } i \notin A \\
& \Hn_i \neq \Gn_i, \fn(\Hn_i) = \fn(\Gn_i) \text{ for } i \in A\}.
\end{align*}  
Similar to the above discussion in Section~\ref{sec:multi-achieve-er}, since for
$A \subset [k]$ nonempty, $A\neq [k]$, the distribution of $(\Gn_A, \Gn_{A^c})$
is identical to a two--source configuration model ensemble, using the analysis
in Section~\ref{sec:proof-achievability-conf}, we realize that the probabilities
of the error events $\mEn_A$, $\mEn_{A^c}$,
  $\mEn_{[k]}$, and $\mEn_1$ vanish as $n \rightarrow
  \infty$ given that $(\alpha_A, R_A)  \succeq ((d_{[k]} - d_{A^c})/2,
  \bch(\mucm_A | \mucm_{A^c}))$,
  $(\alpha_{A^c}, R_{A^c}) \succeq ((d_{[k]} - d_{A})/2, \bch(\mucm_{A^c} | \mucm_{A}))$, and $(\alpha_{[k]}, R_{[k]})
  \succeq (d_{[k]} /2, \bch(\mucm_{[k]}))$. Repeating this argument for all nonempty $A
\subset [k]$, $A \neq [k]$, we realized that the probabilities of all error
events vanish, which completes the proof of achievability. 

\editfinish

{\color{newcolor}
\section{Some Examples of Calculating the  marked  BC Entropy}
\label{app:bc-ent-calc-examples}

\editstart 
In this appendix, we provide some examples of calculating the marked BC
entropy defined in Section~\ref{sec:bc-entropy}. To simplify the discussion, we
focus on a special yet rich class of probability distributions on
$\mG_*(\edgemark, \vermark)$ which we call \emph{depth-1 unimodular Galton-Watson trees}
defined as follows. The reader is referred to \cite{delgosha2019notion} for
more details for the 
general setting.

We first need to make some definitions. As in Section~\ref{sec:prel-notat}, let $\edgemark$ and $\vermark$ be finite
sets of edge marks and vertex marks, respectively. Given a marked graph $G$ and  two adjacent vertices $v$ and $w$ in $G$, let $\xi_G(v, w)
= \xi_G(w,v) \in \edgemark$ be the mark on the edge connecting $v$ to $w$.
Moreover, we denote the mark of a vertex $v$ in $G$ by $\tau_G(v)$. For a rooted
marked graph $(G,o)$, $\theta, \theta'\in \vermark$, and $x \in \edgemark$, let
\begin{equation*}
  E(\theta, x, \theta')(G,o) := |\{v \sim_G o: \tau_G(o) = \theta, \tau_G(v) = \theta', \xi_G(o,v) = x\}|.
\end{equation*}
For $[G,o] \in \mG_*(\edgemark, \vermark)$, we write  $E(\theta, x,
\theta')([G,o])$ for $E(\theta, x, \theta')(G,o)$ when $(G,o)$ is an arbitrary
member of the isomorphism class $[G,o]$. It is easy to verify to see that this
definition 
does not depend on the choice of the representative in the 
isomorphism class.
Given  a probability distribution $P \in
\mP(\mG_*(\edgemark, \vermark))$, for $\theta, \theta' \in \vermark$ and $x \in
\edgemark$, we define
\begin{equation*}
  e_P(\theta, x, \theta') := \evwrt{P}{E(\theta, x, \theta')([G,o])},
\end{equation*}
where the expectation is with respect to $[G,o]$ with distribution $P$.

Recall that $\mT_*(\edgemark, \vermark)$ denotes the subset of $\mG_*(\edgemark,
\vermark)$ which consists of the isomorphism classes $[G,o]$ arising from some
$(G,o)$ where the graph underlying $G$ is a tree. Let $\mT_*^1(\edgemark,
\vermark)$ be the subset of $\mT_*(\edgemark, \vermark)$ consisting of the
isomorphism classes $[T,o] \in \mT_*(\edgemark, \vermark)$ where $[T,o]$ has
depth at most one, i.e.\ all the vertices in $T$ have distance at most one from the root
node $o$. This includes an isolated root with degree zero.

\begin{definition}
  \label{def:admissible}
  A probability distribution $P \in \mP(\mT_*^1(\edgemark, \vermark))$ is called
  \emph{admissible} if $\evwrt{P}{\deg_T(o)} < \infty$ and $e_p(\theta, x,
  \theta') = e_P(\theta', x, \theta)$ for all $\theta, \theta' \in \vermark$ and
  $x \in \edgemark$.
\end{definition}

It can be shown that for a unimodular $\mu \in \mP_u(\mT_*(\edgemark,
\vermark))$ with $\deg(\mu) < \infty$, $\mu_1$ which is defined to be the law of
$[T,o]_1$ when $[T,o]$ has law $\mu$, is admissible \cite[Lemma
1]{delgosha2019notion}. This in particular highlights the importance of the
concept of admissibility. Below, given an admissible $P \in
\mP(\mT_*^1(\edgemark, \vermark))$, we define a unimodular measure in
$\mP_u(\mT_*(\edgemark, \vermark))$ which is called the marked unimodular Galton-Watson
tree with depth--1 neighborhood distribution $P$, and is denoted by $\ugwt_1(P)$\footnote{This
  discussion can be made more general to include any depth,
  see~\cite{delgosha2019notion} for more details.}. For $\theta, \theta' \in
\vermark$ and $x \in \edgemark$ such that $e_P(\theta, x, \theta') > 0$, we
define $\hP_{\theta', x, \theta} \in \mP(\mT_*^1(\edgemark, \vermark))$  via
\begin{equation}
  \label{eq:hP-def}
  \hP_{\theta', x, \theta}([T,o]) := \frac{\one{\tau_T(o) = \theta'} P([\tilde{T}, o]) E(\theta', x, \theta)(\tilde{T}, o)}{e_P(\theta', x, \theta)},
\end{equation}
where $[\tilde{T}, o] \in \mT_*^1(\edgemark, \vermark)$ is obtained from
$[T,o]$ by adding an edge to the root $o$ in $[T,o]$ which has edge mark $x$, and
the vertex mark of the endpoint of this edge other than $o$  is $\theta$. It
is straightforward to verify that $\hP_{\theta', x, \theta}$ is a probability
distribution.

With this, for an admissible $P$ as above, we define $\ugwt_1(P)$ to be the law
of $[T,o]$ when $(T,o)$ is the random rooted marked tree constructed as follows.
First,  we  sample the 1 neighborhood of the root, i.e.\ $(T,o)_1$, according to $P$.
Then, for each offspring $v \sim_T o$ of the root, we sample $[\tilde{T},
\tilde{o}]$ according to the law $\hP_{\theta', x, \theta}$ where $\theta =
\tau_T(o)$, $\theta' = \tau_T(v)$, and $x = \xi_T(o, v)$. Note that by
definition we have $\tau_{\tilde{T}}(\tilde{o})=\theta' = \tau_T(v)$. This means
that we can add $(\tilde{T}, \tilde{o})$ as a subtree below node $v$. We repeat this
process 
independently
for each $v\sim_T o$. At this point, $(T,o)$ has depth
at most $2$. Subsequently, we follow the same procedure for vertices at depth
$2, 3$, and so on inductively to construct $(T,o)$. Finally, we define
$\ugwt_1(P)$ to be the law of $[T,o]$.

For $P \in \mP(\mT_*^1(\edgemark, \vermark))$ admissible such that $d:=
\evwrt{P}{\deg_T(o)} > 0$, let $\pi_P$ denote the probability distribution on
$\vermark \times \edgemark \times \vermark$ defined as
\begin{equation}
  \label{eq:pi-P-def}
  \pi_P(\theta, x, \theta') := \frac{e_P(\theta, x, \theta')}{d}.
\end{equation}
Since 
for $[T,o] \in \mT_*^1(\edgemark, \vermark)$ 
we have $\deg_T(o) =
\sum_{\theta, x, \theta'} E(\theta, x, \theta')([T,o])$, we have $d =
\sum_{\theta, x, \theta'} e_P(\theta, x, \theta')$ and $\pi_P$ is indeed a
probability distribution.

For $P \in \mP(\mT_*^1(\edgemark, \vermark))$ admissible such that $H(P) <
\infty$ and $d := \evwrt{P}{\deg_T(o)} > 0$, define
\begin{equation}
  \label{eq:J-P-def}
  J(P):= -s(d)  + H(P) - \frac{d}{2} H(\pi_P) - \sum_{\theta, x, \theta'} \evwrt{P}{\log E(\theta, x, \theta')([T,o])!},
\end{equation}
where $s(d):=\frac{d}{2} - \frac{d}{2} \log d$. Note that since $0 < d <
\infty$, $s(d)$ is finite. On the other hand, by assumption we have $H(P) < \infty$.
Also, $H(\pi_P) \geq 0$ and $\evwrt{P}{\log E(\theta, x, \theta')([T,o])!} \geq 0$
for all $\theta, \theta' \in \vermark$ and $x \in \edgemark$. Therefore,
$J(P)$ is well defined and is in the range $[-\infty, \infty)$.

\begin{definition}
  We say that a probability distribution $P \in \mP(\mT_*^1(\edgemark,
  \vermark))$ is strongly admissible if $P$ is admissible, $H(P) < \infty$, and
  $\evwrt{P}{\deg_T(o) \log \deg_T(o)} < \infty$.
\end{definition}

The following result gives a recipe for calculating the marked BC entropy of
$\ugwt_1(P)$ when $P$ is strongly admissible. Theorem~\ref{thm:depth-1-ugwt-J} below is a direct
consequence of Theorem~3 and Proposition~5 in \cite{delgosha2019notion}.

\begin{thm}
  \label{thm:depth-1-ugwt-J}
  Let $P \in P(\mT_*^1(\edgemark, \vermark))$ be strongly admissible. Then, with
  $\mu :=\ugwt_1(P)$, we have
  \begin{equation*}
    \bch(\mu) = J(P).
  \end{equation*}
\end{thm}

Now, we apply this result to several examples, namely the local weak limits of
the sequences of \ER and the configuration model ensembles defined in
Section~\ref{sec:prel-notat}, as well as a marked $d$--regular distribution which
will be useful for our discussion in Appendix~\ref{app:constancy-counterexample}.

\subsection{Local Weak Limit of the Sequence of \ER Ensembles}
\label{sec:er-bc-ent-calculation}

Recall from Section~\ref{sec:framework-local-weak} that the local weak limit
of the sequence of \ER ensembles defined in Section~\ref{sec:prel-notat} is
$\muer_{1,2}$. Recalling the definition of $\muer_{1,2}$ from
Section~\ref{sec:framework-local-weak}, since the procedure of generating the
depth--1 neighborhood of the root
 is the same as that for each offspring, we have $\muer_{1,2} =
\ugwt_1(P_{1,2})$ where $P_{1,2} \in \mP(\mT_*^1(\edgemark_{1,2}, \vermark_{1,
  2}))$ is  defined as follows. The root is randomly assigned a mark in
$\vermark_{1,2}$ with distribution $\vec{q}$. For $x \in \edgemark_{1,2}$, we
independently generate $D_x$ with law $\text{Poisson}(p_x)$ and add $D_x$ many
edges with mark $x$ to the root. Then the vertex mark of each offspring of the root
is independently assigned with distribution $\vec{q}$. Using the thinning
property of the Poisson distribution, the number of edges connected to the root
with edge mark $x$ and vertex mark $\theta'$ at the endpoint other than the root
has law $\text{Poisson}(p_x q_{\theta'})$, independent for $x \in
\edgemark_{1,2}$ and  $\theta' \in \vermark_{1,2}$. Since the mark at the root
is $\theta$ with probability $q_{\theta}$, for $\theta, \theta' \in
\vermark_{1,2}$ and $x \in \edgemark_{1,2}$, we have
\begin{equation}
  \label{eq:er-e-txtp-distirbution}
  \pr{E(\theta, x, \theta')([T,o]) = k} = q_\theta \pr{\Lambda_{x,\theta'} = k} + (1-q_\theta) \one{k=0},
\end{equation}
where the probability on the left hand side is with respect to $P_{1,2}$,
and $\Lambda_{x,\theta'}$ denotes the number of edges connected to the root
with edge mark $x$ and vertex mark $\theta'$ at the endpoint other than the
root. From the above discussion, $\Lambda_{x,\theta'}$ is a $\text{Poisson}(p_x q_{\theta'})$ random
variable. Furthermore, $\Lambda_{x,\theta'}$ are independent for $x \in
\edgemark_{1,2}$ and $\theta' \in \vermark_{1,2}$. As a result, we have 
\begin{equation*}
  e_{P_{1,2}}(\theta, x, \theta') = q_\theta p_x q_{\theta'},
\end{equation*}
and
\begin{equation}
\label{eq:pi12-p-er}
  \pi_{P_{1,2}}(\theta, x, \theta') = \frac{e_{P_{1,2}}(\theta, x, \theta')}{\der_{1,2}} = \frac{q_\theta p_x q_{\theta'}}{\der_{1,2}},
\end{equation}
where $\der_{1,2}= \sum_{x \in \edgemark_{1,2}} p_x$ is the expected degree at
the root in $P_{1,2}$. 
Note that an object $[T,o] \in \mT_*^1(\edgemark_{1,2}, \vermark_{1,2})$ is uniquely
determined by knowing the mark at the root as well as the number of edges
connected to the root with mark $x$ and the vertex mark $\theta'$ at the
endpoint other than the root for each $x \in \edgemark_{1,2}$ and $\theta' \in
\vermark_{1,2}$. Since $\Lambda_{x,\theta'}$ are independent for $x \in
\edgemark_{1,2}$ and $\theta' \in \vermark_{1,2}$, and they are all independent
from the vertex mark at the root,  we have
\begin{equation*}
  \begin{aligned}
    H(P_{1,2}) &= H(Q) + \sum_{x \in \edgemark_{1,2}, \theta' \in \vermark_{1,2}} 
    H(\Lambda_{x, \theta'}),\\
    &= H(Q) + \sum_{x \in \edgemark_{1,2}, \theta' \in \vermark_{1,2}} p_x q_{\theta'} - (p_x q_{\theta'}) \log(p_x q_{\theta'}) + \ev{\log \Lambda_{x,\theta'}!},
\end{aligned}
\end{equation*}
where $Q $ has law $\vec{q}$ and $\Lambda_{x,\theta'}$  is a
$\text{Poisson}(p_x q_{\theta'})$ random variable as was defined above.
Further simplifying this expression using the identities $\sum_{x \in
  \edgemark_{1,2}} p_x = \der_{1,2}$ and $\sum_{\theta' \in \vermark_{1,2}}
q_{\theta'} = 1$, we get
\begin{equation}
  \label{eq:er-H-p12}
  \begin{aligned}
    H(P_{1,2}) &= H(Q) + \sum_{x \in \edgemark_{1,2}} p_x - \sum_{x \in \edgemark_{1,2}} p_x \log p_x - \left(\sum_{x \in \edgemark_{1,2}} p_x\right) \left(\sum_{\theta' \in \vermark_{1,2}} q_{\theta'} \log q_{\theta'} \right) + \sum_{x \in \edgemark_{1,2}, \theta' \in \vermark_{1,2}} \ev{\log \Lambda_{x,\theta'}!}\\
    &= H(Q) + \der_{1,2} - \sum_{x \in \edgemark_{1,2}} p_x \log p_x + \der_{1,2} H(Q) + \sum_{x \in \edgemark_{1,2}, \theta' \in \vermark_{1,2}} \ev{\log \Lambda_{x,\theta'}!}.
  \end{aligned}
\end{equation}
On the other hand, using~\eqref{eq:pi12-p-er}, we have
\begin{equation}
  \label{eq:H-pi12-er}
\begin{aligned}
  H(\pi_{P_{1,2}}) &= \sum_{\theta, \theta' \in \vermark_{1,2}, x \in \edgemark_{1,2}} \frac{q_{\theta} p_x q_{\theta'}}{\der_{1,2}} \log \frac{\der_{1,2}}{q_\theta p_x q_{\theta'}} \\
  &= \log \der_{1,2} - \sum_{x \in \edgemark_{1,2}, \theta' \in \vermark_{1,2}} \frac{p_x q_{\theta'}}{\der_{1,2}} \sum_{\theta \in \vermark_{1,2}} q_{\theta} \log q_{\theta} \\
  &\qquad \qquad - \sum_{x \in \edgemark_{1,2}, \theta \in \vermark_{1,2}} \frac{p_x q_{\theta}}{\der_{1,2}} \sum_{\theta' \in \vermark_{1,2}} q_{\theta'} \log q_{\theta'} \\
  &\qquad \qquad -\frac{1}{\der_{1,2}} \sum_{\theta, \theta' \in \vermark_{1,2}} q_\theta q_{\theta'} \sum_{x \in \edgemark_{1,2}} p_x \log p_x \\
  &= \log \der_{1,2} + 2 H(Q) - \frac{1}{\der_{1,2}} \sum_{x \in \edgemark_{1,2}} p_x \log p_x.
  \end{aligned}
\end{equation}
Furthermore, using~\eqref{eq:er-e-txtp-distirbution}, we have
\begin{equation}
  \label{eq:er-sum-log-e-txtp-factorial}
  \begin{aligned}
    \sum_{\theta, x, \theta'} \evwrt{P_{1,2}}{\log E(\theta, x, \theta')([T,o])!}  = \sum_{\theta, x, \theta'} q_\theta \ev{\log \Lambda_{x,\theta'}!} = \sum_{x, \theta'} \ev{\log \Lambda_{x,\theta'}!}.
  \end{aligned}
\end{equation}
Substituting~\eqref{eq:er-H-p12},~\eqref{eq:H-pi12-er},
and~\eqref{eq:er-sum-log-e-txtp-factorial} into \eqref{eq:J-P-def} and
simplifying, we get
\begin{align*}
J(P_{1,2}) &= -s(\der_{1,2}) + H(P_{1,2}) - \frac{\der_{1,2}}{2} H(\pi_{P_{1,2}}) - \sum_{\theta, x, \theta'} \evwrt{P_{1,2}}{\log E(\theta, x, \theta')([T,o])!} \\
                    &= - \frac{\der_{1,2}}{2} + \frac{\der_{1,2}}{2} \log \der_{1,2} + H(Q) + \der_{1,2} - \sum_{x \in \edgemark_{1,2}} p_x \log p_x + \der_{1,2} H(Q) + \sum_{x \in \edgemark_{1,2}, \theta'\in \vermark_{1,2}} \ev{\log \Lambda_{x,\theta'}!}  \\
                    & \qquad - \frac{\der_{1,2}}{2} \left( \log \der_{1,2} + 2H(Q) - \frac{1}{\der_{1,2}} \sum_{x \in \edgemark_{1,2}} p_x \log p_x  \right) \\
                    &\qquad - \sum_{x \in \edgemark_{1,2}, \theta' \in \vermark_{1,2}} \ev{\log \Lambda_{x,\theta'}!} \\
                    &= \frac{\der_{1,2}}{2} + H(Q) - \sum_{x \in \edgemark_{1,2}} \frac{p_x}{2} \log p_x \\
                    &= H(Q) + \sum_{x \in \edgemark_{1,2}} \left( \frac{p_x}{2} - \frac{p_x}{2} \log p_x \right) \\
  &= H(Q) + \sum_{x \in \edgemark_{1,2}} s(p_x).
\end{align*}
It is easy to verify that $P_{1,2}$ is strongly admissible.
Consequently, using Theorem~\ref{thm:depth-1-ugwt-J}, we get
\begin{equation}
  \label{eq:bch-muer12}
  \bch(\muer_{1,2}) = H(Q) + \sum_{x \in \edgemark_{1,2}} s(p_x).
\end{equation}
Using similar arguments for the marginals $\muer_1$ and $\muer_2$, if $Q =
(Q_1,Q_2)$ has law $\vec{q}$, we realize that
\begin{equation*}
  \bch(\muer_1) = H(Q_1) + \sum_{x \in \edgemark_1} s(p_{x_1}),
\end{equation*}
and
\begin{equation*}
  \bch(\muer_2) = H(Q_2) + \sum_{x \in \edgemark_2} s(p_{x_2}).
\end{equation*}

\subsection{Local Weak Limit of the Sequence of Configuration Model Ensembles}
\label{sec:bc-ent-conf-model}

Recall from Section~\ref{sec:framework-local-weak} that the local weak limit of
the sequence of configuration model ensembles defined in
Section~\ref{sec:prel-notat} is $\mucm_{1,2}$.
  In this section, we calculate the marked BC entropy
 of $\mucm_{1,2}$, i.e.\ the quantity $\bch(\mucm_{1,2})$, as well the marked BC
 entropy of the marginals $\mucm_1$ and $\mucm_2$. We do this by using
 Theorem~\ref{thm:depth-1-ugwt-J} discussed above. At the end of this section,
 we justify the result through an intuitive argument.
 It is easy to see that
$\mucm_{1,2} = \ugwt_1(P_{1,2})$ where $P_{1,2} \in \mP(\mT_*(\edgemark_{1,2},
\vermark_{1,2}))$ is defined as follows. The degree of the root is $X$ which has
law $\vec{r}$, the root and each of its offsprings are independently assigned a
vertex mark with law $\vec{q}$, and each edge is independently assigned an edge
mark with law $\vec{\gamma}$. With $[T,o]$ with law $P_{1,2}$, 
let $Q$ denote
the vertex mark at the root. Furthermore, for $x \in \edgemark_{1,2}$ and $\theta' \in
\vermark_{1,2}$, let  $\Lambda_{x, \theta'}$ be the number of edges connected to the
root with edge mark $x$ which have a vertex mark $\theta'$ at the endpoint other
than the root. Observe that for $\theta, \theta' \in \vermark_{1,2}$ and $x \in
\edgemark_{1,2}$, we have
\begin{equation}
  \label{eq:cm-etxtp-prob}
  \pr{E(\theta, x, \theta')([T,o]) = k} = q_\theta \pr{\Lambda_{x,\theta'} = k} + (1-q_\theta)\one{k=0},
\end{equation}
where the probability on the left hand side is with respect to $P_{1,2}$. Note
that, conditioned on $X$, $\{\Lambda_{x,\theta'}\}_{x \in \edgemark_{1,2}, \theta' \in
  \vermark_{1,2}}$ have a multinomial distribution with parameters $\{\gamma_x q_{\theta'}\}_{x \in \edgemark_{1,2}, \theta' \in
  \vermark_{1,2}}$. As a result, for $x \in \edgemark_{1,2}$ and $\theta,
\theta' \in \vermark_{1,2}$, we have
\begin{equation*}
  e_{P_{1,2}}(\theta, x, \theta') = \dcm_{1,2} q_\theta \gamma_x q_{\theta'},
\end{equation*}
and
\begin{equation}
  \label{eq:cm-pi12}
  \pi_{P_{1,2}}(\theta, x, \theta') = q_\theta \gamma_x q_{\theta'}.
\end{equation}
On the other hand, note that there is a one to one correspondence between
$[T,o]$ with law $P_{1,2}$ and the collection of random variables 
$(X, Q,
\{\Lambda_{x,\theta'}\}_{x \in \edgemark_{1,2}, \theta' \in \vermark_{1,2}})$. 
As a
result, we have
\begin{equation}
  \label{eq:cm-HP12-HX-Q-N}
\begin{aligned}
  H(P_{1,2}) &= 
H(X, Q, \{\Lambda_{x,\theta'}\}_{x \in \edgemark_{1,2}, \theta' \in \vermark_{1,2}}) = H(X) + H(Q|X) + H(\{\Lambda_{x, \theta'}\}_{x \in \edgemark_{1,2}, \theta' \in \vermark_{1,2}}|X,Q)\\
  &= H(X) + H(Q) + H(\{\Lambda_{x,\theta'}\}_{x \in \edgemark_{1,2}, \theta' \in \vermark_{1,2}}|X),
\end{aligned}
\end{equation}
where the second line uses the fact that 
$Q$ is independent from everything else. 
Recall that, conditioned on $X$, $\{\Lambda_{x, \theta'}\}_{x \in
  \edgemark_{1,2}, \theta' \in \vermark_{1,2}}$ has a multinomial distribution
with parameters $\{\gamma_x q_{\theta'}\}_{x \in \edgemark_{1,2}, \theta' \in
  \vermark_{1,2}}$. Thereby,
\begin{equation}
  \label{eq:cm-HN-cond-X-simplify}
\begin{aligned}
  H(\{\Lambda_{x,\theta'}\}_{x \in \edgemark_{1,2}, \theta' \in \vermark_{1,2}}|X) &= - \ev{\ev{\log \pr{\{\Lambda_{x,\theta'}\}_{x \in \edgemark_{1,2}, \theta' \in \vermark_{1,2}} | X}}} \\
  &= - \ev{\ev{\log \left( \binom{X}{\{\Lambda_{x,\theta'}\}_{x \in \edgemark_{1,2}, \theta' \in \vermark_{1,2}}} \right)\prod_{x \in \edgemark_{1,2}, \theta' \in \vermark_{1,2}} (\gamma_x q_{\theta'})^{\Lambda_{x, \theta'}} \Bigg |X}} \\
  &= - \ev{\log X!} + \sum_{x \in \edgemark_{1,2}, \theta' \in \vermark_{1,2}} \left( \ev{\log \Lambda_{x,\theta'}!} - \ev{\Lambda_{x,\theta'}}\log (\gamma_x q_{\theta'}) \right) \\
  &= - \ev{\log X!} + \sum_{x \in \edgemark_{1,2}, \theta' \in \vermark_{1,2}} \left( \ev{\log \Lambda_{x,\theta'}!} - \dcm_{1,2} \gamma_x q_{\theta'}\log (\gamma_x q_{\theta'}) \right) \\
  &= - \ev{\log X!} + \dcm_{1,2}(H(\Gamma) + H(Q)) + \sum_{x \in \edgemark_{1,2}, \theta' \in \vermark_{1,2}} \ev{\log \Lambda_{x, \theta'}!}.
\end{aligned}
\end{equation}
Using this in~\eqref{eq:cm-HP12-HX-Q-N}, we get
\begin{equation}
  \label{eq:cm-HP12}
  H(P_{1,2}) = H(X) + H(Q) - \ev{\log X!} + \dcm_{1,2}(H(\Gamma) + H(Q)) + \sum_{x \in \edgemark_{1,2}, \theta' \in \vermark_{1,2}} \ev{\log \Lambda_{x, \theta'}!}.
\end{equation}
Additionally, using~\eqref{eq:cm-pi12}, we have
\begin{equation}
  \label{eq:cm-H-pi12}
  H(\pi_{P_{1,2}}) = 2H(Q) + H(\Gamma).
\end{equation}
Furthermore, from~\eqref{eq:cm-etxtp-prob}, we have
\begin{equation}
  \label{eq:cm-sum-E-txtp}
  \sum_{\theta, x, \theta'} \ev{\log E(\theta, x, \theta')([T,o])!} = \sum_{\theta, x, \theta'} q_\theta \ev{\log \Lambda_{x,\theta'}!} = \sum_{x, \theta'} \ev{\log \Lambda_{x,\theta'}!}.
\end{equation}
Substituting~\eqref{eq:cm-HP12}, \eqref{eq:cm-H-pi12},
and~\eqref{eq:cm-sum-E-txtp} in~\eqref{eq:J-P-def}, we get
\begin{align*}
  J(P_{1,2}) &= -s(\dcm_{1,2}) + H(X) + H(Q) - \ev{\log X!} + \dcm_{1,2}(H(\Gamma) + H(Q)) + \sum_{x, \theta'} \ev{\log \Lambda_{x, \theta'}!} \\
             &\qquad - \frac{\dcm_{1,2}}{2} \left( 2H(Q) + H(\Gamma) \right) - \sum_{x,\theta'} \ev{\log \Lambda_{x,\theta'}!} \\
  &= -s(\dcm_{1,2}) + H(X) - \ev{\log X!} + H(Q) + \frac{\dcm_{1,2}}{2} H(\Gamma).
\end{align*}
It is easy to verify that $P_{1,2}$ is strongly admissible.
As a result, Theorem~\ref{thm:depth-1-ugwt-J} implies that
\begin{equation}
  \label{eq:bch-mucm12}
  \bch(\mucm_{1,2})= -s(\dcm_{1,2}) + H(X) - \ev{\log X!} + H(Q) + \frac{\dcm_{1,2}}{2} H(\Gamma).
\end{equation}
Observe that  if $\Gamma=(\Gamma_1, \Gamma_2)$ has law $\vec{\gamma}$, and with
$X_1$ and $X_2$ defined in~\eqref{eq:X1-X2-def}, we have $\mucm_1 =
\ugwt_1(P_1)$ where in $P_1$, the root has degree $X_1$, each vertex is
independently assigned a mark whose distribution is the same as that of $Q_1$,
and each edge has an independent edge mark whose distribution is the same as
that of $\Gamma_1$ conditioned on $\Gamma_1 \neq \circ_1$. As a result, using similar
calculations as above, we get
\begin{equation}
  \label{eq:bch-mucm-1}
  \bch(\mucm_{1})= -s(\dcm_{1}) + H(X_1) - \ev{\log X_1!} + H(Q_1) + \frac{\dcm_{1}}{2} H(\Gamma_1 | \Gamma_1 \neq \circ_1).
\end{equation}
Similarly, we have
\begin{equation}
  \label{eq:bch-mucm-2}
  \bch(\mucm_{2})= -s(\dcm_{2}) + H(X_2) - \ev{\log X_2!} + H(Q_2) + \frac{\dcm_{2}}{2} H(\Gamma_2 | \Gamma_2 \neq \circ_2).
\end{equation}

To understand the result
in~\eqref{eq:bch-mucm12}, note that the 
set of typical graphs with respect to $\mucm_{1,2}$ on the vertex set
$\{1,\dots, n\}$ is 
roughly
the set of those graphs whose degree sequence $\vec{d} = (d_1, \dots, d_n)$ has
an empirical distribution which is close to $\vec{r}$ (the degree
distribution at the root in $\mucm_{1,2}$), and 
where
the
empirical distributions of the vertex and edge marks are close to $\vec{q}$ and
$\vec{\gamma}$ respectively. The number of degree sequences $\vec{d}$ whose
empirical distribution is close to $\vec{r}$ is asymptotically close to $\exp(n
  H(X))$ where $X$ is a random variable with law $\vec{r}$. On the other hand, 
from Theorem~2.16 in \cite{bollobas1998random}, given such a typical degree sequence
$\vec{d}$, the number of unmarked graphs with  degree sequence
$\vec{d}$ is asymptotically
\begin{equation*}
  \exp(-\lambda/2 - \lambda^2 / 4) \frac{(2m)!}{m!2^m \prod_{i=1}^n d_i!} =: \mGn(\vec{d}),
\end{equation*}
where $m = (\sum_{i=1}^n d_i) /2$ and $\lambda = \frac{1}{m} \sum_{i=1}^n \binom{d_i}{2}$. Since $\vec{d}$ is a typical sequence, we
have $m = n \dcm/2 + o(n)$. Therefore, using Stirling's approximation, it is
straightforward to see that
\begin{equation*}
  \log \mGn(\vec{d}) = m \log n + n(-s(\dcm_{1,2}) - \ev{\log X!}) + o(n),
\end{equation*}
where $X$ is a random variable with law $\vec{r}$. So far, we  have justified
the role of the terms $-s(\dcm_{1,2}) + H(X) - \ev{\log X!}$ in
\eqref{eq:bch-mucm12}. The $H(Q)$ term correspond to vertex marks, and the term
$\frac{\dcm_{1,2}}{2}  H(\Gamma)$ corresponds to edge marks, since there are
$\dcm_{1,2}/2$ many edges per vertex on average in a typical graph. Similar
arguments can be used to justify~\eqref{eq:bch-mucm-1} and
\eqref{eq:bch-mucm-2}.

\subsection{Alternating Red-Blue Regular Rooted Tree}
\label{sec:BC-ent-calc-alt-red-blue}

Let the vertex mark sets for the first and the second domains be $\Theta_1 =
\{\ntheta\}$ and $\Theta_2 = \{\btheta, \rtheta\}$, respectively. Moreover, let
the edge mark sets for the first and the second domains be $\edgemark_1 =
\edgemark_2 = \{\bedge\}$. Furthermore, as in Section~\ref{sec:prel-notat}, let
\begin{align*}
  \vermark_{1,2} &= \vermark_1 \times \vermark_2 = \{(\ntheta, \btheta), (\ntheta, \rtheta)\} \\
  \edgemark_{1,2} &= ((\edgemark_1 \cup \{\circ_1\}) \times (\edgemark_2 \cup \{\circ_2\})) \setminus \{(\circ_1, \circ_2)\} \\
  &= \{(\bedge, \bedge), (\circ_1, \bedge), (\bedge, \circ_2)\}.
\end{align*}
For the sake of simplicity, we may identify $\vermark_{1,2}$ with $\{\btheta,
\rtheta\}$. Fix an integer $d \geq 3$ and let $\mu_{1,2} \in
\mP(\mT_*(\edgemark_{1,2}, \vermark_{1,2}))$  be defined as follows. Let $[T_d,
o]$ be the isomorphism class of a rooted $d$--regular unmarked trees.
Furthermore, we define $[T_d^{\btheta}, o] \in \mT_*(\edgemark_{1,2},
\vermark_{1,2})$ by adding marks to vertices and edges in $[T_d, o]$ as follows.
We give the vertex mark $\btheta$ to the root $o$, all the vertices with an odd
distance from the root receive mark $\rtheta$, and all the vertices with an even
distance from the root receive mark $\btheta$. Additionally, all the edges in
$[T_d^{\btheta}, o]$ have mark $(\bedge, \bedge)$. Similarly, we define
$[T_d^{\rtheta}, o]$ by interchanging  vertex marks $\btheta$ and $\rtheta$. With
this, we define $\mu_{1,2} \in \mP_u(\mT_*(\edgemark_{1,2}, \vermark_{1,2}))$
such that it assigns probability $1/2$ to $[T_d^{\btheta}, o]$ and
probability $1/2$ to $[T_d^{\rtheta}, o]$. 
Observe that the marginal distribution $\mu_1 \in \mP_u(\mT_*(\edgemark_1, \vermark_1))$ is
effectively a point mass on a $d$--regular tree (recall that since
$|\edgemark_1| = |\vermark_1| = 1$, the first domain is effectively unmarked).

Now we focus on calculating $\bch(\mu_{1,2})$. Let $P_{1,2} \in
\mP(\mT_*^1(\edgemark_{1,2}, \vermark_{1,2}))$ be defined as follows. $P_{1,2}$
assigns probability $1/2$ to the element in $\mT_*^1(\edgemark_{1,2},
\vermark_{1,2})$ where the root has vertex mark $\rtheta$, the root has $d$ children
each with vertex mark $\btheta$, and all the edges have edge mark
$(\bedge,\bedge)$. Moreover, $P_{1,2}$ assigns probability $1/2$ to a similar
element in $\mT_*^1(\edgemark_{1,2}, \vermark_{1,2})$ with the only difference
that the role of vertex marks $\btheta$ and $\rtheta$ are interchanged. It is
easy to verify that $P_{1,2}$ is strongly admissible, and  $\mu_{1,2} =
\ugwt_1(P_{1,2})$. Therefore, we may use Theorem~\ref{thm:depth-1-ugwt-J} to
calculate $\bch(\mu_{1,2})$. Indeed, we have
\begin{equation}
  \label{eq:HP12-log2}
  H(P_{1,2}) = \log 2.
\end{equation}
Furthermore, we have
\begin{equation}
  \label{eq:ep12}
  e_{P_{1,2}}(\btheta, (\bedge, \bedge), \rtheta) = e_{P_{1,2}}(\rtheta, (\bedge, \bedge), \btheta) = \frac{d}{2}.
\end{equation}
This implies
\begin{equation}
  \label{eq:pi-P12}
  \pi_{P_{1,2}}(\btheta, (\bedge, \bedge) = \pi_{P_{1,2}} (\rtheta, (\bedge, \bedge), \btheta) = \frac{1}{2},
\end{equation}
and
\begin{equation}
  \label{eq:H-piP12}
  H(\pi_{P_{1,2}}) = \log 2.
\end{equation}
On the other hand,
\begin{equation}
  \label{eq:EP12-sum}
\begin{aligned}
  \sum_{\theta, x, \theta'} \evwrt{P_{1,2}}{\log E(\theta, x, \theta')([T,o])!} &= \evwrt{P_{1,2}}{\log E(\btheta, (\bedge, \bedge), \rtheta)([T,o])!} \\
  &\qquad + \evwrt{P_{1,2}}{\log E(\rtheta, (\bedge, \bedge), \btheta)([T,o])!} \\
  &= \frac{1}{2} \log (d!) + \frac{1}{2} \log (d!) \\
  &= \log (d!).
\end{aligned}
\end{equation}
Substituting~\eqref{eq:HP12-log2}, \eqref{eq:H-piP12}, and \eqref{eq:ep12}
into~\eqref{eq:J-P-def}, we get
\begin{equation*}
  J(P_{1,2}) = -s(d) + \log 2 - \frac{d}{2} \log 2 - \log (d!).
\end{equation*}
It is easy to verify that $P_{1,2}$ is strongly admissible.
Thereby, simplifying and using Theorem~\ref{thm:depth-1-ugwt-J}, we get
\begin{equation}
  \label{eq:bch-mu12-rb}
  \bch(\mu_{1,2}) = -\frac{d}{2} + \frac{d}{2} \log \frac{d}{2} + \log 2 - \log (d!).
\end{equation}
It is straightforward to see that using similar calculations, we get
\begin{equation}
  \label{eq:bch-mu2-rb}
  \bch(\mu_2) = \bch(\mu_{1,2}) = -\frac{d}{2} + \frac{d}{2} \log \frac{d}{2} + \log 2 - \log (d!).
\end{equation}
Finally, to calculate $\bch(\mu_1)$, we define $P_1 \in \mP(\mT_*^1(\edgemark_1,
\vermark_1))$ to be the point mass on a root with $d$ children. It is easy to
verify that $\mu_1 = \ugwt_1(P_1)$. Also, $H(P_1) = 0$. Moreover,
$\pi_{P_1}(\ntheta, \bedge, \ntheta) = 1$ which means that $H(\pi_{P_1}) = 0$.
Additionally, we have
\begin{equation*}
  \sum_{\theta, x, \theta'} \evwrt{P_1}{\log E(\theta, x, \theta')([T,o])!} = \evwrt{P_1}{\log E(\ntheta, \bedge, \ntheta)([T,o])!)} = \log (d!).
\end{equation*}
Hence, using Theorem~\ref{thm:depth-1-ugwt-J}, we get
\begin{equation}
  \label{eq:bch-mu1-rb-}
  \bch(\mu_1) = -s(d) - \log(d!) = -\frac{d}{2} + \frac{d}{2} \log d - \log(d!).
\end{equation}

Now, we provide an intuitive explanation for the entropy formulas derived above.
We begin with $\bch(\mu_1)$ in~\eqref{eq:bch-mu1-rb-}. Roughly speaking, the set
of typical graphs with respect to $\mu_1$ on the vertex set $\{1, \dots, n\}$ is
the set of labeled unmarked $d$-regular graphs. Using Theorem~2.16 in
\cite{bollobas1998random}, the number of such graphs is asymptotically equal to
\begin{equation*}
  \exp(-(d-1)/2 - (d-1)^2 / 4) \frac{(nd)!}{(nd/2)! 2^{nd/2} (d!)^n} =: \text{Reg}_{n,d}.
\end{equation*}
Using Stirling's approximation, it is easy to verify that
\begin{equation*}
  \log \text{Reg}_{n,d} = \frac{nd}{2} \log n + n (-s(d) - \log(d!)) + o(n).
\end{equation*}
Note that $nd/2$ is the number of  edges in a $d$--regular graph, and the coefficient of $n$ in this expression is equal to $\bch(\mu_1)$
as was demonstrated in~\eqref{eq:bch-mu1-rb-}.
Note that since all the edges in $\mu_{1,2}$ also appear in $\mu_2$, and the
vertex and edge marks in $\mu_{1,2}$ can be recovered from those in $\mu_2$, we
 have $\bch(\mu_{1,2}) = \bch(\mu_2)$ (as was stated in~\eqref{eq:bch-mu2-rb}
above). Observe that roughly speaking, due to the alternating red--blue vertex
marks in $\mu_2$, and the fact that the root mark is red with probability $1/2$
and blue with probability $1/2$, the set of  $\mu_2$ typical graphs is more or less the set of
$d$--regular graphs which have a red--blue vertex marking such that almost half
of the vertices are red and the rest half are blue, most of the red vertices
have all of their $d$ neighbors marked as blue, and most of the blue vertices
have all of their $d$ neighbors marked as red. In other words, the set of
$\mu_2$ typical graphs is more or less the set of bipartite $d$--regular graphs
where one partition has $n/2$ red vertices, and the other partition has $n/2$ blue
vertices.
Given such a marked graph, we construct an unmarked bipartite graph by
relabeling the vertices with mark red to $\{1, \dots, n/2\}$, preserving their
order, and relabeling the vertices with mark blue to $\{n/2 + 1, \dots, n\}$,
also preserving their order. 
From \cite{bekessy1972asymptotic}, the
number of $d$--regular bipartite graphs on the vertex set $\{1, \dots, n\}$ with
vertices $\{1, \dots, n/2\}$ in one partition and vertices $\{n/2+1, \dots, n\}$
in the second partition is asymptotically equal to
\begin{equation*}
  \exp(-(d-1)^2/2) \frac{(nd/2)!}{(d!)^n}  =: \text{Bip}_{n,d}.
\end{equation*}
Because of the above relabeling of vertices,  given  such an unmarked bipartite graph, there are $\binom{n}{n/2}$ marked
bipartite graphs as above. As a result, using Stirling's approximation, the
logarithm of the number of $\mu_2$ typical graphs
is asymptotically 
\begin{equation*}
  \log \text{Bip}_{n,d} + \log \binom{n}{n/2} = \frac{dn}{2} \log n + n \left( \frac{d}{2} \log \frac{d}{2} - \frac{d}{2} - \log (d!) + \log 2 \right) + o(n).
\end{equation*}
Note that $nd/2$ is the number of edges in a $d$--regular graph, and the
coefficient of $n$ in the above expression is precisely $\bch(\mu_{2}) =
\bch(\mu_{1,2})$ as in~\eqref{eq:bch-mu2-rb}.

\editfinish

\section{Counterexample for the Constancy of the Size of the Set of Conditional
  Typical Graphs}
\label{app:constancy-counterexample}

\editstart 

In this section, we study the asymptotic  size of joint,
marginal, and conditional typical graphs for the example of
Appendix~\ref{sec:BC-ent-calc-alt-red-blue} and we observe a behavior which is
fundamentally different from what we expect from classical information theory,
namely the constancy of the size of conditional typical sequences in classical
information theory. This kind of behavior 
in part makes
  our analysis more complicated compared to the classical setting as we need to
  carefully control the number of jointly typical graphs.

Fix an integer $d \geq 3$ and let $\mu_{1,2} \in \mP_u(\mT_*(\edgemark_{1,2},
\vermark_{1,2}))$ be the alternating red--blue $d$--regular random rooted tree explained in
Section~\ref{sec:BC-ent-calc-alt-red-blue} of
Appendix~\ref{app:bc-ent-calc-examples}. In order to study the joint, marginal,
and conditional typical graphs for this example, fix  sequences $\vmn  = \{\mn(x)\}_{x \in
  \edgemark_{1,2}}$ and $\vun = \{\un(\theta)\}_{\theta \in \vermark_{1,2}}$ for $n
\geq 1$ of edge mark and vertex mark count vectors adapted to $(\vdeg(\mu_{1,2}),
\vvtype(\mu_{1,2}))$ in the sense of Definition~\ref{def:deg-seq-adapt}. Note
that since $\deg_{(\bedge, \circ_2)}(\mu_{1,2}) = \deg_{(\circ_1,
  \bedge)}(\mu_{1,2}) = 0$, condition~\ref{item:cond-adapt-dx0} in
Definition~\ref{def:deg-seq-adapt} implies that
\begin{equation}
  \label{eq:rb-mn-circ-12-zero}
  \mn(\bedge, \circ_2) = \mn(\circ_1, \bedge) = 0 \qquad \forall n.
\end{equation}
Motivated by this, the only nonzero element in $\vmn$ is $\mn(\bedge, \bedge)$
which is the total number of edges. Therefore,  we
define $\mn := \mn(\bedge, \bedge)$, and to simplify the notation we 
write $\mn$ 
instead of $\vmn$.
Note
that~\eqref{eq:rb-mn-circ-12-zero} in particular implies that for every marked
graph in the joint domain $\mGn_{\mn, \vun}$, all the edges appear in both
marginals. Following the convention in~\eqref{eq:utheta1}, we define the
marginal vertex mark count vectors $\vun_1 = (\un_1(\theta_1))_{\theta_1 \in
  \vermark_1}$ and $\vun_2 = (\un_2(\theta_2))_{\theta_2 \in \vermark_2}$. Note
that since $\vermark_1 = \{\ntheta\}$, graphs on the first domain are
effectively unmarked. Hence, we may simply write $\mGn_{\mn}$ and
$\mGn_{\mn}(\mu_1, \epsilon)$ instead of $\mGn_{\mn, \vun_1}$ and $\mGn_{\mn,
  \vun_1}(\mu_1, \epsilon)$, respectively.

Given $0 < \epsilon < \epsilon'$ and $G_1 \in \mGn_{\mn}(\mu, \epsilon)$, we
define the conditional typical set as 
\begin{equation}
  \label{eq:mGn-2-cond-1-rb}
  \mGn_{\mn, \vun_2}(\mu_2, \epsilon'|G_1) := \{G_2 \in \mGn_{\mn, \vun_2}: G_1 \oplus G_2 \in \mGn_{\mn, \vun}(\mu_{1,2}, \epsilon')\}. 
\end{equation}
In words, this is the set of graphs on the second domain which are jointly
typical with $G_1$. Note that, due to~\eqref{eq:rb-mn-circ-12-zero}, each graph
in $\mGn_{\mn, \vun_2}(\mu_2, \epsilon'|G_1)$ has the same set of edges as in
$G_1$, and only has vertex marks added to $G_1$.
Extrapolating the results from classical information theory, we might expect
that for each $G_1 \in \mGn_{\mn}(\mu_1, \epsilon)$, the set $\mGn_{\mn,
  \vun_2}(\mu_2, \epsilon'|G_1)$ has roughly the same size, and this size is
dependent on the conditional marked BC entropy $\bch(\mu_2|\mu_1)$. However, as
we will see below, this is not true.

\begin{prop}
  \label{prop:rb-example-conditional-typical-size}
  For the above example, there exists $\epsilon_0 > 0$ such that for all $0 <
  \epsilon < \epsilon' < \epsilon_0$, for $n$ large enough, the set
  \begin{equation*}
    A_{n, \epsilon, \epsilon'} := \{G_1 \in \mGn_{\mn}(\mu_1, \epsilon): \mGn_{\mn, \vun_2}(\mu_2, \epsilon' | G_1) \text{ is empty}\},
  \end{equation*}
  is not empty.
\end{prop}

Note that if we take a graph $G_{1,2}$ which is $\mu_{1,2}$--typical, i.e.\ if
$G_1 \in \mGn_{\mn, \vun}(\mu_{1,2}, \epsilon)$, then it is easy to verify that
the marginal graph $G_1$ is $\mu_1$--typical, i.e.\ $G_1 \in \mGn_{\mn}(\mu_1,
\epsilon)$, and also by definition $\mGn_{\mn, \vun_2}(\mu_2, \epsilon'|G_1)$ is
not empty. This behavior is fundamentally different from what we know from
classical information theory, where roughly speaking, all marginal typical
sequences have a nonempty set of conditional typical sequences with an asymptotically
constant size.

Before proving this result, we intuitively discuss why it holds. Roughly
speaking, the set of  $\mu_1$--typical graphs is less or more the set of
almost $d$--regular graphs. On the other hand, $\mu_{1,2}$--regular graphs in
addition to being almost $d$--regular, should also have a vertex marking which
results in an almost bipartite partitioning. Therefore, only those
$\mu_1$--regular graphs which also have at least one such almost bipartite
marking  can have conditional typical graphs on the second domain. But since not all
$d$--regular graphs have such a bipartite partitioning, there are $\mu_1$
typical graphs for which their corresponding conditional typical set is empty.

\begin{proof}[Proof of Proposition~\ref{prop:rb-example-conditional-typical-size}]
  Observe that for two distinct $G_1$ and $G'_1$ in $\mGn_{\mn}(\mu_1,
  \epsilon)$, the sets $\mGn_{\mn, \vun_2}(\mu_2, \epsilon'|G_1)$ and
  $\mGn_{\mn, \vun_2}(\mu_2, \epsilon'|G'_1)$ are distinct. To see this, assume
  that $G_2 \in \mGn_{\mn, \vun_2}(\mu_2, \epsilon'|G_1) \cap \mGn_{\mn,
    \vun_2}(\mu_2, \epsilon'|G'_1)$ and note that due
  to~\eqref{eq:rb-mn-circ-12-zero}, the set of edges in $G_2$ is identical to
  those of $G_1$ and $G'_1$. But each edge and vertex in $G_1$ and $G'_1$ can
  have only one possible mark. This implies that $G_1 = G'_1$ which is a contradiction.
  This implies that
  \begin{equation}
    \label{eq:sum-mu2-cond-G1-less-mu2}
    \sum_{G_1 \in \mGn_{\mn}(\mu_1, \epsilon)} |\mGn_{\mn, \vun_2}(\mu_2, \epsilon'|G_1)| \leq |\mGn_{\mn, \vun_2}(\mu_2, \epsilon')|.
  \end{equation}
  On the other hand,
  \begin{align*}
    \sum_{G_1 \in \mGn_{\mn}(\mu_1, \epsilon)} |\mGn_{\mn, \vun_2}(\mu_2, \epsilon'|G_1)| &\geq |\mGn_{\mn}(\mu_1, \epsilon) \setminus A_{n, \epsilon, \epsilon'}| \\
    &= |\mGn_{\mn}(\mu_1, \epsilon)| - |A_{n, \epsilon, \epsilon'}|.
  \end{align*}
  Comparing this with~\eqref{eq:sum-mu2-cond-G1-less-mu2}, we get
  \begin{equation}
    \label{eq:mgn-mu1-anee-mgn-mu2}
    |\mGn_{\mn}(\mu_1, \epsilon)| - |A_{n, \epsilon, \epsilon'}| \leq |\mGn_{\mn, \vun_2}(\mu_2, \epsilon')|.
  \end{equation}
  Using the definition of the marked BC entropy (Definition~\ref{def:BC-entropy}
  in Section~\ref{sec:bc-entropy}) and Theorem~\ref{thm:bch-properties}, we have
  \begin{equation*}
    \lim_{\epsilon'\downarrow 0} \limsup_{n \rightarrow \infty} \frac{\log |\mGn_{\mn, \vun_2}(\mu_2, \epsilon')| - \mn \log n}{n} = \bchover(\mu_2) = \bch(\mu_2),
  \end{equation*}
  and
  \begin{equation*}
    \lim_{\epsilon\downarrow 0} \liminf_{n \rightarrow \infty} \frac{\log |\mGn_{\mn}(\mu_1, \epsilon)| - \mn \log n}{n} = \bchunder(\mu_1) = \bch(\mu_1),
  \end{equation*}
But from the calculations in Section~\ref{sec:BC-ent-calc-alt-red-blue} in
Appendix~\ref{app:bc-ent-calc-examples}, we have
\begin{equation*}
  \bch(\mu_1) = -\frac{d}{2} + \frac{d}{2} \log d - \log (d!) > -\frac{d}{2} + \frac{d}{2} \log \frac{d}{2} + \log 2 - \log(d!) = \bch(\mu_2),
\end{equation*}
where the inequality holds since $d \geq 3$ by assumption. This means that there
exists $\epsilon_0 > 0$ such that for all $0 < \epsilon < \epsilon' <
\epsilon_0$, when $n$ is large, we have
\begin{equation*}
  |\mGn_{\mn}(\mu_1, \epsilon)| > |\mGn_{\mn, \vun_2}(\mu_2, \epsilon')|.
\end{equation*}
Comparing this with~\eqref{eq:mgn-mu1-anee-mgn-mu2}, we realize that for this
$\epsilon_0$, for all $0 < \epsilon < \epsilon' < \epsilon_0$, when $n$ is
large, $A_{n, \epsilon, \epsilon'}$ is not empty. This is precisely what we
wanted to prove.
\end{proof}

\editfinish

}
\end{document}